\documentclass[aps,prl,reprint,longbibliography,superscriptaddress,dvipsnames]{revtex4-2}
\usepackage[colorlinks=true,allcolors=blue]{hyperref}
\usepackage{amsmath,amssymb,graphicx,enumitem,braket,bm,bbm,times,xcolor,xfp,amsthm,cancel}
\newcommand{\figref}[2]{\hyperref[#1]{\ref*{#1}\textcolor{blue}{(#2)}}}

\begin{document}

\title{Non-Hermitian Fermi-Dirac Distribution in Persistent Current Transport}
\author{Pei-Xin Shen}
\thanks{\hypertarget{Email:Shen}{\href{mailto:peixin.shen@outlook.com}{peixin.shen@outlook.com}}}
\affiliation{International Research Centre MagTop, Institute of Physics, Polish Academy of Sciences, Aleja Lotnikow 32/46, PL-02668 Warsaw, Poland}
\affiliation{Institute for Interdisciplinary Information Sciences, Tsinghua University, Beijing 100084, China}
\author{Zhide Lu}
\affiliation{Institute for Interdisciplinary Information Sciences, Tsinghua University, Beijing 100084, China}
\author{Jose L. Lado}
\affiliation{Department of Applied Physics, Aalto University, FI-00076 Aalto, Espoo, Finland}
\author{Mircea Trif}
\thanks{\hypertarget{Email:Trif}{\href{mailto:mtrif@magtop.ifpan.edu.pl}{mtrif@magtop.ifpan.edu.pl}}}
\affiliation{International Research Centre MagTop, Institute of Physics, Polish Academy of Sciences, Aleja Lotnikow 32/46, PL-02668 Warsaw, Poland}
\date{\today}

\begin{abstract}
Persistent currents circulate continuously without requiring external power sources. 
Here, we extend their theory to include dissipation within the framework of non-Hermitian quantum Hamiltonians. Using Green's function formalism, we introduce
a non-Hermitian Fermi-Dirac distribution
and derive an analytical expression for the persistent current that relies solely on the complex spectrum. 
We apply our formula to two dissipative models supporting persistent currents: (i) a phase-biased superconducting-normal-superconducting junction; 
(ii) a normal ring threaded by a magnetic flux.
We show that the persistent currents in both systems
exhibit no anomalies at any emergent exceptional points, whose signatures are only discernible in the current susceptibility. We validate our findings by exact diagonalization and extend them to account for finite temperatures and interaction effects. 
Our formalism offers a general framework for computing quantum many-body observables of non-Hermitian systems in equilibrium, 
with potential extensions to non-equilibrium scenarios.

\end{abstract}

\maketitle

\emph{Introduction.---}Recent intensive research in non-Hermitian (NH) physics \cite{Ashida2020NonHermitian,Bergholtz2021Exceptional,Ding2022NonHermitian,Okuma2023NonHermitian} has revealed intriguing phenomena in both the classical \cite{Xue2021Simple,Xue2022NonHermitian,Hu2023Steadystate,Li2024NonBloch,Zou2024Dissipative} and quantum realms \cite{Yu2022Experimental,Chen2023Topological,Ochkan2024NonHermitian}. The biorthogonal and non-Bloch frameworks have reshaped the conventional bulk-edge correspondence \cite{Kunst2018Biorthogonal,Yao2018Edge,Yao2018NonHermitian,Yang2020NonHermitian}, while the symmetry classifications of the NH matrices have enriched the topological phases compared to their Hermitian counterparts \cite{Bernard2002Classification,Kawabata2019Symmetry,Altland2021Symmetry,Yu2021Unsupervised,Denner2021Exceptional}. 
Exceptional points (EPs), where the NH Hamiltonian is not diagonalizable \cite{Kato1995Perturbation,Heiss2012Physics,Golub2013Matrix}, 
can enhance sensing capabilities 
\cite{Wiersig2014Enhancing,Chen2017Exceptional,Hodaei2017Enhanced,Lau2018Fundamental} 
and trigger new critical phenomena 
\cite{Lee2014Entanglement,San-Jose2016Majorana,Ashida2017Paritytimesymmetric,Tserkovnyak2020Exceptional}.

\begin{figure}
\includegraphics[width=\linewidth]{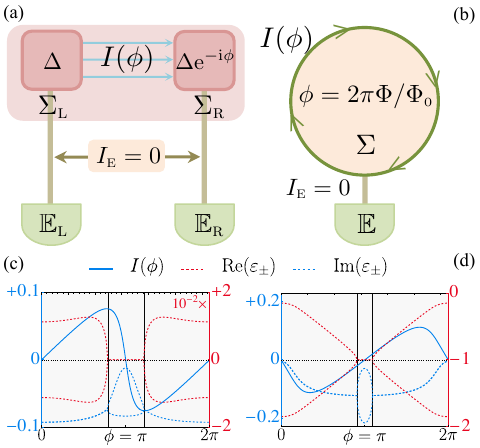}
    \caption{Schematic of systems coupled to external reservoirs $\mathbb{E}$: 
    (a) an SNS junction with a phase bias $\phi$; 
    (b) a normal metallic ring threaded by normalized magnetic flux $\phi$. Each system is characterized by an effective NH Hamiltonian $\mathcal{H}_{\rm eff}$ that includes a complex self-energy $\Sigma$ from $\mathbb{E}$. In equilibrium, both models maintain persistent currents $I(\phi)$ and zero leakage currents $I_\textsc{e}$. (c) and (d) Complex spectra $\varepsilon_\pm$ showing EPs around $\pi$ (black lines) and persistent current $I(\phi)$ as a function of $\phi$. The current $I(\phi)$ calculated using Eq.~\eqref{Eqn:CurrentImTrLog} shows no signs of singularity at the EPs. Parameters: (c) $N_\textsc{l}=N_\textsc{m}=N_\textsc{r}=4$, $N_\textsc{e} = 101$, $t = -\Delta = -1$, $\kappa=0.4 t$, $\mu=g = -1.1$; (d) $N=6$, $N_\textsc{e} = 101$, $t = \kappa = \mu =-1$, $g = 0$, $t_i \sim t*\mathrm{Unif}(0.7, 1.3)$.}
    \label{Fig:SketchSpectrum}
\end{figure}

In the field of open quantum systems, NH physics is instrumental in characterizing the dissipative nature of systems \cite{Kohler2005Driven,Datta2005Quantum,Landi2022Nonequilibrium}. 
The Lindblad formalism \cite{Gorini1976Completely,Lindblad1976Generators,Prosen2008Third,McDonald2023Third} 
provides routes to address system-reservoir interactions. By neglecting quantum jumps or focusing on Gaussian systems \cite{Minganti2019Quantum,McDonald2022Nonequilibrium,Song2019NonHermitian}, the dynamics is dictated solely by an effective NH Hamiltonian $\mathcal{H}_{\rm eff}$.
The Green's function formalism presents an alternative path to $\mathcal{H}_{\rm eff}$ by integrating out external reservoirs $\mathbb{E}$ to include complex self-energies $\Sigma \neq \Sigma^\dagger$ \cite{Rickayzen1980Green,Economou2006Green,Wang2012Green,Odashima2016Pedagogical,Shen2018Quantum,Kozii2024NonHermitian}. Although the spectral properties of $\mathcal{H}_{\rm eff}$ have been extensively explored, quantum many-body observables \cite{Meden2023Mathcal}, such as the supercurrents in phase-biased superconducting-normal-superconducting (SNS) junctions shown in Fig.~\figref{Fig:SketchSpectrum}{a}, are currently under active discussion.
Existing approaches, such as the derivative of complex eigenvalues \cite{Cayao2023NonHermitian,Li2024Anomalous} and the expectation values obtained from the left-right (LR) or right-right (RR) eigenvectors \cite{Kornich2023CurrentVoltage}, often yield anomalies at EPs \footnote{We adopt the LR/RR-basis definition in Ref.~\cite{Kornich2023CurrentVoltage}, where the electron density is anomalous at EPs, and the PT-symmetric Hamiltonian is distinct from $\mathcal{H}_{\rm eff}$ studied in this Letter within the context of open quantum systems \cite{Meden2023Mathcal}.}, calling for a microscopic approach to grasp the subtleties of the NH persistent current transport.

In this Letter, we provide a resolution to this conundrum grounded on a NH Fermi-Dirac distribution
associated with the biorthogonal single-particle eigenstates.
We find that the supercurrent $I(\phi)$ in an SNS junction biased by a phase $\phi$ and coupled to reservoirs is given by (in units of $e/\hbar$):
\begin{equation}
    I(\phi) = - \frac{1}{\pi} \frac{\mathrm{d}}{\mathrm{d}\phi} \mathrm{Im} \mathrm{Tr} (\mathcal{H}_{\rm eff} \ln \mathcal{H}_{\rm eff}) \,.
    \label{Eqn:CurrentImTrLog}
\end{equation}
This formula is derived in the wide-band limit, which is non-perturbative and can also accurately describe the strong coupling regime. Moreover, it also applies to the persistent current in a normal mesoscopic ring threaded by a magnetic flux, as shown in Fig.~\figref{Fig:SketchSpectrum}{b}. Our results align with the exact diagonalization of the full Hermitian system including the reservoir and do not exhibit any singularities at EPs for both models [see Figs.~\figref{Fig:SketchSpectrum}{c} and \figref{Fig:SketchSpectrum}{d}]. We further generalize Eq.~\eqref{Eqn:CurrentImTrLog} to finite temperatures and find that persistent currents are reduced, which is also observed when many-body interactions are taken into account. Finally, as shown in Fig.~\ref{Fig:Susceptibility}, the signatures of EPs can instead become evident in the current susceptibility associated with response to an $ac$ phase bias drive. Our formalism not only clarifies the behavior of persistent currents in fermionic NH systems but also sets the stage for analyzing other quantum many-body observables with dissipation.

\emph{Phenomenology and methodology.---}Initially, we outline a heuristic explanation of our main findings, deferring the technical details to subsequent sections and the Supplemental Material (SM) \cite{Note0}.
\nocite{Shen2023Majoranamagnon}
We focus, for simplicity, on spinless fermionic systems with Hamiltonians that depend on a parameter $\phi$, such as in phase-biased SNS junctions $H_{\rm sys}(\phi) = \vec{C}^\dagger \mathcal{H}_{\rm sys} (\phi) \vec{C}/\mathbf{2}$, where $\mathcal{H}_{\rm sys} (\phi)$ is the Bogoliubov-de Gennes (BdG) Hamiltonian, $\vec{C} = (c_1, \dots, c_\textsc{n}, c^\dagger_1, \dots, c^\dagger_\textsc{n})^\textsc{t}$ is a $\mathbf{2}N$-dimensional spinor, and $c_j^\dagger$ ($c_j$) is the fermionic creation (annihilation) operator at site $j$. The discussion below also applies to normal metals described by $\mathcal{H}_{\rm sys} (\phi)$ on an $N$-dimensional basis $\vec{C} = (c_1, \dots, c_\textsc{n})^\textsc{t}$, substituting $\mathbf{2}$ for $\mathbf{1}$. For brevity, we will use calligraphy to denote the first quantized operators and omit the explicit dependence on $\phi$ in our notation for Hamiltonians, eigenvalues, and eigenvectors hereafter.

For isolated and Hermitian systems $\mathcal{H}_{\rm sys}$, 
the persistent current $I_{\rm iso}(\phi)$ in the many-body ground state $E_0$ follows \cite{Beenakker1992Superconducting},
\begin{equation}
\label{Eqn:CurrentHermitian}
I_{\rm iso}(\phi) 
    = \mathbf{2}\frac{\mathrm{d} E_0}{\mathrm{d} \phi}  
    = \sum_{\epsilon_n \leqslant 0} \frac{\mathrm{d} \epsilon_n}{\mathrm{d} \phi} 
    = \sum_{\epsilon_n \leqslant 0} \frac{\braket{\psi_n |\mathcal{J}| \psi_n}}{\mathbf{2}} \, ,
\end{equation}
where the factor of two stems from the Cooper pair, 
$J = \vec{C}^\dagger \mathcal{J} \vec{C}/\mathbf{2}$ is the persistent current operator,
$\epsilon_n$ and $\ket{\psi_n}$ are eigenvalues and eigenstates of $\mathcal{H}_{\rm sys}$. Given the local conservation law of $n_j = c^\dagger_j c_j$ in the $\mathbb{N}$ segment (e.g., the normal part in SNS junctions), the site-resolved current operator \cite{Landi2022Nonequilibrium}, 
$J_j = -\mathrm{i}t_j (c^\dagger_j c_{j+1} - c^\dagger_{j+1} c_j )$, 
follows the continuity equation in equilibrium: $0 = \braket{\dot{n}_j} = \mathrm{i} \braket{[H_{\rm sys}, n_j]} = \braket{J_{j}} - \braket{J_{j-1}}, \forall j \in \mathbb{N}$, where $t_j$ is the hopping strength at site $j$. Therefore,
we set $J \equiv J_j$ at the first site of $\mathbb{N}$ and omit the subscript.

To account for dissipation, we couple the system to a thermal reservoir $\mathbb{E}$. In the long-time limit, the system reaches equilibrium and $\mathbb{E}$ acts as a source of dephasing \cite{Buttiker1985Small,Stern1990Phase,Loss1991Dephasing,Chang1997Control,Mortensen2000Dephasing,Gramespacher2000Distribution,Belogolovskii2001Charge,Belogolovskii2003Phasebreaking,Beri2009Dephasingenabled}. This coupling leads to the emergence of a complex self-energy $\Sigma(\omega)$ in the system.
In the wide-band limit $\Sigma(\omega) \approx \Sigma(0)$ \cite{Jauho1994Timedependent,Potts2021Thermodynamically,Yu2024NonHermitian}, the system is effectively described by the NH Hamiltonian $\mathcal{H}_{\rm eff} \equiv \mathcal{H}_{\rm sys} + \Sigma(0)$,
which exhibits a complex spectrum $\varepsilon_{n}$ with $\mathrm{Im}\, \varepsilon_n \leqslant 0$ 
and supports biorthogonal single-particle modes \cite{Brody2013Biorthogonal}: $\mathcal{H}_{\rm eff} \ket{\psi^\textsc{r}_n} = \varepsilon_n \ket{\psi^\textsc{r}_n}$, $\mathcal{H}_{\rm eff}^\dagger \ket{\psi^\textsc{l}_n} = \varepsilon^*_n \ket{\psi^\textsc{l}_n}$, and $\braket{\psi^\textsc{l}_n|\psi^\textsc{r}_m} = \delta_{nm}$. Using this biorthogonal basis, we represent the retarded Green's function of the system as \cite{Chen2018Hall,Yang2021Dissipative,Arouca2023Topological}
\begin{align}
\label{Eqn:BiorthogonalGF}
G_{\rm sys} (\omega) 
    = \frac{1}{\omega - \mathcal{H}_{\rm eff}}
    = \sum_n\frac{\ket{\psi^\textsc{r}_n}\bra{\psi^\textsc{l}_n}}{\omega - \varepsilon_{n}} \,,
\end{align}
and obtain the density of states operator $\rho (\omega) = \mathrm{i}[G_{\rm sys} (\omega) - G^\dagger_{\rm sys} (\omega)]/2\pi$. 
In thermal equilibrium, any correlator can be calculated by $\braket{c^\dagger_i c_j} = \int \braket{j|\rho(\omega)|i} f_\textsc{fd}(\omega) \mathrm{d} \omega$, where $f_\textsc{fd}(\omega)$ is the Fermi-Dirac distribution of the entire system. Given $f_\textsc{fd}(\omega) = \Theta(-\omega)$ at zero temperature, we derive an analytical correlator $\braket{c^\dagger_i c_j}$ by integrating over $\omega$:
\begin{align}
\label{Eqn:CorrelatorGF}
\braket{c^\dagger_i c_j} 
    &= \frac{\mathrm{i}}{2\pi} \sum_n 
    \left( \psi^\textsc{l*}_{ni}\psi^\textsc{r}_{nj} \ln \varepsilon_n - \psi^\textsc{r*}_{ni}\psi^\textsc{l}_{nj} \ln \varepsilon^*_n  \right) \,,
\end{align}
where $\psi^\textsc{l/r}_{nj} \equiv \braket{j|\psi^\textsc{l/r}_{n}}$, $\ln \varepsilon_n \equiv \ln |\varepsilon_n| + \mathrm{i} \arg \varepsilon_n$ and $-\pi \leqslant \arg \varepsilon_n \leqslant 0$ \footnote{The $\ln |\varepsilon_n|$ term stems from the principal value (PV) in the integrand, which is typically disregarded in the literature \cite{Rickayzen1980Green}, since $\ket{\psi^\textsc{l}_n} \approx \ket{\psi^\textsc{r}_n}$ in the weak coupling limit. However, PV plays a crucial role in correctly determining observables in EPs, which typically occur in the strong coupling regime.}. Similarly, $\braket{c_i c_j}$ is obtained by replacing $i$ with $N+i$ on the right-hand side of Eq.~\eqref{Eqn:CorrelatorGF}. Therefore, the expectation value of a general quadratic Hermitian operator $O = \vec{C}^\dagger \mathcal{O} \vec{C}/\mathbf{2}$ is \cite{Note0}
\begin{align}
\label{Eqn:HermitianO}
\braket{O} 
    = \mathrm{Im} \sum_n 
    \frac{\braket{\mathcal{O}}^\textsc{lr}_n f_{\rm eff}(\varepsilon_n)}{\mathbf{2}}
    = 
    \frac{\mathrm{Im} \mathrm{Tr} [ \mathcal{O} f_{\rm eff}(\mathcal{H}_{\rm eff})]}{\mathbf{2}} 
    \,, 
\end{align}
where $\braket{\mathcal{O}}^\textsc{lr}_n \equiv \braket{\psi^\textsc{l}_n |\mathcal{O}| \psi^\textsc{r}_n}$ and $f_{\rm eff} (\varepsilon) \equiv - (1/\pi)\ln \varepsilon$ acts as a Fermi-Dirac distribution for NH systems, whose imaginary part reduces to $\Theta(-\varepsilon)$ as $\mathrm{Im}\, \varepsilon \rightarrow 0$ \footnote{Observables are gauge-invariant under $f_{\rm eff} \rightarrow f_{\rm eff} + \mathbb{R}$: All observables stay the same with $f_{\rm eff} (\varepsilon) = - (1/\pi)\ln (\varepsilon/ C)$, $\forall C \in \mathbb{R}$ \cite{Note0}. Hence, we set $C = 1$ for simplicity.}. Eq.~\eqref{Eqn:HermitianO} represents one of our main results and remains continuous at EPs \footnote{Owing to the analytic property of $f_{\rm eff}$, the EP is a removable singularity for physical observables, albeit a branch point for biorthogonal wavefunctions \cite{Note0}}.
The persistent current 
can be calculated by substituting $\mathcal{O}$ with $\mathcal{J}$ in Eq.~\eqref{Eqn:HermitianO}.
Furthermore, applying the identity $\braket{\mathcal{J}}^\textsc{lr}_n = \mathbf{2} \partial_{\phi} \varepsilon_n$ for each biorthogonal single-particle mode
and rearranging the derivatives, one can obtain Eq.~\eqref{Eqn:CurrentImTrLog} and verify that it recovers Eq.~\eqref{Eqn:CurrentHermitian} in the Hermitian limit.
It is important to emphasize that our Eq.~\eqref{Eqn:CurrentImTrLog} is distinct from a simple continuation of Eq.~\eqref{Eqn:CurrentHermitian} to complex eigenvalues $\sum_{\mathrm{Re}\, \varepsilon_n \leqslant 0} \partial_{\phi} \varepsilon_n$ (or, equivalently, the LR-basis current $I_\textsc{lr}(\phi) \equiv \sum_{\mathrm{Re}\, \varepsilon_n \leqslant 0} \braket{\mathcal{J}}^\textsc{lr}_n/ \mathbf{2}$) recently proposed in Refs.~\cite{Cayao2023NonHermitian,Li2024Anomalous,Kornich2023CurrentVoltage}, as well as the RR-basis current $I_\textsc{rr}(\phi) \equiv \sum_{\mathrm{Re}\, \varepsilon_n \leqslant 0} \braket{\mathcal{J}}^\textsc{rr}_n/ \mathbf{2}$ widely adopted with post-selection \cite{Kawabata2023Entanglement,Herviou2019Entanglement,Naghiloo2019Quantum}. As demonstrated below, both of these definitions fail to accurately describe the persistent current in equilibrium, whereas Eq.~\eqref{Eqn:CurrentImTrLog} is in full agreement with the exact diagonalization. 

\emph{Model reservoir and self-energy.---}To validate our findings, we connect the system $H_{\rm sys}$ to a $N_\textsc{e}$-site fermionic reservoir $H_{\rm res} = \sum\nolimits_{j} \big[(tc^\dagger_j c_{j+1}+{\rm H. c. }) + g c^\dagger_j c_j \big]$,
where $t<0$ is the hopping strength and $g$ is the chemical potential. This specific reservoir is chosen for its dual analytical and numerical merits. First, connecting one end of $\mathbb{E}$ to the $l$-site of the system via $H_{\rm tun} = \kappa (c^\dagger_{N_\textsc{e}} c_{l} + c^\dagger_{l} c_{N_\textsc{e}})$  with coupling strength $\kappa<0$
will induce a self-energy $\Sigma_l(0) = \Sigma(0) \otimes \ket{l}\bra{l}$ onto $\mathcal{H}_{\rm sys}$, where $\Sigma(0) = - \kappa^2/t^2 (\tau_z g /2 + \mathrm{i} \sqrt{t^2 - (g/2)^2} )$
and $\tau_z$ is the Pauli-$Z$ matrix acting in the particle-hole space \cite{Datta2005Quantum}. This expression is exact when $N_\textsc{e} \rightarrow \infty$ and also applies to normal metals upon removal of $\tau_z$ \cite{Note0}. 
Second, the tight-binding form of $H_{\rm res}$ allows us to compare Eqs.~\eqref{Eqn:CurrentImTrLog} and \eqref{Eqn:CurrentHermitian} by performing an exact diagonalization of the entire Hermitian system $H_{\rm tot} = H_{\rm sys} + H_{\rm res} + H_{\rm tun}$. Next, we apply this benchmark framework to two concrete NH models: a phase-biased SNS junction and a normal ring threaded by a magnetic flux.

\emph{NH SNS junctions.---}The SNS junction is a pivotal platform for quantum transport, whose Hamiltonian reads
\begin{align} 
    \label{Eqn:HamiltonianSNS}
    \mkern-8mu
    H_\textsc{sns}
    \mkern-5mu = \mkern-5mu
    \sum\nolimits_j \mkern-4mu
    (t_j c^\dagger_j c_{j+1} + \Delta_j \mathrm{e}^{-\mathrm{i} \phi_j} c^\dagger_j c^\dagger_{j+1} + 
    \mkern-2mu \mathrm{H.c.} ) \mkern-2mu + \mkern-2mu \mu c^\dagger_j c_j\,,
\end{align}
where $\Delta_j$ is the superconducting gap with phase $\phi_j$ at site $j$, $\mu$ is the chemical potential and $t_j = t$. The number of sites in the left, middle, and right parts is $N_\textsc{l}, N_\textsc{m}, N_\textsc{r}$, respectively. The middle segment is normal metal by setting $\Delta_j=\phi_j=0$, $\forall j\in \mathbb{N} \equiv [N_\textsc{l},N_\textsc{l} + N_\textsc{m}]$. The outer segments are superconductors $\Delta_j=\Delta\neq0$ with phase bias applied such that $\phi_j=\phi$ in the right segment and $\phi_j=0$ in the left segment. %

As depicted in Fig.~\figref{Fig:SketchSpectrum}{a}, the SNS junction incorporates self-energies $\Sigma_\textsc{l} = \Sigma_1(0)$ and $\Sigma_\textsc{r} = \Sigma_\textsc{n}(0)$, after being connected to two separate reservoirs at its ends. This results in an effective NH Hamiltonian $\mathcal{H}_{\rm eff} = \mathcal{H}_\textsc{sns} + \Sigma_\textsc{l} + \Sigma_\textsc{r}$, whose spectrum exhibits pairs $(+\varepsilon_n, - \varepsilon^*_n)$ pertaining to the particle-hole symmetry of NH systems \cite{Kawabata2019Symmetry}. Fig.~\figref{Fig:SketchSpectrum}{c} highlights a pair of complex spectra $\varepsilon_\pm$ with EPs near $\phi=\pi$, where $\mathrm{Re}\, \varepsilon_\pm$ are pinned to zero. The calculation of supercurrents in the presence of EPs has recently garnered attention and sparked ongoing debates. As mentioned above, considering $\braket{\mathcal{J}}^\textsc{lr}_n = \mathbf{2} \partial_{\phi} \varepsilon_n$, a simple generalization of Eq.~\eqref{Eqn:CurrentHermitian} to complex eigenvalues $\sum_{\mathrm{Re}\, \varepsilon_n \leqslant 0} \partial_{\phi} \varepsilon_n$ \cite{Cayao2023NonHermitian,Li2024Anomalous} is equivalent to the LR-basis current $I_\textsc{lr}(\phi)$ \cite{Kornich2023CurrentVoltage}. However, as shown in Fig.~\figref{Fig:EPsComparison}{a}, this approach results in a divergent supercurrent due to the non-differentiable nature of EPs. On the other hand, the RR-basis current $I_\textsc{rr}(\phi)$ has a finite but non-smooth value at EPs, and also exhibits asymmetry around $\pi$. Furthermore, $I_\textsc{rr}(\phi)$ does not adhere to the local conservation law and will show distinct curves for different $j \in \mathbb{N}$ (see SM \cite{Note0} for details). In stark contrast, 
the current $I(\phi)$ computed by Eq.~\eqref{Eqn:CurrentImTrLog} exhibits no anomalies at EPs. It matches excellently with the current calculated from Eq.~\eqref{Eqn:CurrentHermitian} 
by exact diagonalization for the entire Hermitian $H_{\rm tot}$ with a large reservoir. 
Compared to current $I_{\rm iso}(\phi)$ in an isolated SNS junction, we observe an enhancement in $I(\phi)$ within the moderate coupling regime $\kappa = 0.4 t$. This seemingly counterintuitive effect arises because
the lower-energy mode $\varepsilon_-$ has a negative contribution to the current, which is effectively balanced by $\varepsilon_+$ due to level broadening in the NH case. However, as $\kappa/t$ increases, $I(\phi)$ starts to decrease, since dissipation also suppresses the positive current contribution from other states \cite{Note0}.

\begin{figure}
    \centering
    \includegraphics[width=\linewidth]{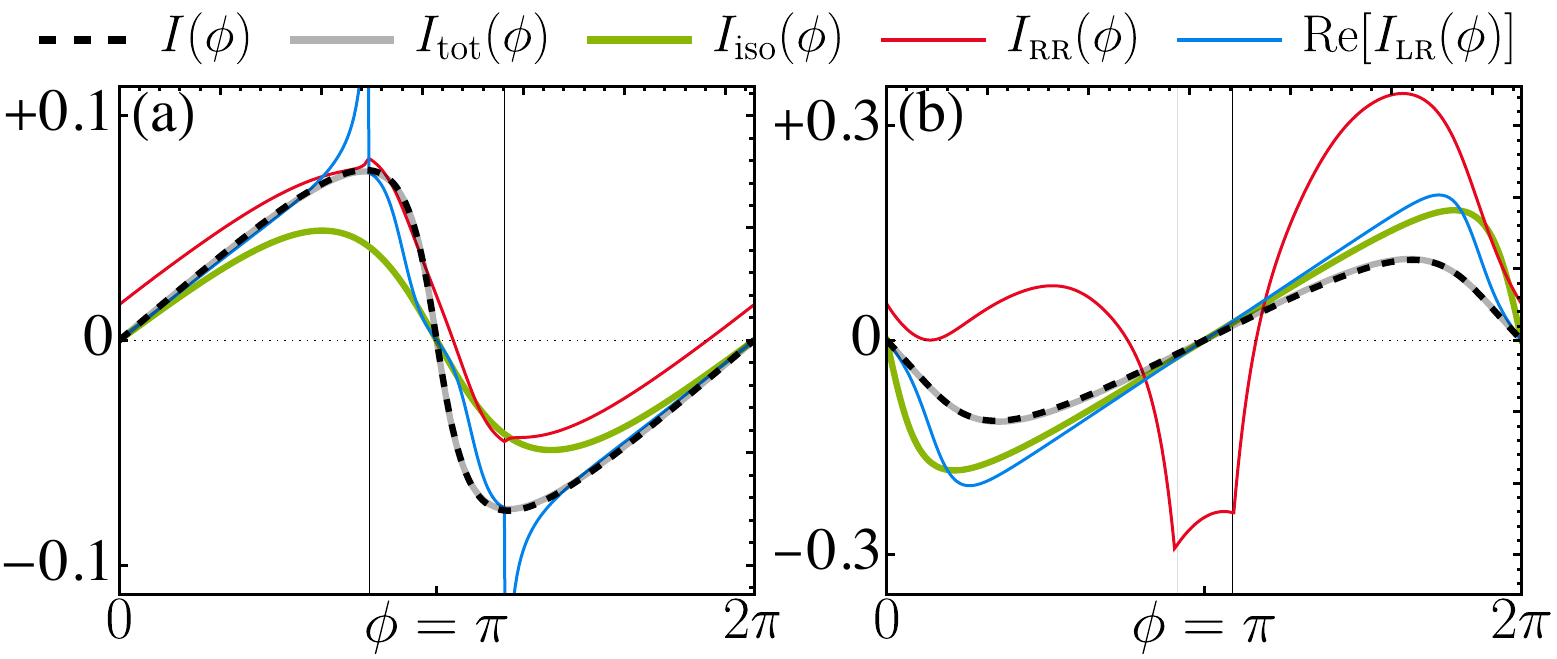}
    \caption{Methodological comparisons of persistent currents in two NH systems: (a) a phase-biased SNS junction; (b) a normal ring threaded by a magnetic flux. The current $I(\phi)$ (dashed) computed by Eq.~\eqref{Eqn:CurrentImTrLog} matches the exact diagonalization (gray) of the full Hermitian system. The isolated system (green) acts as a reference, showing an enhanced (reduced) current for the SNS (ring) model upon coupling to reservoirs. In (a), LR-basis currents (blue) diverge at the EPs (black lines), a feature not observed in (b), where a pair of EP-modes $\varepsilon_\pm$ cancel out the divergence. 
    RR-basis currents (red) violate the local conservation law and exhibit asymmetry around $\pi$.
    The parameters are the same as in Fig.~\ref{Fig:SketchSpectrum}.}
    \label{Fig:EPsComparison}
\end{figure}

\emph{NH normal rings.---}A mesoscopic ring threaded by a magnetic flux $\Phi$ also carries a persistent current because the coherence length of the wavefunction extends over its entire circumference \cite{Buttiker1983Josephson,Levy1990Magnetization,Bary-Soroker2008Effect,Bluhm2009Persistent,Bleszynski-Jayich2009Persistent}. The gauge-invariant tight-binding Hamiltonian is given by \cite{Carini1984Origin,Browne1984Periodicity,Cheung1988Persistent}
\begin{align} 
\label{Eqn:HamiltonianRing}
H_{\rm ring} =
    &\sum\nolimits_j (t_j \mathrm{e}^{-\mathrm{i} \phi_j} c^\dagger_j c_{j+1}  + \mathrm{H.c.} ) + \mu c^\dagger_j c_j \,,
\end{align}
where the normalized magnetic flux $\phi_\textsc{n} = \phi = 2 \pi \Phi/\Phi_0$ is placed between the $N$-th and first site, leaving other $\phi_j = 0$. Here, $\Phi_0 = h/e$ is the flux quantum that reflects the periodicity of Eq.~\eqref{Eqn:HamiltonianRing} in the flux $\Phi$ \cite{Byers1961Theoretical}. 
We account for elastic scatterings by assigning uniformly random hopping strengths along the ring \footnote{$\{ t_j \}_{j=1}^{6} \mkern-5mu = \mkern-5mu \{-0.859915, -0.884918, -0.918446, -0.846311,$ $-1.19937, \mkern-3mu -0.984676 \}$ for numerical results in the Letter.}. However, since the local conservation law spans the whole ring, $\braket{J_j}$ remains uniform across all sites.

As illustrated in Fig.~\figref{Fig:SketchSpectrum}{b}, the fermionic reservoir is connected to a single site within the ring \cite{Akkermans1991Relation}, inducing a self-energy $\Sigma \equiv \Sigma_\textsc{n}(0)$. Consequently, the ring is described by $\mathcal{H}_{\rm eff} = \mathcal{H}_{\rm ring} + \Sigma$. When $\mu = 0$, similar to the SNS junctions, the LR-basis current $I_\textsc{lr}(\phi)$ of the ring will diverge at EPs  near $\phi = \pi$.
Here, in order to explore different impacts of EPs,
we set $\mu = -1$ and shift two EP-modes $\varepsilon_\pm$ below the Fermi level, as depicted in Fig.~\figref{Fig:SketchSpectrum}{d}. Since $\varepsilon_\pm$ contribute to $I_\textsc{lr}(\phi)$ in pairs, their divergences cancel out, resulting in a smooth curvature for $I_\textsc{lr}(\phi)$ in Fig.~\figref{Fig:EPsComparison}{b}. 
The RR-basis current $I_\textsc{rr}(\phi)$ violates the local conservation law and exhibits a non-sinusoidal curve due to inhomogeneous hopping strengths. Neither of these approaches can accurately describe the persistent current $I(\phi)$. However, the current $I(\phi)$ calculated using Eq.~\eqref{Eqn:CurrentImTrLog} aligns with exact diagonalization results that include a large reservoir. Compared to the current in the isolated ring, $I(\phi)$ is reduced because the positive current contributions from single-particle modes are diluted by dissipation-induced level broadening \cite{Landauer1985Resistance}. These conclusions remain consistent regardless of the number of reservoirs.

\begin{figure*}
    \centering
    \includegraphics[width=\textwidth]{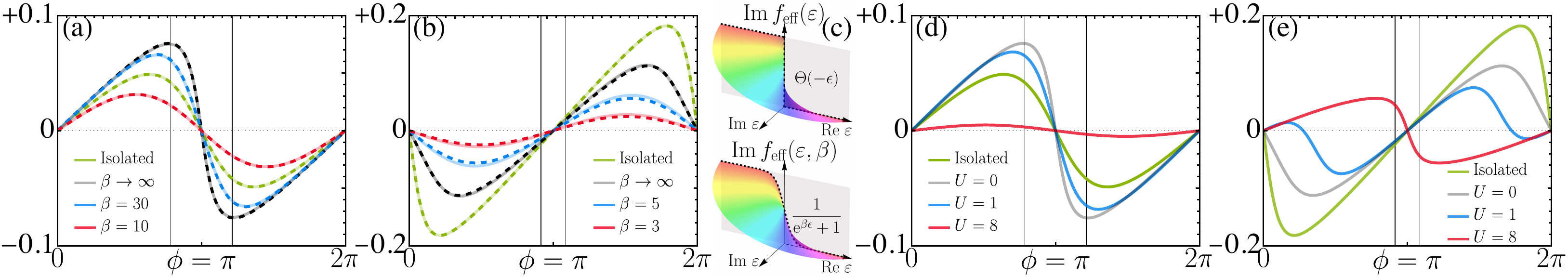}	
    \caption{Effects of temperature and interactions on persistent currents $I(\phi)$ in NH systems. Panels (a) and (b) display $I(\phi)$ at finite temperatures $\beta=1/k_\textsc{b}T$ for an SNS junction and a mesoscopic ring, compared to isolated systems at $T=0$ (green). The currents computed using Eq.~\eqref{Eqn:CurrentReTrLogGamma} (dashed) are consistent with the exact diagonalization (solid), indicating a decrease in $I(\phi)$ as $T$ increases. (c) Illustration of the imaginary part of the NH Fermi-Dirac distribution at zero and finite temperatures, where $\mathrm{Im}\, f_{\rm eff} (\varepsilon, \beta)$ in Eq.~\eqref{Eqn:EffectiveDistributionBeta} will reduce to the conventional Fermi-Dirac distribution (dashed) as $\mathrm{Im}\, \varepsilon \rightarrow 0$.
    In (d) and (e), many-body interactions are shown to suppress the amplitude of $I(\phi)$ in both systems. The parameters are as in Fig.~\ref{Fig:SketchSpectrum}. No singularities are found at EPs (black lines) in all cases presented.}
    \label{Fig:TemparatureInteraction}
\end{figure*}

\emph{Finite temperature and interaction effects.---}First, we consider the effect of thermal fluctuations on the persistent current \footnote{To maintain a unified formalism across normal rings and SNS junctions, we adopt a constant $\Delta$ and refrain from self-consistent calculations of observables for SNS junctions.}. This requires integrating $\rho(\omega)$ over $\omega$ using $f_\textsc{fd}(\omega, \beta) = 1/(\mathrm{e}^{\beta \omega} +1)$ with $\beta=1/k_\textsc{b} T$ and $k_\textsc{b}$ is the Boltzmann constant. Subsequently, $f_{\rm eff}(\varepsilon)$ in Eq.~\eqref{Eqn:HermitianO} is extended to form an effective
NH Fermi-Dirac distribution
at finite temperatures:
\begin{equation}
\label{Eqn:EffectiveDistributionBeta}
    f_{\rm eff}(\varepsilon, \beta) 
    = - \frac{1}{\pi} \left[ \Psi \left(\frac{1}{2}+\frac{\mathrm{i} \beta \varepsilon}{2\pi}\right) - \frac{\mathrm{i} \pi}{2} \right] \,,
\end{equation}
where the digamma function $\Psi$ \cite{Weisstein2024Digamma} is defined as the derivative of the log-gamma function $\mathrm{log\Gamma}$ \cite{Weisstein2024LogGamma}. Using Eq.~\eqref{Eqn:EffectiveDistributionBeta}, we find that at finite temperatures, Eq.~\eqref{Eqn:CurrentImTrLog} becomes (see SM \cite{Note0}):
\begin{equation}
\label{Eqn:CurrentReTrLogGamma}
    I(\phi, \beta) = \frac{2}{\beta} \frac{\mathrm{d}}{\mathrm{d} \phi} \mathrm{Re} \,\mathrm{Tr}\, \mathrm{log\Gamma} \left( \frac{1}{2} + \frac{\mathrm{i} \beta }{2\pi} \mathcal{H}_{\rm eff} \right) \,,
\end{equation}
which extends the expression $\mathbf{2}\partial_\phi F$ for the persistent current, where $F$ denotes the free energy in Hermitian systems \cite{Beenakker1992Superconducting}, to encompass NH scenarios.
As shown in Figs.~\figref{Fig:TemparatureInteraction}{a} and \figref{Fig:TemparatureInteraction}{b}, Eq.~\eqref{Eqn:CurrentReTrLogGamma} includes Eq.~\eqref{Eqn:CurrentImTrLog} when at $T=0$ and accurately matches the currents for $T\neq0$ calculated by exact diagonalization. This indicates a decrease in currents for both systems as $T$ increases. As illustrated in Fig.~\figref{Fig:TemparatureInteraction}{c}, such an excellent agreement is grounded on the fact that $\mathrm{Im}\, f_{\rm eff}(\varepsilon, \beta)$ in Eq.~\eqref{Eqn:EffectiveDistributionBeta} will revert to $f_\textsc{fd}(\varepsilon, \beta)$ as $\mathrm{Im}\, \varepsilon \rightarrow 0$. The smoothness of persistent currents near EPs can be attributed to the analytic properties of $f_{\rm eff}(\varepsilon, \beta)$ in the lower half of the complex plane \cite{Note0}.

To examine potential characteristics of EPs
in the presence of many-body interactions, we introduce the electrostatic repulsion $H_{\rm int} = U \sum_{j\in \mathbb{N}} (n_j - 1/2)(n_{j+1} - 1/2)$, where $U$ is the interaction strength. In such interacting scenarios, the first equality in Eq.~\eqref{Eqn:CurrentHermitian} remains valid for the ground state.
To maintain each reservoir as large as $N_\textsc{e} = 101$, we perform the density matrix renormalization group (DMRG) algorithm via DMRGpy \cite{Note10}. 
Figs.~\figref{Fig:TemparatureInteraction}{d} and \figref{Fig:TemparatureInteraction}{e} show that as $U$ increases, the amplitude of $I(\phi)$ in both systems will eventually be suppressed to zero due to the enhanced electron-electron scattering \footnote{In the moderate range $U$, the current amplitude may fluctuate in normal rings due to the shift of the effective Fermi level.}. No signatures of EPs are detected in the current in any of the cases presented with respect to temperatures and interactions.

\emph{Current susceptibility.---}To elucidate the presence of EPs in systems with a phase-dependent spectrum,
here we derive their linear response to a time-dependent phase driving $\phi(\tau)=\phi+\delta\phi(\tau)$, with $\delta\phi(\tau)\ll1$. The current susceptibility that characterizes the response is given by \cite{Trivedi1988Mesoscopic,Ferrier2013Phasedependent,Dassonneville2013Dissipation}: 
\begin{align}
\label{Eqn:CurrentSusceptibility}
\Pi(\phi,\tau) 
    &= - \mathrm{i} \Theta(\tau) \braket{[J(\tau), J(0)]} \,.
\end{align}
We first transform Eq.~\eqref{Eqn:CurrentSusceptibility} to the frequency space $\Pi (\phi, \omega) = \int \Pi(\phi, \tau) \mathrm{e}^{+\mathrm{i} \omega \tau} \mathrm{d} \tau $
and use the biorthogonal modes to obtain
\begin{align}
    \label{Eqn:ImaginarySusceptibility}
    \mathrm{Im}\, \Pi (\phi, \omega) 
    =&\, \pi t_j^2 \, \mathrm{Re} [\mathbb{P}(\phi, +\omega) -  \mathbb{P}(\phi, -\omega)] \,, \nonumber\\
    \mathbb{P}(\phi, \omega) 
    =&+P_{j+1,j,j+1,j} (\omega) - P_{j,j,j+1,j+1} (\omega) \nonumber \\
    &+ P_{j,j+1,j,j+1} (\omega) - P_{j+1,j+1,j,j} (\omega) \,, 
\end{align}
with $P_{ijkl} (\omega) \equiv \int \braket{i|\rho(\omega')|j} \braket{k|\rho(\omega + \omega')|l} f_\textsc{fd}(\omega')\mathrm{d} \omega'$. In the case  of 
SNS junctions, $\mathbb{P}(\phi, \omega)$ contains four additional terms stemming from the contributions of the holes
\footnote{The four additional anomalous contributions are
\begin{align*}
    &+P_{N+j+1,j,j+1,N+j} (\omega) - P_{N+j,j,j+1,N+j+1} (\omega) \\
    &+P_{N+j,j+1,j,N+j+1} (\omega) - P_{N+j+1,j+1,j,N+j} (\omega) \,.
\end{align*}
Due to the local conservation law, $\mathbb{P}(\phi, \omega)$ is uniform $\forall j \in \mathbb{N}$ and thus we set $j$ as the first site of $\mathbb{N}$ in the calculation.}.
Nevertheless, the integral $P_{ijkl} (\omega)$ is shared by both systems and possesses an analytical expression at $T=0$:
\begin{align}
    \label{Eqn:Pintegral}
    &P_{ijkl} (\omega)
    = \sum_{nm} \frac{p^{nm-}_{ijkl} + p^{n\tilde{m}+}_{ijkl}+ p^{\tilde{n}m-}_{ijkl}+ p^{\tilde{n}\tilde{m}+}_{ijkl}}{4} \,, \\
    &p^{nm\pm}_{ijkl} \equiv - \frac{1}{\pi} 
    \psi^\textsc{r}_{ni} \psi^\textsc{l*}_{nj} \psi^\textsc{r}_{mk} \psi^\textsc{l*}_{ml} \frac{f_{\rm eff} (\pm \varepsilon_n) - f_{\rm eff} (\pm \varepsilon_m \mp \omega)}{(\pm \varepsilon_n) - (\pm \varepsilon_m \mp \omega)} \,, \nonumber
\end{align}
where the tilde over $m$ conjugates the $m$-eigenvalue and exchange L $ \leftrightarrow$ R on the $m$-biorthogonal wavefunctions. Eq.~\eqref{Eqn:Pintegral} embeds $f_{\rm eff} (\varepsilon)$ and reduces to the Hermitian case $P_{ijkl}(\omega) = \sum_{nm} \psi_{ni}\psi^{*}_{nj}\psi_{mk}\psi^{*}_{ml} \Theta( - \epsilon_n) \delta(\omega + \epsilon_n - \epsilon_m)$ in the decoupled limit $\kappa \rightarrow 0$. In NH systems, $\mathrm{Im}\, \Pi (\phi, \omega)$ will peak at level transitions $\omega = \mathrm{Re}\, \varepsilon_m - \mathrm{Re}\, \varepsilon_n$ with a larger linewidth due to a finite $\mathrm{Im}\, \varepsilon_m$. As shown in Fig.~\ref{Fig:Susceptibility}, this broadened effect is more evident when transitions between levels encounter EPs. As $\kappa/t$ increases, these peaks will be significantly enhanced and accumulate towards the regions between EPs. 
Our results agree with the full exact diagonalization \cite{Note0} and are consistent with the Lindblad formalism \cite{Minganti2019Quantum}: (i) the effect of EPs cannot be observed in the steady state (including equilibrium); (ii) any manifestation of an EP has a dynamical nature.

\begin{figure}[b]
    \centering
    \includegraphics[width=\linewidth]{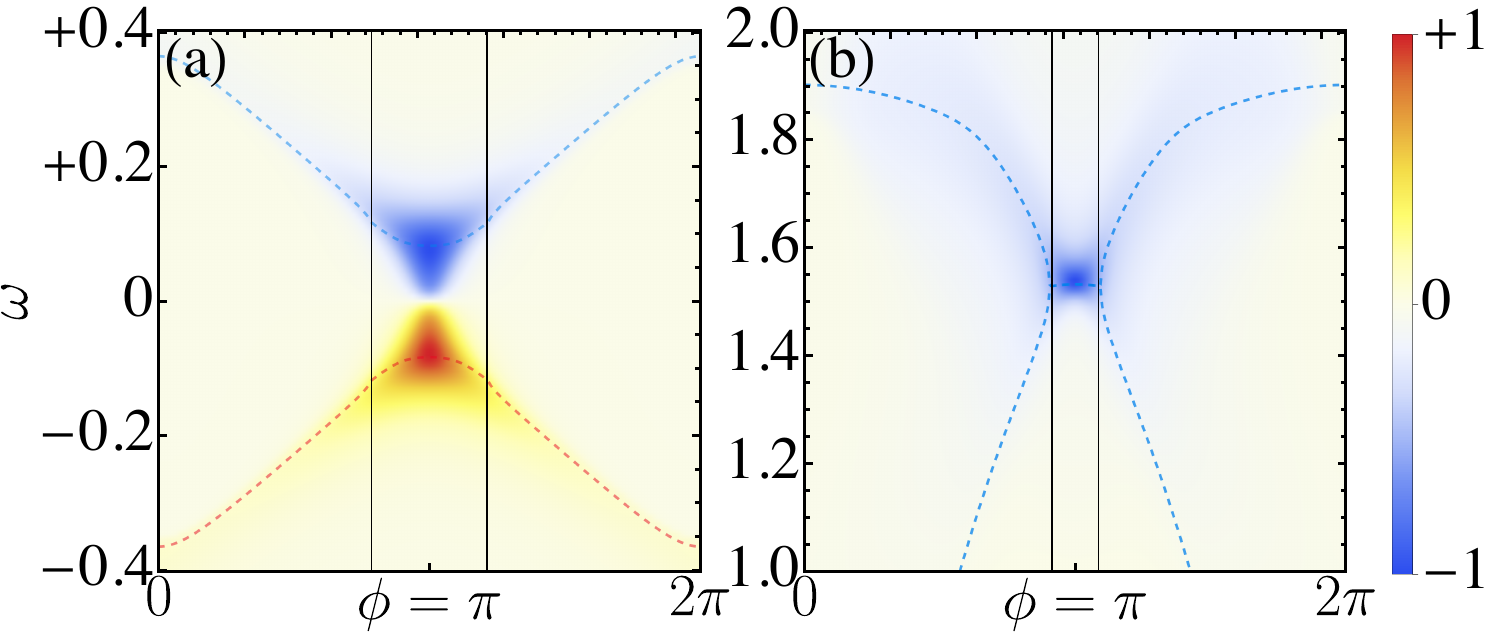}	
    \caption{Normalized imaginary part of current susceptibility $\Pi (\phi, \omega)$ of NH systems. (a)  $\mathrm{Im}\,\Pi (\phi, \omega)$ for the SNS junction, incorporating negative $\omega$ to reflect the particle-hole symmetry. (b) $\mathrm{Im}\,\Pi (\phi, \omega)$ for the normal mesoscopic ring. Peaks of $\mathrm{Im}\,\Pi (\phi, \omega)$ correspond to energy level transitions (dashed lines) and are concentrated between two EPs (black lines). The parameters are the same as in Fig.~\ref{Fig:SketchSpectrum}.}
    \label{Fig:Susceptibility}
\end{figure}

\emph{Conclusions and outlook.---}In this Letter, we identified an effective distribution 
that captures the quantum many-body observables of NH fermionic systems in equilibrium. 
This distribution, derived microscopically from the biorthogonal Green's function, serves as an extension of the Fermi-Dirac distribution for NH systems. 
We utilized this formalism in the context of quantum transport and derived an analytical equation %
for the persistent current flowing in SNS junctions and normal mesoscopic rings connected to reservoirs.
We demonstrated that there are no anomalies in the persistent currents near EPs, showing that their amplitudes are suppressed by thermal fluctuations and many-body interactions. 
Our findings have been validated through exact diagonalization with excellent agreement. We conclude that the signatures of EPs are only discernible in a dynamical quantity---the current susceptibility---rather than a static observable.

Our formalism extends beyond quantum persistent current transport and holds promise for broader applications. It can be adapted to systems such as multi-Josephson junctions \cite{Riwar2016Multiterminal,Pankratova2020Multiterminal,Coraiola2023Phaseengineering} and quantum spin chains \cite{Shen2021Theory}, potentially unveiling new insights into their topological and entanglement characteristics. 
Furthermore, generalizing this formalism to encompass non-equilibrium scenarios, such as quantum pumps \cite{Moskalets2002Floquet,Blaauboer2002Charge,Becerra2023Quantized}, shows great potential.

\emph{Note added.---}Our formalism fully agrees with the scattering matrix theory \cite{Beenakker2024Josephson} and can be applied to calculating additional thermodynamic quantities \cite{Pino2024Thermodynamics}.

\emph{Acknowledgments.---}We thank C.~W.~J. Beenakker, D.-L. Deng, W. Brzezicki, W.-T. Xue, L.-W. Yu, S.-Y. Zhang, W. Li, and Z. Liu for helpful discussions, and H.-R. Wang, in particular, for his valuable feedback from reading the first version of this manuscript. This work is supported by the Foundation for Polish Science project MagTop (No.~FENG.02.01-IP.05-0028/23) co-financed by the European Union from the funds of Priority 2 of the European Funds for a Smart Economy Program 2021–2027 (FENG) and by the National Science Centre (Poland) OPUS Grant No.~2021/41/B/ST3/04475. P.-X.S. and Z.L. acknowledge support from the Tsinghua University Dushi Program and Shanghai Qi Zhi Institute. J.L.L. acknowledges the computational resources provided by the Aalto Science-IT project the financial support from the Academy of Finland Projects No.~331342 and No.~358088. P.-X.S. acknowledges additional support from the European Union's Horizon Europe research and innovation programme under the Marie Skłodowska-Curie Grant Agreement No.~101180589 (SymPhysAI). Views and opinions expressed are however those of the author(s) only and do not necessarily reflect those of the European Union or the European Research Executive Agency. Neither the European Union nor the granting authority can be held responsible for them.

\footnotetext[0]{See Supplemental Material at [URL will be inserted by publisher] for: (i) Green's function and self-energy, (ii) properties of the NH Fermi-Dirac distribution, (iii) derivation of persistent current formula and its continuity at EPs, and (iv) extended numerical results on low-energy spectra, persistent currents, and current susceptibility, which includes Ref.~\cite{Shen2023Majoranamagnon}.}

\urlstyle{same}
\footnotetext[10]{The source code is available at \url{https://github.com/peixinshen/NonHermitianFermiDiracDistributionPersistentCurrent}}

\bibliography{refs}

\begin{thebibliography}{108}%
\makeatletter
\providecommand \@ifxundefined [1]{%
 \@ifx{#1\undefined}
}%
\providecommand \@ifnum [1]{%
 \ifnum #1\expandafter \@firstoftwo
 \else \expandafter \@secondoftwo
 \fi
}%
\providecommand \@ifx [1]{%
 \ifx #1\expandafter \@firstoftwo
 \else \expandafter \@secondoftwo
 \fi
}%
\providecommand \natexlab [1]{#1}%
\providecommand \enquote  [1]{``#1''}%
\providecommand \bibnamefont  [1]{#1}%
\providecommand \bibfnamefont [1]{#1}%
\providecommand \citenamefont [1]{#1}%
\providecommand \href@noop [0]{\@secondoftwo}%
\providecommand \href [0]{\begingroup \@sanitize@url \@href}%
\providecommand \@href[1]{\@@startlink{#1}\@@href}%
\providecommand \@@href[1]{\endgroup#1\@@endlink}%
\providecommand \@sanitize@url [0]{\catcode `\\12\catcode `\$12\catcode `\&12\catcode `\#12\catcode `\^12\catcode `\_12\catcode `\%12\relax}%
\providecommand \@@startlink[1]{}%
\providecommand \@@endlink[0]{}%
\providecommand \url  [0]{\begingroup\@sanitize@url \@url }%
\providecommand \@url [1]{\endgroup\@href {#1}{\urlprefix }}%
\providecommand \urlprefix  [0]{URL }%
\providecommand \Eprint [0]{\href }%
\providecommand \doibase [0]{https://doi.org/}%
\providecommand \selectlanguage [0]{\@gobble}%
\providecommand \bibinfo  [0]{\@secondoftwo}%
\providecommand \bibfield  [0]{\@secondoftwo}%
\providecommand \translation [1]{[#1]}%
\providecommand \BibitemOpen [0]{}%
\providecommand \bibitemStop [0]{}%
\providecommand \bibitemNoStop [0]{.\EOS\space}%
\providecommand \EOS [0]{\spacefactor3000\relax}%
\providecommand \BibitemShut  [1]{\csname bibitem#1\endcsname}%
\let\auto@bib@innerbib\@empty
\bibitem [{\citenamefont {Ashida}\ \emph {et~al.}(2020)\citenamefont {Ashida}, \citenamefont {Gong},\ and\ \citenamefont {Ueda}}]{Ashida2020NonHermitian}%
  \BibitemOpen
  \bibfield  {author} {\bibinfo {author} {\bibfnamefont {Y.}~\bibnamefont {Ashida}}, \bibinfo {author} {\bibfnamefont {Z.}~\bibnamefont {Gong}},\ and\ \bibinfo {author} {\bibfnamefont {M.}~\bibnamefont {Ueda}},\ }\bibfield  {title} {\bibinfo {title} {Non-{{Hermitian}} physics},\ }\href {https://doi.org/10.1080/00018732.2021.1876991} {\bibfield  {journal} {\bibinfo  {journal} {Adv. Phys.}\ }\textbf {\bibinfo {volume} {69}},\ \bibinfo {pages} {249} (\bibinfo {year} {2020})}\BibitemShut {NoStop}%
\bibitem [{\citenamefont {Bergholtz}\ \emph {et~al.}(2021)\citenamefont {Bergholtz}, \citenamefont {Budich},\ and\ \citenamefont {Kunst}}]{Bergholtz2021Exceptional}%
  \BibitemOpen
  \bibfield  {author} {\bibinfo {author} {\bibfnamefont {E.~J.}\ \bibnamefont {Bergholtz}}, \bibinfo {author} {\bibfnamefont {J.~C.}\ \bibnamefont {Budich}},\ and\ \bibinfo {author} {\bibfnamefont {F.~K.}\ \bibnamefont {Kunst}},\ }\bibfield  {title} {\bibinfo {title} {Exceptional topology of non-{{Hermitian}} systems},\ }\href {https://doi.org/10.1103/RevModPhys.93.015005} {\bibfield  {journal} {\bibinfo  {journal} {Rev. Mod. Phys.}\ }\textbf {\bibinfo {volume} {93}},\ \bibinfo {pages} {015005} (\bibinfo {year} {2021})}\BibitemShut {NoStop}%
\bibitem [{\citenamefont {Ding}\ \emph {et~al.}(2022)\citenamefont {Ding}, \citenamefont {Fang},\ and\ \citenamefont {Ma}}]{Ding2022NonHermitian}%
  \BibitemOpen
  \bibfield  {author} {\bibinfo {author} {\bibfnamefont {K.}~\bibnamefont {Ding}}, \bibinfo {author} {\bibfnamefont {C.}~\bibnamefont {Fang}},\ and\ \bibinfo {author} {\bibfnamefont {G.}~\bibnamefont {Ma}},\ }\bibfield  {title} {\bibinfo {title} {Non-{{Hermitian}} topology and exceptional-point geometries},\ }\href {https://doi.org/10.1038/s42254-022-00516-5} {\bibfield  {journal} {\bibinfo  {journal} {Nat. Rev. Phys.}\ }\textbf {\bibinfo {volume} {4}},\ \bibinfo {pages} {745} (\bibinfo {year} {2022})}\BibitemShut {NoStop}%
\bibitem [{\citenamefont {Okuma}\ and\ \citenamefont {Sato}(2023)}]{Okuma2023NonHermitian}%
  \BibitemOpen
  \bibfield  {author} {\bibinfo {author} {\bibfnamefont {N.}~\bibnamefont {Okuma}}\ and\ \bibinfo {author} {\bibfnamefont {M.}~\bibnamefont {Sato}},\ }\bibfield  {title} {\bibinfo {title} {Non-{{Hermitian Topological Phenomena}}: {{A Review}}},\ }\href {https://doi.org/10.1146/annurev-conmatphys-040521-033133} {\bibfield  {journal} {\bibinfo  {journal} {Annu. Rev. Condens. Matter Phys.}\ }\textbf {\bibinfo {volume} {14}},\ \bibinfo {pages} {83} (\bibinfo {year} {2023})}\BibitemShut {NoStop}%
\bibitem [{\citenamefont {Xue}\ \emph {et~al.}(2021)\citenamefont {Xue}, \citenamefont {Li}, \citenamefont {Hu}, \citenamefont {Song},\ and\ \citenamefont {Wang}}]{Xue2021Simple}%
  \BibitemOpen
  \bibfield  {author} {\bibinfo {author} {\bibfnamefont {W.-T.}\ \bibnamefont {Xue}}, \bibinfo {author} {\bibfnamefont {M.-R.}\ \bibnamefont {Li}}, \bibinfo {author} {\bibfnamefont {Y.-M.}\ \bibnamefont {Hu}}, \bibinfo {author} {\bibfnamefont {F.}~\bibnamefont {Song}},\ and\ \bibinfo {author} {\bibfnamefont {Z.}~\bibnamefont {Wang}},\ }\bibfield  {title} {\bibinfo {title} {Simple formulas of directional amplification from non-{{Bloch}} band theory},\ }\href {https://doi.org/10.1103/PhysRevB.103.L241408} {\bibfield  {journal} {\bibinfo  {journal} {Phys. Rev. B}\ }\textbf {\bibinfo {volume} {103}},\ \bibinfo {pages} {L241408} (\bibinfo {year} {2021})}\BibitemShut {NoStop}%
\bibitem [{\citenamefont {Xue}\ \emph {et~al.}(2022)\citenamefont {Xue}, \citenamefont {Hu}, \citenamefont {Song},\ and\ \citenamefont {Wang}}]{Xue2022NonHermitian}%
  \BibitemOpen
  \bibfield  {author} {\bibinfo {author} {\bibfnamefont {W.-T.}\ \bibnamefont {Xue}}, \bibinfo {author} {\bibfnamefont {Y.-M.}\ \bibnamefont {Hu}}, \bibinfo {author} {\bibfnamefont {F.}~\bibnamefont {Song}},\ and\ \bibinfo {author} {\bibfnamefont {Z.}~\bibnamefont {Wang}},\ }\bibfield  {title} {\bibinfo {title} {Non-{{Hermitian Edge Burst}}},\ }\href {https://doi.org/10.1103/PhysRevLett.128.120401} {\bibfield  {journal} {\bibinfo  {journal} {Phys. Rev. Lett.}\ }\textbf {\bibinfo {volume} {128}},\ \bibinfo {pages} {120401} (\bibinfo {year} {2022})}\BibitemShut {NoStop}%
\bibitem [{\citenamefont {Hu}\ \emph {et~al.}(2023)\citenamefont {Hu}, \citenamefont {Xue}, \citenamefont {Song},\ and\ \citenamefont {Wang}}]{Hu2023Steadystate}%
  \BibitemOpen
  \bibfield  {author} {\bibinfo {author} {\bibfnamefont {Y.-M.}\ \bibnamefont {Hu}}, \bibinfo {author} {\bibfnamefont {W.-T.}\ \bibnamefont {Xue}}, \bibinfo {author} {\bibfnamefont {F.}~\bibnamefont {Song}},\ and\ \bibinfo {author} {\bibfnamefont {Z.}~\bibnamefont {Wang}},\ }\bibfield  {title} {\bibinfo {title} {Steady-state edge burst: {{From}} free-particle systems to interaction-induced phenomena},\ }\href {https://doi.org/10.1103/PhysRevB.108.235422} {\bibfield  {journal} {\bibinfo  {journal} {Phys. Rev. B}\ }\textbf {\bibinfo {volume} {108}},\ \bibinfo {pages} {235422} (\bibinfo {year} {2023})}\BibitemShut {NoStop}%
\bibitem [{\citenamefont {Li}\ \emph {et~al.}(2024{\natexlab{a}})\citenamefont {Li}, \citenamefont {Wang}, \citenamefont {Song},\ and\ \citenamefont {Wang}}]{Li2024NonBloch}%
  \BibitemOpen
  \bibfield  {author} {\bibinfo {author} {\bibfnamefont {B.}~\bibnamefont {Li}}, \bibinfo {author} {\bibfnamefont {H.-R.}\ \bibnamefont {Wang}}, \bibinfo {author} {\bibfnamefont {F.}~\bibnamefont {Song}},\ and\ \bibinfo {author} {\bibfnamefont {Z.}~\bibnamefont {Wang}},\ }\bibfield  {title} {\bibinfo {title} {Non-{{Bloch}} dynamics and topology in a classical nonequilibrium process},\ }\href {https://doi.org/10.1103/PhysRevB.109.L201121} {\bibfield  {journal} {\bibinfo  {journal} {Phys. Rev. B}\ }\textbf {\bibinfo {volume} {109}},\ \bibinfo {pages} {L201121} (\bibinfo {year} {2024}{\natexlab{a}})}\BibitemShut {NoStop}%
\bibitem [{\citenamefont {Zou}\ \emph {et~al.}(2024)\citenamefont {Zou}, \citenamefont {Bosco}, \citenamefont {Thingstad}, \citenamefont {Klinovaja},\ and\ \citenamefont {Loss}}]{Zou2024Dissipative}%
  \BibitemOpen
  \bibfield  {author} {\bibinfo {author} {\bibfnamefont {J.}~\bibnamefont {Zou}}, \bibinfo {author} {\bibfnamefont {S.}~\bibnamefont {Bosco}}, \bibinfo {author} {\bibfnamefont {E.}~\bibnamefont {Thingstad}}, \bibinfo {author} {\bibfnamefont {J.}~\bibnamefont {Klinovaja}},\ and\ \bibinfo {author} {\bibfnamefont {D.}~\bibnamefont {Loss}},\ }\bibfield  {title} {\bibinfo {title} {Dissipative {{Spin-Wave Diode}} and {{Nonreciprocal Magnonic Amplifier}}},\ }\href {https://doi.org/10.1103/PhysRevLett.132.036701} {\bibfield  {journal} {\bibinfo  {journal} {Phys. Rev. Lett.}\ }\textbf {\bibinfo {volume} {132}},\ \bibinfo {pages} {036701} (\bibinfo {year} {2024})}\BibitemShut {NoStop}%
\bibitem [{\citenamefont {Yu}\ \emph {et~al.}(2022)\citenamefont {Yu}, \citenamefont {Yu}, \citenamefont {Zhang}, \citenamefont {Zhang}, \citenamefont {Ouyang}, \citenamefont {Liu}, \citenamefont {Deng},\ and\ \citenamefont {Duan}}]{Yu2022Experimental}%
  \BibitemOpen
  \bibfield  {author} {\bibinfo {author} {\bibfnamefont {Y.}~\bibnamefont {Yu}}, \bibinfo {author} {\bibfnamefont {L.-W.}\ \bibnamefont {Yu}}, \bibinfo {author} {\bibfnamefont {W.}~\bibnamefont {Zhang}}, \bibinfo {author} {\bibfnamefont {H.}~\bibnamefont {Zhang}}, \bibinfo {author} {\bibfnamefont {X.}~\bibnamefont {Ouyang}}, \bibinfo {author} {\bibfnamefont {Y.}~\bibnamefont {Liu}}, \bibinfo {author} {\bibfnamefont {D.-L.}\ \bibnamefont {Deng}},\ and\ \bibinfo {author} {\bibfnamefont {L.-M.}\ \bibnamefont {Duan}},\ }\bibfield  {title} {\bibinfo {title} {Experimental unsupervised learning of non-{{Hermitian}} knotted phases with solid-state spins},\ }\href {https://doi.org/10.1038/s41534-022-00629-w} {\bibfield  {journal} {\bibinfo  {journal} {npj Quantum Inf.}\ }\textbf {\bibinfo {volume} {8}},\ \bibinfo {pages} {116} (\bibinfo {year} {2022})}\BibitemShut {NoStop}%
\bibitem [{\citenamefont {Chen}\ \emph {et~al.}(2023)\citenamefont {Chen}, \citenamefont {Song},\ and\ \citenamefont {Lado}}]{Chen2023Topological}%
  \BibitemOpen
  \bibfield  {author} {\bibinfo {author} {\bibfnamefont {G.}~\bibnamefont {Chen}}, \bibinfo {author} {\bibfnamefont {F.}~\bibnamefont {Song}},\ and\ \bibinfo {author} {\bibfnamefont {J.~L.}\ \bibnamefont {Lado}},\ }\bibfield  {title} {\bibinfo {title} {Topological {{Spin Excitations}} in {{Non-Hermitian Spin Chains}} with a {{Generalized Kernel Polynomial Algorithm}}},\ }\href {https://doi.org/10.1103/PhysRevLett.130.100401} {\bibfield  {journal} {\bibinfo  {journal} {Phys. Rev. Lett.}\ }\textbf {\bibinfo {volume} {130}},\ \bibinfo {pages} {100401} (\bibinfo {year} {2023})}\BibitemShut {NoStop}%
\bibitem [{\citenamefont {Ochkan}\ \emph {et~al.}(2024)\citenamefont {Ochkan}, \citenamefont {Chaturvedi}, \citenamefont {K{\"o}nye}, \citenamefont {Veyrat}, \citenamefont {Giraud}, \citenamefont {Mailly}, \citenamefont {Cavanna}, \citenamefont {Gennser}, \citenamefont {Hankiewicz}, \citenamefont {B{\"u}chner}, \citenamefont {{van den Brink}}, \citenamefont {Dufouleur},\ and\ \citenamefont {Fulga}}]{Ochkan2024NonHermitian}%
  \BibitemOpen
  \bibfield  {author} {\bibinfo {author} {\bibfnamefont {K.}~\bibnamefont {Ochkan}}, \bibinfo {author} {\bibfnamefont {R.}~\bibnamefont {Chaturvedi}}, \bibinfo {author} {\bibfnamefont {V.}~\bibnamefont {K{\"o}nye}}, \bibinfo {author} {\bibfnamefont {L.}~\bibnamefont {Veyrat}}, \bibinfo {author} {\bibfnamefont {R.}~\bibnamefont {Giraud}}, \bibinfo {author} {\bibfnamefont {D.}~\bibnamefont {Mailly}}, \bibinfo {author} {\bibfnamefont {A.}~\bibnamefont {Cavanna}}, \bibinfo {author} {\bibfnamefont {U.}~\bibnamefont {Gennser}}, \bibinfo {author} {\bibfnamefont {E.~M.}\ \bibnamefont {Hankiewicz}}, \bibinfo {author} {\bibfnamefont {B.}~\bibnamefont {B{\"u}chner}}, \bibinfo {author} {\bibfnamefont {J.}~\bibnamefont {{van den Brink}}}, \bibinfo {author} {\bibfnamefont {J.}~\bibnamefont {Dufouleur}},\ and\ \bibinfo {author} {\bibfnamefont {I.~C.}\ \bibnamefont {Fulga}},\ }\bibfield  {title} {\bibinfo {title} {Non-{{Hermitian}} topology in a multi-terminal quantum {{Hall}} device},\ }\href {https://doi.org/10.1038/s41567-023-02337-4} {\bibfield  {journal} {\bibinfo  {journal} {Nat. Phys.}\ }\textbf {\bibinfo {volume} {20}},\ \bibinfo {pages} {395} (\bibinfo {year} {2024})}\BibitemShut {NoStop}%
\bibitem [{\citenamefont {Kunst}\ \emph {et~al.}(2018)\citenamefont {Kunst}, \citenamefont {Edvardsson}, \citenamefont {Budich},\ and\ \citenamefont {Bergholtz}}]{Kunst2018Biorthogonal}%
  \BibitemOpen
  \bibfield  {author} {\bibinfo {author} {\bibfnamefont {F.~K.}\ \bibnamefont {Kunst}}, \bibinfo {author} {\bibfnamefont {E.}~\bibnamefont {Edvardsson}}, \bibinfo {author} {\bibfnamefont {J.~C.}\ \bibnamefont {Budich}},\ and\ \bibinfo {author} {\bibfnamefont {E.~J.}\ \bibnamefont {Bergholtz}},\ }\bibfield  {title} {\bibinfo {title} {Biorthogonal {{Bulk-Boundary Correspondence}} in {{Non-Hermitian Systems}}},\ }\href {https://doi.org/10.1103/PhysRevLett.121.026808} {\bibfield  {journal} {\bibinfo  {journal} {Phys. Rev. Lett.}\ }\textbf {\bibinfo {volume} {121}},\ \bibinfo {pages} {026808} (\bibinfo {year} {2018})}\BibitemShut {NoStop}%
\bibitem [{\citenamefont {Yao}\ and\ \citenamefont {Wang}(2018)}]{Yao2018Edge}%
  \BibitemOpen
  \bibfield  {author} {\bibinfo {author} {\bibfnamefont {S.}~\bibnamefont {Yao}}\ and\ \bibinfo {author} {\bibfnamefont {Z.}~\bibnamefont {Wang}},\ }\bibfield  {title} {\bibinfo {title} {Edge {{States}} and {{Topological Invariants}} of {{Non-Hermitian Systems}}},\ }\href {https://doi.org/10.1103/PhysRevLett.121.086803} {\bibfield  {journal} {\bibinfo  {journal} {Phys. Rev. Lett.}\ }\textbf {\bibinfo {volume} {121}},\ \bibinfo {pages} {086803} (\bibinfo {year} {2018})}\BibitemShut {NoStop}%
\bibitem [{\citenamefont {Yao}\ \emph {et~al.}(2018)\citenamefont {Yao}, \citenamefont {Song},\ and\ \citenamefont {Wang}}]{Yao2018NonHermitian}%
  \BibitemOpen
  \bibfield  {author} {\bibinfo {author} {\bibfnamefont {S.}~\bibnamefont {Yao}}, \bibinfo {author} {\bibfnamefont {F.}~\bibnamefont {Song}},\ and\ \bibinfo {author} {\bibfnamefont {Z.}~\bibnamefont {Wang}},\ }\bibfield  {title} {\bibinfo {title} {Non-{{Hermitian Chern Bands}}},\ }\href {https://doi.org/10.1103/PhysRevLett.121.136802} {\bibfield  {journal} {\bibinfo  {journal} {Phys. Rev. Lett.}\ }\textbf {\bibinfo {volume} {121}},\ \bibinfo {pages} {136802} (\bibinfo {year} {2018})}\BibitemShut {NoStop}%
\bibitem [{\citenamefont {Yang}\ \emph {et~al.}(2020)\citenamefont {Yang}, \citenamefont {Zhang}, \citenamefont {Fang},\ and\ \citenamefont {Hu}}]{Yang2020NonHermitian}%
  \BibitemOpen
  \bibfield  {author} {\bibinfo {author} {\bibfnamefont {Z.}~\bibnamefont {Yang}}, \bibinfo {author} {\bibfnamefont {K.}~\bibnamefont {Zhang}}, \bibinfo {author} {\bibfnamefont {C.}~\bibnamefont {Fang}},\ and\ \bibinfo {author} {\bibfnamefont {J.}~\bibnamefont {Hu}},\ }\bibfield  {title} {\bibinfo {title} {Non-{{Hermitian Bulk-Boundary Correspondence}} and {{Auxiliary Generalized Brillouin Zone Theory}}},\ }\href {https://doi.org/10.1103/PhysRevLett.125.226402} {\bibfield  {journal} {\bibinfo  {journal} {Phys. Rev. Lett.}\ }\textbf {\bibinfo {volume} {125}},\ \bibinfo {pages} {226402} (\bibinfo {year} {2020})}\BibitemShut {NoStop}%
\bibitem [{\citenamefont {Bernard}\ and\ \citenamefont {LeClair}(2002)}]{Bernard2002Classification}%
  \BibitemOpen
  \bibfield  {author} {\bibinfo {author} {\bibfnamefont {D.}~\bibnamefont {Bernard}}\ and\ \bibinfo {author} {\bibfnamefont {A.}~\bibnamefont {LeClair}},\ }\bibfield  {title} {\bibinfo {title} {A {{Classification}} of {{Non-Hermitian Random Matrices}}},\ }in\ \href {https://doi.org/10.1007/978-94-010-0514-2_19} {\emph {\bibinfo {booktitle} {Statistical {{Field Theories}}}}},\ \bibinfo {series and number} {{{NATO Science Series}}},\ \bibinfo {editor} {edited by\ \bibinfo {editor} {\bibfnamefont {A.}~\bibnamefont {Cappelli}}\ and\ \bibinfo {editor} {\bibfnamefont {G.}~\bibnamefont {Mussardo}}}\ (\bibinfo  {publisher} {Springer Netherlands},\ \bibinfo {address} {Dordrecht},\ \bibinfo {year} {2002})\ pp.\ \bibinfo {pages} {207--214}\BibitemShut {NoStop}%
\bibitem [{\citenamefont {Kawabata}\ \emph {et~al.}(2019)\citenamefont {Kawabata}, \citenamefont {Shiozaki}, \citenamefont {Ueda},\ and\ \citenamefont {Sato}}]{Kawabata2019Symmetry}%
  \BibitemOpen
  \bibfield  {author} {\bibinfo {author} {\bibfnamefont {K.}~\bibnamefont {Kawabata}}, \bibinfo {author} {\bibfnamefont {K.}~\bibnamefont {Shiozaki}}, \bibinfo {author} {\bibfnamefont {M.}~\bibnamefont {Ueda}},\ and\ \bibinfo {author} {\bibfnamefont {M.}~\bibnamefont {Sato}},\ }\bibfield  {title} {\bibinfo {title} {Symmetry and {{Topology}} in {{Non-Hermitian Physics}}},\ }\href {https://doi.org/10.1103/PhysRevX.9.041015} {\bibfield  {journal} {\bibinfo  {journal} {Phys. Rev. X}\ }\textbf {\bibinfo {volume} {9}},\ \bibinfo {pages} {041015} (\bibinfo {year} {2019})}\BibitemShut {NoStop}%
\bibitem [{\citenamefont {Altland}\ \emph {et~al.}(2021)\citenamefont {Altland}, \citenamefont {Fleischhauer},\ and\ \citenamefont {Diehl}}]{Altland2021Symmetry}%
  \BibitemOpen
  \bibfield  {author} {\bibinfo {author} {\bibfnamefont {A.}~\bibnamefont {Altland}}, \bibinfo {author} {\bibfnamefont {M.}~\bibnamefont {Fleischhauer}},\ and\ \bibinfo {author} {\bibfnamefont {S.}~\bibnamefont {Diehl}},\ }\bibfield  {title} {\bibinfo {title} {Symmetry {{Classes}} of {{Open Fermionic Quantum Matter}}},\ }\href {https://doi.org/10.1103/PhysRevX.11.021037} {\bibfield  {journal} {\bibinfo  {journal} {Phys. Rev. X}\ }\textbf {\bibinfo {volume} {11}},\ \bibinfo {pages} {021037} (\bibinfo {year} {2021})}\BibitemShut {NoStop}%
\bibitem [{\citenamefont {Yu}\ and\ \citenamefont {Deng}(2021)}]{Yu2021Unsupervised}%
  \BibitemOpen
  \bibfield  {author} {\bibinfo {author} {\bibfnamefont {L.-W.}\ \bibnamefont {Yu}}\ and\ \bibinfo {author} {\bibfnamefont {D.-L.}\ \bibnamefont {Deng}},\ }\bibfield  {title} {\bibinfo {title} {Unsupervised {{Learning}} of {{Non-Hermitian Topological Phases}}},\ }\href {https://doi.org/10.1103/PhysRevLett.126.240402} {\bibfield  {journal} {\bibinfo  {journal} {Phys. Rev. Lett.}\ }\textbf {\bibinfo {volume} {126}},\ \bibinfo {pages} {240402} (\bibinfo {year} {2021})}\BibitemShut {NoStop}%
\bibitem [{\citenamefont {Denner}\ \emph {et~al.}(2021)\citenamefont {Denner}, \citenamefont {Skurativska}, \citenamefont {Schindler}, \citenamefont {Fischer}, \citenamefont {Thomale}, \citenamefont {Bzdu{\v s}ek},\ and\ \citenamefont {Neupert}}]{Denner2021Exceptional}%
  \BibitemOpen
  \bibfield  {author} {\bibinfo {author} {\bibfnamefont {M.~M.}\ \bibnamefont {Denner}}, \bibinfo {author} {\bibfnamefont {A.}~\bibnamefont {Skurativska}}, \bibinfo {author} {\bibfnamefont {F.}~\bibnamefont {Schindler}}, \bibinfo {author} {\bibfnamefont {M.~H.}\ \bibnamefont {Fischer}}, \bibinfo {author} {\bibfnamefont {R.}~\bibnamefont {Thomale}}, \bibinfo {author} {\bibfnamefont {T.}~\bibnamefont {Bzdu{\v s}ek}},\ and\ \bibinfo {author} {\bibfnamefont {T.}~\bibnamefont {Neupert}},\ }\bibfield  {title} {\bibinfo {title} {Exceptional topological insulators},\ }\href {https://doi.org/10.1038/s41467-021-25947-z} {\bibfield  {journal} {\bibinfo  {journal} {Nat. Commun.}\ }\textbf {\bibinfo {volume} {12}},\ \bibinfo {pages} {5681} (\bibinfo {year} {2021})}\BibitemShut {NoStop}%
\bibitem [{\citenamefont {Kato}(1995)}]{Kato1995Perturbation}%
  \BibitemOpen
  \bibfield  {author} {\bibinfo {author} {\bibfnamefont {T.}~\bibnamefont {Kato}},\ }\href {https://doi.org/10.1007/978-3-642-66282-9} {\emph {\bibinfo {title} {Perturbation {{Theory}} for {{Linear Operators}}}}},\ \bibinfo {edition} {2nd}\ ed.,\ Classics in {{Mathematics}}\ (\bibinfo  {publisher} {Springer-Verlag},\ \bibinfo {address} {Berlin Heidelberg},\ \bibinfo {year} {1995})\BibitemShut {NoStop}%
\bibitem [{\citenamefont {Heiss}(2012)}]{Heiss2012Physics}%
  \BibitemOpen
  \bibfield  {author} {\bibinfo {author} {\bibfnamefont {W.~D.}\ \bibnamefont {Heiss}},\ }\bibfield  {title} {\bibinfo {title} {The physics of exceptional points},\ }\href {https://doi.org/10.1088/1751-8113/45/44/444016} {\bibfield  {journal} {\bibinfo  {journal} {J. Phys. A: Math. Theor.}\ }\textbf {\bibinfo {volume} {45}},\ \bibinfo {pages} {444016} (\bibinfo {year} {2012})}\BibitemShut {NoStop}%
\bibitem [{\citenamefont {Golub}\ and\ \citenamefont {Van~Loan}(2013)}]{Golub2013Matrix}%
  \BibitemOpen
  \bibfield  {author} {\bibinfo {author} {\bibfnamefont {G.~H.}\ \bibnamefont {Golub}}\ and\ \bibinfo {author} {\bibfnamefont {C.~F.}\ \bibnamefont {Van~Loan}},\ }\href {https://jhupbooks.press.jhu.edu/title/matrix-computations} {\emph {\bibinfo {title} {Matrix {{Computations}}}}},\ \bibinfo {edition} {4th}\ ed.\ (\bibinfo  {publisher} {The Johns Hopkins University Press},\ \bibinfo {address} {Baltimore},\ \bibinfo {year} {2013})\BibitemShut {NoStop}%
\bibitem [{\citenamefont {Wiersig}(2014)}]{Wiersig2014Enhancing}%
  \BibitemOpen
  \bibfield  {author} {\bibinfo {author} {\bibfnamefont {J.}~\bibnamefont {Wiersig}},\ }\bibfield  {title} {\bibinfo {title} {Enhancing the {{Sensitivity}} of {{Frequency}} and {{Energy Splitting Detection}} by {{Using Exceptional Points}}: {{Application}} to {{Microcavity Sensors}} for {{Single-Particle Detection}}},\ }\href {https://doi.org/10.1103/PhysRevLett.112.203901} {\bibfield  {journal} {\bibinfo  {journal} {Phys. Rev. Lett.}\ }\textbf {\bibinfo {volume} {112}},\ \bibinfo {pages} {203901} (\bibinfo {year} {2014})}\BibitemShut {NoStop}%
\bibitem [{\citenamefont {Chen}\ \emph {et~al.}(2017)\citenamefont {Chen}, \citenamefont {Kaya~{\"O}zdemir}, \citenamefont {Zhao}, \citenamefont {Wiersig},\ and\ \citenamefont {Yang}}]{Chen2017Exceptional}%
  \BibitemOpen
  \bibfield  {author} {\bibinfo {author} {\bibfnamefont {W.}~\bibnamefont {Chen}}, \bibinfo {author} {\bibfnamefont {{\c S}.}~\bibnamefont {Kaya~{\"O}zdemir}}, \bibinfo {author} {\bibfnamefont {G.}~\bibnamefont {Zhao}}, \bibinfo {author} {\bibfnamefont {J.}~\bibnamefont {Wiersig}},\ and\ \bibinfo {author} {\bibfnamefont {L.}~\bibnamefont {Yang}},\ }\bibfield  {title} {\bibinfo {title} {Exceptional points enhance sensing in an optical microcavity},\ }\href {https://doi.org/10.1038/nature23281} {\bibfield  {journal} {\bibinfo  {journal} {Nature}\ }\textbf {\bibinfo {volume} {548}},\ \bibinfo {pages} {192} (\bibinfo {year} {2017})}\BibitemShut {NoStop}%
\bibitem [{\citenamefont {Hodaei}\ \emph {et~al.}(2017)\citenamefont {Hodaei}, \citenamefont {Hassan}, \citenamefont {Wittek}, \citenamefont {{Garcia-Gracia}}, \citenamefont {{El-Ganainy}}, \citenamefont {Christodoulides},\ and\ \citenamefont {Khajavikhan}}]{Hodaei2017Enhanced}%
  \BibitemOpen
  \bibfield  {author} {\bibinfo {author} {\bibfnamefont {H.}~\bibnamefont {Hodaei}}, \bibinfo {author} {\bibfnamefont {A.~U.}\ \bibnamefont {Hassan}}, \bibinfo {author} {\bibfnamefont {S.}~\bibnamefont {Wittek}}, \bibinfo {author} {\bibfnamefont {H.}~\bibnamefont {{Garcia-Gracia}}}, \bibinfo {author} {\bibfnamefont {R.}~\bibnamefont {{El-Ganainy}}}, \bibinfo {author} {\bibfnamefont {D.~N.}\ \bibnamefont {Christodoulides}},\ and\ \bibinfo {author} {\bibfnamefont {M.}~\bibnamefont {Khajavikhan}},\ }\bibfield  {title} {\bibinfo {title} {Enhanced sensitivity at higher-order exceptional points},\ }\href {https://doi.org/10.1038/nature23280} {\bibfield  {journal} {\bibinfo  {journal} {Nature}\ }\textbf {\bibinfo {volume} {548}},\ \bibinfo {pages} {187} (\bibinfo {year} {2017})}\BibitemShut {NoStop}%
\bibitem [{\citenamefont {Lau}\ and\ \citenamefont {Clerk}(2018)}]{Lau2018Fundamental}%
  \BibitemOpen
  \bibfield  {author} {\bibinfo {author} {\bibfnamefont {H.-K.}\ \bibnamefont {Lau}}\ and\ \bibinfo {author} {\bibfnamefont {A.~A.}\ \bibnamefont {Clerk}},\ }\bibfield  {title} {\bibinfo {title} {Fundamental limits and non-reciprocal approaches in non-{{Hermitian}} quantum sensing},\ }\href {https://doi.org/10.1038/s41467-018-06477-7} {\bibfield  {journal} {\bibinfo  {journal} {Nat. Commun.}\ }\textbf {\bibinfo {volume} {9}},\ \bibinfo {pages} {4320} (\bibinfo {year} {2018})}\BibitemShut {NoStop}%
\bibitem [{\citenamefont {Lee}\ \emph {et~al.}(2014)\citenamefont {Lee}, \citenamefont {Reiter},\ and\ \citenamefont {Moiseyev}}]{Lee2014Entanglement}%
  \BibitemOpen
  \bibfield  {author} {\bibinfo {author} {\bibfnamefont {T.~E.}\ \bibnamefont {Lee}}, \bibinfo {author} {\bibfnamefont {F.}~\bibnamefont {Reiter}},\ and\ \bibinfo {author} {\bibfnamefont {N.}~\bibnamefont {Moiseyev}},\ }\bibfield  {title} {\bibinfo {title} {Entanglement and {{Spin Squeezing}} in {{Non-Hermitian Phase Transitions}}},\ }\href {https://doi.org/10.1103/PhysRevLett.113.250401} {\bibfield  {journal} {\bibinfo  {journal} {Phys. Rev. Lett.}\ }\textbf {\bibinfo {volume} {113}},\ \bibinfo {pages} {250401} (\bibinfo {year} {2014})}\BibitemShut {NoStop}%
\bibitem [{\citenamefont {{San-Jose}}\ \emph {et~al.}(2016)\citenamefont {{San-Jose}}, \citenamefont {Cayao}, \citenamefont {Prada},\ and\ \citenamefont {Aguado}}]{San-Jose2016Majorana}%
  \BibitemOpen
  \bibfield  {author} {\bibinfo {author} {\bibfnamefont {P.}~\bibnamefont {{San-Jose}}}, \bibinfo {author} {\bibfnamefont {J.}~\bibnamefont {Cayao}}, \bibinfo {author} {\bibfnamefont {E.}~\bibnamefont {Prada}},\ and\ \bibinfo {author} {\bibfnamefont {R.}~\bibnamefont {Aguado}},\ }\bibfield  {title} {\bibinfo {title} {Majorana bound states from exceptional points in non-topological superconductors},\ }\href {https://doi.org/10.1038/srep21427} {\bibfield  {journal} {\bibinfo  {journal} {Sci. Rep.}\ }\textbf {\bibinfo {volume} {6}},\ \bibinfo {pages} {21427} (\bibinfo {year} {2016})}\BibitemShut {NoStop}%
\bibitem [{\citenamefont {Ashida}\ \emph {et~al.}(2017)\citenamefont {Ashida}, \citenamefont {Furukawa},\ and\ \citenamefont {Ueda}}]{Ashida2017Paritytimesymmetric}%
  \BibitemOpen
  \bibfield  {author} {\bibinfo {author} {\bibfnamefont {Y.}~\bibnamefont {Ashida}}, \bibinfo {author} {\bibfnamefont {S.}~\bibnamefont {Furukawa}},\ and\ \bibinfo {author} {\bibfnamefont {M.}~\bibnamefont {Ueda}},\ }\bibfield  {title} {\bibinfo {title} {Parity-time-symmetric quantum critical phenomena},\ }\href {https://doi.org/10.1038/ncomms15791} {\bibfield  {journal} {\bibinfo  {journal} {Nat. Commun.}\ }\textbf {\bibinfo {volume} {8}},\ \bibinfo {pages} {15791} (\bibinfo {year} {2017})}\BibitemShut {NoStop}%
\bibitem [{\citenamefont {Tserkovnyak}(2020)}]{Tserkovnyak2020Exceptional}%
  \BibitemOpen
  \bibfield  {author} {\bibinfo {author} {\bibfnamefont {Y.}~\bibnamefont {Tserkovnyak}},\ }\bibfield  {title} {\bibinfo {title} {Exceptional points in dissipatively coupled spin dynamics},\ }\href {https://doi.org/10.1103/PhysRevResearch.2.013031} {\bibfield  {journal} {\bibinfo  {journal} {Phys. Rev. Res.}\ }\textbf {\bibinfo {volume} {2}},\ \bibinfo {pages} {013031} (\bibinfo {year} {2020})}\BibitemShut {NoStop}%
\bibitem [{\citenamefont {Kohler}\ \emph {et~al.}(2005)\citenamefont {Kohler}, \citenamefont {Lehmann},\ and\ \citenamefont {H{\"a}nggi}}]{Kohler2005Driven}%
  \BibitemOpen
  \bibfield  {author} {\bibinfo {author} {\bibfnamefont {S.}~\bibnamefont {Kohler}}, \bibinfo {author} {\bibfnamefont {J.}~\bibnamefont {Lehmann}},\ and\ \bibinfo {author} {\bibfnamefont {P.}~\bibnamefont {H{\"a}nggi}},\ }\bibfield  {title} {\bibinfo {title} {Driven quantum transport on the nanoscale},\ }\href {https://doi.org/10.1016/j.physrep.2004.11.002} {\bibfield  {journal} {\bibinfo  {journal} {Phys. Rep.}\ }\textbf {\bibinfo {volume} {406}},\ \bibinfo {pages} {379} (\bibinfo {year} {2005})}\BibitemShut {NoStop}%
\bibitem [{\citenamefont {Datta}(2005)}]{Datta2005Quantum}%
  \BibitemOpen
  \bibfield  {author} {\bibinfo {author} {\bibfnamefont {S.}~\bibnamefont {Datta}},\ }\href {https://books.google.com/books?id=Yj50EJoS224C} {\emph {\bibinfo {title} {Quantum {{Transport}}: {{Atom}} to {{Transistor}}}}}\ (\bibinfo  {publisher} {Cambridge University Press},\ \bibinfo {address} {Cambridge, England},\ \bibinfo {year} {2005})\BibitemShut {NoStop}%
\bibitem [{\citenamefont {Landi}\ \emph {et~al.}(2022)\citenamefont {Landi}, \citenamefont {Poletti},\ and\ \citenamefont {Schaller}}]{Landi2022Nonequilibrium}%
  \BibitemOpen
  \bibfield  {author} {\bibinfo {author} {\bibfnamefont {G.~T.}\ \bibnamefont {Landi}}, \bibinfo {author} {\bibfnamefont {D.}~\bibnamefont {Poletti}},\ and\ \bibinfo {author} {\bibfnamefont {G.}~\bibnamefont {Schaller}},\ }\bibfield  {title} {\bibinfo {title} {Nonequilibrium boundary-driven quantum systems: {{Models}}, methods, and properties},\ }\href {https://doi.org/10.1103/RevModPhys.94.045006} {\bibfield  {journal} {\bibinfo  {journal} {Rev. Mod. Phys.}\ }\textbf {\bibinfo {volume} {94}},\ \bibinfo {pages} {045006} (\bibinfo {year} {2022})}\BibitemShut {NoStop}%
\bibitem [{\citenamefont {Gorini}\ \emph {et~al.}(1976)\citenamefont {Gorini}, \citenamefont {Kossakowski},\ and\ \citenamefont {Sudarshan}}]{Gorini1976Completely}%
  \BibitemOpen
  \bibfield  {author} {\bibinfo {author} {\bibfnamefont {V.}~\bibnamefont {Gorini}}, \bibinfo {author} {\bibfnamefont {A.}~\bibnamefont {Kossakowski}},\ and\ \bibinfo {author} {\bibfnamefont {E.~C.~G.}\ \bibnamefont {Sudarshan}},\ }\bibfield  {title} {\bibinfo {title} {Completely positive dynamical semigroups of {{N}}-level systems},\ }\href {https://doi.org/10.1063/1.522979} {\bibfield  {journal} {\bibinfo  {journal} {J. Math. Phys.}\ }\textbf {\bibinfo {volume} {17}},\ \bibinfo {pages} {821} (\bibinfo {year} {1976})}\BibitemShut {NoStop}%
\bibitem [{\citenamefont {Lindblad}(1976)}]{Lindblad1976Generators}%
  \BibitemOpen
  \bibfield  {author} {\bibinfo {author} {\bibfnamefont {G.}~\bibnamefont {Lindblad}},\ }\bibfield  {title} {\bibinfo {title} {On the generators of quantum dynamical semigroups},\ }\href {https://doi.org/10.1007/BF01608499} {\bibfield  {journal} {\bibinfo  {journal} {Commun.Math. Phys.}\ }\textbf {\bibinfo {volume} {48}},\ \bibinfo {pages} {119} (\bibinfo {year} {1976})}\BibitemShut {NoStop}%
\bibitem [{\citenamefont {Prosen}(2008)}]{Prosen2008Third}%
  \BibitemOpen
  \bibfield  {author} {\bibinfo {author} {\bibfnamefont {T.}~\bibnamefont {Prosen}},\ }\bibfield  {title} {\bibinfo {title} {Third quantization: A general method to solve master equations for quadratic open {{Fermi}} systems},\ }\href {https://doi.org/10.1088/1367-2630/10/4/043026} {\bibfield  {journal} {\bibinfo  {journal} {New J. Phys.}\ }\textbf {\bibinfo {volume} {10}},\ \bibinfo {pages} {043026} (\bibinfo {year} {2008})}\BibitemShut {NoStop}%
\bibitem [{\citenamefont {McDonald}\ and\ \citenamefont {Clerk}(2023)}]{McDonald2023Third}%
  \BibitemOpen
  \bibfield  {author} {\bibinfo {author} {\bibfnamefont {A.}~\bibnamefont {McDonald}}\ and\ \bibinfo {author} {\bibfnamefont {A.~A.}\ \bibnamefont {Clerk}},\ }\bibfield  {title} {\bibinfo {title} {Third quantization of open quantum systems: {{Dissipative}} symmetries and connections to phase-space and {{Keldysh}} field-theory formulations},\ }\href {https://doi.org/10.1103/PhysRevResearch.5.033107} {\bibfield  {journal} {\bibinfo  {journal} {Phys. Rev. Res.}\ }\textbf {\bibinfo {volume} {5}},\ \bibinfo {pages} {033107} (\bibinfo {year} {2023})}\BibitemShut {NoStop}%
\bibitem [{\citenamefont {Minganti}\ \emph {et~al.}(2019)\citenamefont {Minganti}, \citenamefont {Miranowicz}, \citenamefont {Chhajlany},\ and\ \citenamefont {Nori}}]{Minganti2019Quantum}%
  \BibitemOpen
  \bibfield  {author} {\bibinfo {author} {\bibfnamefont {F.}~\bibnamefont {Minganti}}, \bibinfo {author} {\bibfnamefont {A.}~\bibnamefont {Miranowicz}}, \bibinfo {author} {\bibfnamefont {R.~W.}\ \bibnamefont {Chhajlany}},\ and\ \bibinfo {author} {\bibfnamefont {F.}~\bibnamefont {Nori}},\ }\bibfield  {title} {\bibinfo {title} {Quantum exceptional points of non-{{Hermitian Hamiltonians}} and {{Liouvillians}}: {{The}} effects of quantum jumps},\ }\href {https://doi.org/10.1103/PhysRevA.100.062131} {\bibfield  {journal} {\bibinfo  {journal} {Phys. Rev. A}\ }\textbf {\bibinfo {volume} {100}},\ \bibinfo {pages} {062131} (\bibinfo {year} {2019})}\BibitemShut {NoStop}%
\bibitem [{\citenamefont {McDonald}\ \emph {et~al.}(2022)\citenamefont {McDonald}, \citenamefont {Hanai},\ and\ \citenamefont {Clerk}}]{McDonald2022Nonequilibrium}%
  \BibitemOpen
  \bibfield  {author} {\bibinfo {author} {\bibfnamefont {A.}~\bibnamefont {McDonald}}, \bibinfo {author} {\bibfnamefont {R.}~\bibnamefont {Hanai}},\ and\ \bibinfo {author} {\bibfnamefont {A.~A.}\ \bibnamefont {Clerk}},\ }\bibfield  {title} {\bibinfo {title} {Nonequilibrium stationary states of quantum non-{{Hermitian}} lattice models},\ }\href {https://doi.org/10.1103/PhysRevB.105.064302} {\bibfield  {journal} {\bibinfo  {journal} {Phys. Rev. B}\ }\textbf {\bibinfo {volume} {105}},\ \bibinfo {pages} {064302} (\bibinfo {year} {2022})}\BibitemShut {NoStop}%
\bibitem [{\citenamefont {Song}\ \emph {et~al.}(2019)\citenamefont {Song}, \citenamefont {Yao},\ and\ \citenamefont {Wang}}]{Song2019NonHermitian}%
  \BibitemOpen
  \bibfield  {author} {\bibinfo {author} {\bibfnamefont {F.}~\bibnamefont {Song}}, \bibinfo {author} {\bibfnamefont {S.}~\bibnamefont {Yao}},\ and\ \bibinfo {author} {\bibfnamefont {Z.}~\bibnamefont {Wang}},\ }\bibfield  {title} {\bibinfo {title} {Non-{{Hermitian Skin Effect}} and {{Chiral Damping}} in {{Open Quantum Systems}}},\ }\href {https://doi.org/10.1103/PhysRevLett.123.170401} {\bibfield  {journal} {\bibinfo  {journal} {Phys. Rev. Lett.}\ }\textbf {\bibinfo {volume} {123}},\ \bibinfo {pages} {170401} (\bibinfo {year} {2019})}\BibitemShut {NoStop}%
\bibitem [{\citenamefont {Rickayzen}(1980)}]{Rickayzen1980Green}%
  \BibitemOpen
  \bibfield  {author} {\bibinfo {author} {\bibfnamefont {G.}~\bibnamefont {Rickayzen}},\ }\href@noop {} {\emph {\bibinfo {title} {Green's {{Functions}} and {{Condensed Matter}}}}}\ (\bibinfo  {publisher} {Academic Press},\ \bibinfo {address} {New York},\ \bibinfo {year} {1980})\BibitemShut {NoStop}%
\bibitem [{\citenamefont {Economou}(2006)}]{Economou2006Green}%
  \BibitemOpen
  \bibfield  {author} {\bibinfo {author} {\bibfnamefont {E.~N.}\ \bibnamefont {Economou}},\ }\href {https://link.springer.com/book/10.1007/3-540-28841-4} {\emph {\bibinfo {title} {Green's {{Functions}} in {{Quantum Physics}}}}},\ \bibinfo {edition} {3rd}\ ed.\ (\bibinfo  {publisher} {Springer},\ \bibinfo {address} {New York},\ \bibinfo {year} {2006})\BibitemShut {NoStop}%
\bibitem [{\citenamefont {Wang}(2012)}]{Wang2012Green}%
  \BibitemOpen
  \bibfield  {author} {\bibinfo {author} {\bibfnamefont {H.}~\bibnamefont {Wang}},\ }\href {https://www.amazon.com/Greens-Function-Condensed-Matter-Physics/dp/1842657143} {\emph {\bibinfo {title} {Green's {{Function}} in {{Condensed Matter Physics}}}}},\ \bibinfo {edition} {1st}\ ed.\ (\bibinfo  {publisher} {Alpha Science International},\ \bibinfo {address} {Oxford, UK},\ \bibinfo {year} {2012})\BibitemShut {NoStop}%
\bibitem [{\citenamefont {Odashima}\ \emph {et~al.}(2016)\citenamefont {Odashima}, \citenamefont {Prado},\ and\ \citenamefont {Vernek}}]{Odashima2016Pedagogical}%
  \BibitemOpen
  \bibfield  {author} {\bibinfo {author} {\bibfnamefont {M.~M.}\ \bibnamefont {Odashima}}, \bibinfo {author} {\bibfnamefont {B.~G.}\ \bibnamefont {Prado}},\ and\ \bibinfo {author} {\bibfnamefont {E.}~\bibnamefont {Vernek}},\ }\bibfield  {title} {\bibinfo {title} {Pedagogical introduction to equilibrium {{Green}}'s functions: Condensed-matter examples with numerical implementations},\ }\href {https://doi.org/10.1590/1806-9126-RBEF-2016-0087} {\bibfield  {journal} {\bibinfo  {journal} {Rev. Bras. Ensino F{\'i}s.}\ }\textbf {\bibinfo {volume} {39}},\ \bibinfo {pages} {1} (\bibinfo {year} {2016})}\BibitemShut {NoStop}%
\bibitem [{\citenamefont {Shen}\ and\ \citenamefont {Fu}(2018)}]{Shen2018Quantum}%
  \BibitemOpen
  \bibfield  {author} {\bibinfo {author} {\bibfnamefont {H.}~\bibnamefont {Shen}}\ and\ \bibinfo {author} {\bibfnamefont {L.}~\bibnamefont {Fu}},\ }\bibfield  {title} {\bibinfo {title} {Quantum {{Oscillation}} from {{In-Gap States}} and a {{Non-Hermitian Landau Level Problem}}},\ }\href {https://doi.org/10.1103/PhysRevLett.121.026403} {\bibfield  {journal} {\bibinfo  {journal} {Phys. Rev. Lett.}\ }\textbf {\bibinfo {volume} {121}},\ \bibinfo {pages} {026403} (\bibinfo {year} {2018})}\BibitemShut {NoStop}%
\bibitem [{\citenamefont {Kozii}\ and\ \citenamefont {Fu}(2024)}]{Kozii2024NonHermitian}%
  \BibitemOpen
  \bibfield  {author} {\bibinfo {author} {\bibfnamefont {V.}~\bibnamefont {Kozii}}\ and\ \bibinfo {author} {\bibfnamefont {L.}~\bibnamefont {Fu}},\ }\bibfield  {title} {\bibinfo {title} {Non-{{Hermitian}} topological theory of finite-lifetime quasiparticles: {{Prediction}} of bulk {{Fermi}} arc due to exceptional point},\ }\href {https://doi.org/10.1103/PhysRevB.109.235139} {\bibfield  {journal} {\bibinfo  {journal} {Phys. Rev. B}\ }\textbf {\bibinfo {volume} {109}},\ \bibinfo {pages} {235139} (\bibinfo {year} {2024})}\BibitemShut {NoStop}%
\bibitem [{\citenamefont {Meden}\ \emph {et~al.}(2023)\citenamefont {Meden}, \citenamefont {Grunwald},\ and\ \citenamefont {Kennes}}]{Meden2023Mathcal}%
  \BibitemOpen
  \bibfield  {author} {\bibinfo {author} {\bibfnamefont {V.}~\bibnamefont {Meden}}, \bibinfo {author} {\bibfnamefont {L.}~\bibnamefont {Grunwald}},\ and\ \bibinfo {author} {\bibfnamefont {D.~M.}\ \bibnamefont {Kennes}},\ }\bibfield  {title} {\bibinfo {title} {{$\mathcal{PT}$}-symmetric, non-{{Hermitian}} quantum many-body physics---a methodological perspective},\ }\href {https://doi.org/10.1088/1361-6633/ad05f3} {\bibfield  {journal} {\bibinfo  {journal} {Rep. Prog. Phys.}\ }\textbf {\bibinfo {volume} {86}},\ \bibinfo {pages} {124501} (\bibinfo {year} {2023})}\BibitemShut {NoStop}%
\bibitem [{\citenamefont {Cayao}\ and\ \citenamefont {Sato}(2023)}]{Cayao2023NonHermitian}%
  \BibitemOpen
  \bibfield  {author} {\bibinfo {author} {\bibfnamefont {J.}~\bibnamefont {Cayao}}\ and\ \bibinfo {author} {\bibfnamefont {M.}~\bibnamefont {Sato}},\ }\bibfield  {title} {\bibinfo {title} {Non-{{Hermitian}} phase-biased {{Josephson}} junctions},\ }\Eprint {https://arxiv.org/abs/2307.15472v1} {arXiv:2307.15472v1}  (\bibinfo {year} {2023})\BibitemShut {NoStop}%
\bibitem [{\citenamefont {Li}\ \emph {et~al.}(2024{\natexlab{b}})\citenamefont {Li}, \citenamefont {Sun},\ and\ \citenamefont {Trauzettel}}]{Li2024Anomalous}%
  \BibitemOpen
  \bibfield  {author} {\bibinfo {author} {\bibfnamefont {C.-A.}\ \bibnamefont {Li}}, \bibinfo {author} {\bibfnamefont {H.-P.}\ \bibnamefont {Sun}},\ and\ \bibinfo {author} {\bibfnamefont {B.}~\bibnamefont {Trauzettel}},\ }\bibfield  {title} {\bibinfo {title} {Anomalous {{Andreev Spectrum}} and {{Transport}} in {{Non-Hermitian Josephson Junctions}}},\ }\Eprint {https://arxiv.org/abs/2307.04789v3} {arXiv:2307.04789v3}  (\bibinfo {year} {2024}{\natexlab{b}})\BibitemShut {NoStop}%
\bibitem [{\citenamefont {Kornich}(2023)}]{Kornich2023CurrentVoltage}%
  \BibitemOpen
  \bibfield  {author} {\bibinfo {author} {\bibfnamefont {V.}~\bibnamefont {Kornich}},\ }\bibfield  {title} {\bibinfo {title} {Current-{{Voltage Characteristics}} of the {{Normal Metal-Insulator-PT-Symmetric Non-Hermitian Superconductor Junction}} as a {{Probe}} of {{Non-Hermitian Formalisms}}},\ }\href {https://doi.org/10.1103/PhysRevLett.131.116001} {\bibfield  {journal} {\bibinfo  {journal} {Phys. Rev. Lett.}\ }\textbf {\bibinfo {volume} {131}},\ \bibinfo {pages} {116001} (\bibinfo {year} {2023})}\BibitemShut {NoStop}%
\bibitem [{Note1()}]{Note1}%
  \BibitemOpen
  \bibinfo {note} {We adopt the LR/RR-basis definition in Ref.~\cite {Kornich2023CurrentVoltage}, where the electron density is anomalous at EPs, and the PT-symmetric Hamiltonian is distinct from $\protect \mathcal {H}_{\protect \rm eff}$ studied in this Letter within the context of open quantum systems \cite {Meden2023Mathcal}.}\BibitemShut {Stop}%
\bibitem [{Note0()}]{Note0}%
  \BibitemOpen
  \bibinfo {note} {See Supplemental Material at [URL will be inserted by publisher] for: (i) Green's function and self-energy, (ii) properties of the NH Fermi-Dirac distribution, (iii) derivation of persistent current formula and its continuity at EPs, and (iv) extended numerical results on low-energy spectra, persistent currents, and current susceptibility, which includes Ref.~\cite {Shen2023Majoranamagnon}.}\BibitemShut {Stop}%
\bibitem [{\citenamefont {Shen}\ \emph {et~al.}(2023)\citenamefont {Shen}, \citenamefont {Perrin}, \citenamefont {Trif},\ and\ \citenamefont {Simon}}]{Shen2023Majoranamagnon}%
  \BibitemOpen
  \bibfield  {author} {\bibinfo {author} {\bibfnamefont {P.-X.}\ \bibnamefont {Shen}}, \bibinfo {author} {\bibfnamefont {V.}~\bibnamefont {Perrin}}, \bibinfo {author} {\bibfnamefont {M.}~\bibnamefont {Trif}},\ and\ \bibinfo {author} {\bibfnamefont {P.}~\bibnamefont {Simon}},\ }\bibfield  {title} {\bibinfo {title} {Majorana-magnon interactions in topological {{Shiba}} chains},\ }\href {https://doi.org/10.1103/PhysRevResearch.5.033207} {\bibfield  {journal} {\bibinfo  {journal} {Phys. Rev. Res.}\ }\textbf {\bibinfo {volume} {5}},\ \bibinfo {pages} {033207} (\bibinfo {year} {2023})}\BibitemShut {NoStop}%
\bibitem [{\citenamefont {Beenakker}\ and\ \citenamefont {{van Houten}}(1992)}]{Beenakker1992Superconducting}%
  \BibitemOpen
  \bibfield  {author} {\bibinfo {author} {\bibfnamefont {C.~W.~J.}\ \bibnamefont {Beenakker}}\ and\ \bibinfo {author} {\bibfnamefont {H.}~\bibnamefont {{van Houten}}},\ }\bibfield  {title} {\bibinfo {title} {The {{Superconducting Quantum Point Contact}}},\ }in\ \href {https://doi.org/10.1016/B978-0-12-409660-8.50051-1} {\emph {\bibinfo {booktitle} {Nanostructures and {{Mesoscopic Systems}}}}},\ \bibinfo {editor} {edited by\ \bibinfo {editor} {\bibfnamefont {W.~P.}\ \bibnamefont {Kirk}}\ and\ \bibinfo {editor} {\bibfnamefont {M.~A.}\ \bibnamefont {Reed}}}\ (\bibinfo  {publisher} {Academic Press},\ \bibinfo {address} {New York},\ \bibinfo {year} {1992})\ pp.\ \bibinfo {pages} {481--497}\BibitemShut {NoStop}%
\bibitem [{\citenamefont {B{\"u}ttiker}(1985)}]{Buttiker1985Small}%
  \BibitemOpen
  \bibfield  {author} {\bibinfo {author} {\bibfnamefont {M.}~\bibnamefont {B{\"u}ttiker}},\ }\bibfield  {title} {\bibinfo {title} {Small normal-metal loop coupled to an electron reservoir},\ }\href {https://doi.org/10.1103/PhysRevB.32.1846} {\bibfield  {journal} {\bibinfo  {journal} {Phys. Rev. B}\ }\textbf {\bibinfo {volume} {32}},\ \bibinfo {pages} {1846} (\bibinfo {year} {1985})}\BibitemShut {NoStop}%
\bibitem [{\citenamefont {Stern}\ \emph {et~al.}(1990)\citenamefont {Stern}, \citenamefont {Aharonov},\ and\ \citenamefont {Imry}}]{Stern1990Phase}%
  \BibitemOpen
  \bibfield  {author} {\bibinfo {author} {\bibfnamefont {A.}~\bibnamefont {Stern}}, \bibinfo {author} {\bibfnamefont {Y.}~\bibnamefont {Aharonov}},\ and\ \bibinfo {author} {\bibfnamefont {Y.}~\bibnamefont {Imry}},\ }\bibfield  {title} {\bibinfo {title} {Phase uncertainty and loss of interference: {{A}} general picture},\ }\href {https://doi.org/10.1103/PhysRevA.41.3436} {\bibfield  {journal} {\bibinfo  {journal} {Phys. Rev. A}\ }\textbf {\bibinfo {volume} {41}},\ \bibinfo {pages} {3436} (\bibinfo {year} {1990})}\BibitemShut {NoStop}%
\bibitem [{\citenamefont {Loss}\ and\ \citenamefont {Mullen}(1991)}]{Loss1991Dephasing}%
  \BibitemOpen
  \bibfield  {author} {\bibinfo {author} {\bibfnamefont {D.}~\bibnamefont {Loss}}\ and\ \bibinfo {author} {\bibfnamefont {K.}~\bibnamefont {Mullen}},\ }\bibfield  {title} {\bibinfo {title} {Dephasing by a dynamic asymmetric environment},\ }\href {https://doi.org/10.1103/PhysRevB.43.13252} {\bibfield  {journal} {\bibinfo  {journal} {Phys. Rev. B}\ }\textbf {\bibinfo {volume} {43}},\ \bibinfo {pages} {13252} (\bibinfo {year} {1991})}\BibitemShut {NoStop}%
\bibitem [{\citenamefont {Chang}\ and\ \citenamefont {Bagwell}(1997)}]{Chang1997Control}%
  \BibitemOpen
  \bibfield  {author} {\bibinfo {author} {\bibfnamefont {L.-F.}\ \bibnamefont {Chang}}\ and\ \bibinfo {author} {\bibfnamefont {P.~F.}\ \bibnamefont {Bagwell}},\ }\bibfield  {title} {\bibinfo {title} {Control of {{Andreev-level}} occupation in a {{Josephson}} junction by a normal-metal probe},\ }\href {https://doi.org/10.1103/PhysRevB.55.12678} {\bibfield  {journal} {\bibinfo  {journal} {Phys. Rev. B}\ }\textbf {\bibinfo {volume} {55}},\ \bibinfo {pages} {12678} (\bibinfo {year} {1997})}\BibitemShut {NoStop}%
\bibitem [{\citenamefont {Mortensen}\ \emph {et~al.}(2000)\citenamefont {Mortensen}, \citenamefont {Jauho},\ and\ \citenamefont {Flensberg}}]{Mortensen2000Dephasing}%
  \BibitemOpen
  \bibfield  {author} {\bibinfo {author} {\bibfnamefont {N.~A.}\ \bibnamefont {Mortensen}}, \bibinfo {author} {\bibfnamefont {A.-P.}\ \bibnamefont {Jauho}},\ and\ \bibinfo {author} {\bibfnamefont {K.}~\bibnamefont {Flensberg}},\ }\bibfield  {title} {\bibinfo {title} {Dephasing in semiconductor--superconductor structures by coupling to a voltage probe},\ }\href {https://doi.org/10.1006/spmi.2000.0890} {\bibfield  {journal} {\bibinfo  {journal} {Superlattices Microstruct.}\ }\textbf {\bibinfo {volume} {28}},\ \bibinfo {pages} {67} (\bibinfo {year} {2000})}\BibitemShut {NoStop}%
\bibitem [{\citenamefont {Gramespacher}\ and\ \citenamefont {B{\"u}ttiker}(2000)}]{Gramespacher2000Distribution}%
  \BibitemOpen
  \bibfield  {author} {\bibinfo {author} {\bibfnamefont {T.}~\bibnamefont {Gramespacher}}\ and\ \bibinfo {author} {\bibfnamefont {M.}~\bibnamefont {B{\"u}ttiker}},\ }\bibfield  {title} {\bibinfo {title} {Distribution functions and current-current correlations in normal-metal--superconductor heterostructures},\ }\href {https://doi.org/10.1103/PhysRevB.61.8125} {\bibfield  {journal} {\bibinfo  {journal} {Phys. Rev. B}\ }\textbf {\bibinfo {volume} {61}},\ \bibinfo {pages} {8125} (\bibinfo {year} {2000})}\BibitemShut {NoStop}%
\bibitem [{\citenamefont {Belogolovskii}\ \emph {et~al.}(2001)\citenamefont {Belogolovskii}, \citenamefont {Golubov}, \citenamefont {Grajcar}, \citenamefont {Kupriyanov},\ and\ \citenamefont {Seidel}}]{Belogolovskii2001Charge}%
  \BibitemOpen
  \bibfield  {author} {\bibinfo {author} {\bibfnamefont {M.}~\bibnamefont {Belogolovskii}}, \bibinfo {author} {\bibfnamefont {A.}~\bibnamefont {Golubov}}, \bibinfo {author} {\bibfnamefont {M.}~\bibnamefont {Grajcar}}, \bibinfo {author} {\bibfnamefont {M.~{\relax Yu}.}\ \bibnamefont {Kupriyanov}},\ and\ \bibinfo {author} {\bibfnamefont {P.}~\bibnamefont {Seidel}},\ }\bibfield  {title} {\bibinfo {title} {Charge transport across a mesoscopic superconductor--normal metal junction: Coherence and decoherence effects},\ }\href {https://doi.org/10.1016/S0921-4534(01)00559-7} {\bibfield  {journal} {\bibinfo  {journal} {Phys. C Supercond.}\ }\textbf {\bibinfo {volume} {357--360}},\ \bibinfo {pages} {1592} (\bibinfo {year} {2001})}\BibitemShut {NoStop}%
\bibitem [{\citenamefont {Belogolovskii}(2003)}]{Belogolovskii2003Phasebreaking}%
  \BibitemOpen
  \bibfield  {author} {\bibinfo {author} {\bibfnamefont {M.}~\bibnamefont {Belogolovskii}},\ }\bibfield  {title} {\bibinfo {title} {Phase-breaking effects in superconducting heterostructures},\ }\href {https://doi.org/10.1103/PhysRevB.67.100503} {\bibfield  {journal} {\bibinfo  {journal} {Phys. Rev. B}\ }\textbf {\bibinfo {volume} {67}},\ \bibinfo {pages} {100503} (\bibinfo {year} {2003})}\BibitemShut {NoStop}%
\bibitem [{\citenamefont {B{\'e}ri}(2009)}]{Beri2009Dephasingenabled}%
  \BibitemOpen
  \bibfield  {author} {\bibinfo {author} {\bibfnamefont {B.}~\bibnamefont {B{\'e}ri}},\ }\bibfield  {title} {\bibinfo {title} {Dephasing-enabled triplet {{Andreev}} conductance},\ }\href {https://doi.org/10.1103/PhysRevB.79.245315} {\bibfield  {journal} {\bibinfo  {journal} {Phys. Rev. B}\ }\textbf {\bibinfo {volume} {79}},\ \bibinfo {pages} {245315} (\bibinfo {year} {2009})}\BibitemShut {NoStop}%
\bibitem [{\citenamefont {Jauho}\ \emph {et~al.}(1994)\citenamefont {Jauho}, \citenamefont {Wingreen},\ and\ \citenamefont {Meir}}]{Jauho1994Timedependent}%
  \BibitemOpen
  \bibfield  {author} {\bibinfo {author} {\bibfnamefont {A.-P.}\ \bibnamefont {Jauho}}, \bibinfo {author} {\bibfnamefont {N.~S.}\ \bibnamefont {Wingreen}},\ and\ \bibinfo {author} {\bibfnamefont {Y.}~\bibnamefont {Meir}},\ }\bibfield  {title} {\bibinfo {title} {Time-dependent transport in interacting and noninteracting resonant-tunneling systems},\ }\href {https://doi.org/10.1103/PhysRevB.50.5528} {\bibfield  {journal} {\bibinfo  {journal} {Phys. Rev. B}\ }\textbf {\bibinfo {volume} {50}},\ \bibinfo {pages} {5528} (\bibinfo {year} {1994})}\BibitemShut {NoStop}%
\bibitem [{\citenamefont {Potts}\ \emph {et~al.}(2021)\citenamefont {Potts}, \citenamefont {Kalaee},\ and\ \citenamefont {Wacker}}]{Potts2021Thermodynamically}%
  \BibitemOpen
  \bibfield  {author} {\bibinfo {author} {\bibfnamefont {P.~P.}\ \bibnamefont {Potts}}, \bibinfo {author} {\bibfnamefont {A.~A.~S.}\ \bibnamefont {Kalaee}},\ and\ \bibinfo {author} {\bibfnamefont {A.}~\bibnamefont {Wacker}},\ }\bibfield  {title} {\bibinfo {title} {A thermodynamically consistent {{Markovian}} master equation beyond the secular approximation},\ }\href {https://doi.org/10.1088/1367-2630/ac3b2f} {\bibfield  {journal} {\bibinfo  {journal} {New J. Phys.}\ }\textbf {\bibinfo {volume} {23}},\ \bibinfo {pages} {123013} (\bibinfo {year} {2021})}\BibitemShut {NoStop}%
\bibitem [{\citenamefont {Yu}\ \emph {et~al.}(2024)\citenamefont {Yu}, \citenamefont {Zou}, \citenamefont {Zeng}, \citenamefont {Rao},\ and\ \citenamefont {Xia}}]{Yu2024NonHermitian}%
  \BibitemOpen
  \bibfield  {author} {\bibinfo {author} {\bibfnamefont {T.}~\bibnamefont {Yu}}, \bibinfo {author} {\bibfnamefont {J.}~\bibnamefont {Zou}}, \bibinfo {author} {\bibfnamefont {B.}~\bibnamefont {Zeng}}, \bibinfo {author} {\bibfnamefont {J.~W.}\ \bibnamefont {Rao}},\ and\ \bibinfo {author} {\bibfnamefont {K.}~\bibnamefont {Xia}},\ }\bibfield  {title} {\bibinfo {title} {Non-{{Hermitian}} topological magnonics},\ }\href {https://doi.org/10.1016/j.physrep.2024.01.006} {\bibfield  {journal} {\bibinfo  {journal} {Phys. Rep.}\ }\textbf {\bibinfo {volume} {1062}},\ \bibinfo {pages} {1} (\bibinfo {year} {2024})}\BibitemShut {NoStop}%
\bibitem [{\citenamefont {Brody}(2013)}]{Brody2013Biorthogonal}%
  \BibitemOpen
  \bibfield  {author} {\bibinfo {author} {\bibfnamefont {D.~C.}\ \bibnamefont {Brody}},\ }\bibfield  {title} {\bibinfo {title} {Biorthogonal quantum mechanics},\ }\href {https://doi.org/10.1088/1751-8113/47/3/035305} {\bibfield  {journal} {\bibinfo  {journal} {J. Phys. A: Math. Theor.}\ }\textbf {\bibinfo {volume} {47}},\ \bibinfo {pages} {035305} (\bibinfo {year} {2013})}\BibitemShut {NoStop}%
\bibitem [{\citenamefont {Chen}\ and\ \citenamefont {Zhai}(2018)}]{Chen2018Hall}%
  \BibitemOpen
  \bibfield  {author} {\bibinfo {author} {\bibfnamefont {Y.}~\bibnamefont {Chen}}\ and\ \bibinfo {author} {\bibfnamefont {H.}~\bibnamefont {Zhai}},\ }\bibfield  {title} {\bibinfo {title} {Hall conductance of a non-{{Hermitian Chern}} insulator},\ }\href {https://doi.org/10.1103/PhysRevB.98.245130} {\bibfield  {journal} {\bibinfo  {journal} {Phys. Rev. B}\ }\textbf {\bibinfo {volume} {98}},\ \bibinfo {pages} {245130} (\bibinfo {year} {2018})}\BibitemShut {NoStop}%
\bibitem [{\citenamefont {Yang}\ \emph {et~al.}(2021)\citenamefont {Yang}, \citenamefont {Yang}, \citenamefont {Hu},\ and\ \citenamefont {Liu}}]{Yang2021Dissipative}%
  \BibitemOpen
  \bibfield  {author} {\bibinfo {author} {\bibfnamefont {Z.}~\bibnamefont {Yang}}, \bibinfo {author} {\bibfnamefont {Q.}~\bibnamefont {Yang}}, \bibinfo {author} {\bibfnamefont {J.}~\bibnamefont {Hu}},\ and\ \bibinfo {author} {\bibfnamefont {D.~E.}\ \bibnamefont {Liu}},\ }\bibfield  {title} {\bibinfo {title} {Dissipative {{Floquet Majorana Modes}} in {{Proximity-Induced Topological Superconductors}}},\ }\href {https://doi.org/10.1103/PhysRevLett.126.086801} {\bibfield  {journal} {\bibinfo  {journal} {Phys. Rev. Lett.}\ }\textbf {\bibinfo {volume} {126}},\ \bibinfo {pages} {086801} (\bibinfo {year} {2021})}\BibitemShut {NoStop}%
\bibitem [{\citenamefont {Arouca}\ \emph {et~al.}(2023)\citenamefont {Arouca}, \citenamefont {Cayao},\ and\ \citenamefont {{Black-Schaffer}}}]{Arouca2023Topological}%
  \BibitemOpen
  \bibfield  {author} {\bibinfo {author} {\bibfnamefont {R.}~\bibnamefont {Arouca}}, \bibinfo {author} {\bibfnamefont {J.}~\bibnamefont {Cayao}},\ and\ \bibinfo {author} {\bibfnamefont {A.~M.}\ \bibnamefont {{Black-Schaffer}}},\ }\bibfield  {title} {\bibinfo {title} {Topological superconductivity enhanced by exceptional points},\ }\href {https://doi.org/10.1103/PhysRevB.108.L060506} {\bibfield  {journal} {\bibinfo  {journal} {Phys. Rev. B}\ }\textbf {\bibinfo {volume} {108}},\ \bibinfo {pages} {L060506} (\bibinfo {year} {2023})}\BibitemShut {NoStop}%
\bibitem [{Note2()}]{Note2}%
  \BibitemOpen
  \bibinfo {note} {The $\ln |\varepsilon _n|$ term stems from the principal value (PV) in the integrand, which is typically disregarded in the literature \cite {Rickayzen1980Green}, since $\mathinner {|{\psi ^\protect \textsc {l}_n}\rangle } \approx \mathinner {|{\psi ^\protect \textsc {r}_n}\rangle }$ in the weak coupling limit. However, PV plays a crucial role in correctly determining observables in EPs, which typically occur in the strong coupling regime.}\BibitemShut {Stop}%
\bibitem [{Note3()}]{Note3}%
  \BibitemOpen
  \bibinfo {note} {Observables are gauge-invariant under $f_{\protect \rm eff} \rightarrow f_{\protect \rm eff} + \protect \mathbb {R}$: All observables stay the same with $f_{\protect \rm eff} (\varepsilon ) = - (1/\pi )\ln (\varepsilon / C)$, $\forall C \in \protect \mathbb {R}$ \cite {Note0}. Hence, we set $C = 1$ for simplicity.}\BibitemShut {Stop}%
\bibitem [{Note4()}]{Note4}%
  \BibitemOpen
  \bibinfo {note} {Owing to the analytic property of $f_{\protect \rm eff}$, the EP is a removable singularity for physical observables, albeit a branch point for biorthogonal wavefunctions \cite {Note0}}\BibitemShut {NoStop}%
\bibitem [{\citenamefont {Kawabata}\ \emph {et~al.}(2023)\citenamefont {Kawabata}, \citenamefont {Numasawa},\ and\ \citenamefont {Ryu}}]{Kawabata2023Entanglement}%
  \BibitemOpen
  \bibfield  {author} {\bibinfo {author} {\bibfnamefont {K.}~\bibnamefont {Kawabata}}, \bibinfo {author} {\bibfnamefont {T.}~\bibnamefont {Numasawa}},\ and\ \bibinfo {author} {\bibfnamefont {S.}~\bibnamefont {Ryu}},\ }\bibfield  {title} {\bibinfo {title} {Entanglement {{Phase Transition Induced}} by the {{Non-Hermitian Skin Effect}}},\ }\href {https://doi.org/10.1103/PhysRevX.13.021007} {\bibfield  {journal} {\bibinfo  {journal} {Phys. Rev. X}\ }\textbf {\bibinfo {volume} {13}},\ \bibinfo {pages} {021007} (\bibinfo {year} {2023})}\BibitemShut {NoStop}%
\bibitem [{\citenamefont {Herviou}\ \emph {et~al.}(2019)\citenamefont {Herviou}, \citenamefont {Regnault},\ and\ \citenamefont {Bardarson}}]{Herviou2019Entanglement}%
  \BibitemOpen
  \bibfield  {author} {\bibinfo {author} {\bibfnamefont {L.}~\bibnamefont {Herviou}}, \bibinfo {author} {\bibfnamefont {N.}~\bibnamefont {Regnault}},\ and\ \bibinfo {author} {\bibfnamefont {J.~H.}\ \bibnamefont {Bardarson}},\ }\bibfield  {title} {\bibinfo {title} {Entanglement spectrum and symmetries in non-{{Hermitian}} fermionic non-interacting models},\ }\href {https://doi.org/10.21468/SciPostPhys.7.5.069} {\bibfield  {journal} {\bibinfo  {journal} {SciPost Phys.}\ }\textbf {\bibinfo {volume} {7}},\ \bibinfo {pages} {069} (\bibinfo {year} {2019})}\BibitemShut {NoStop}%
\bibitem [{\citenamefont {Naghiloo}\ \emph {et~al.}(2019)\citenamefont {Naghiloo}, \citenamefont {Abbasi}, \citenamefont {Joglekar},\ and\ \citenamefont {Murch}}]{Naghiloo2019Quantum}%
  \BibitemOpen
  \bibfield  {author} {\bibinfo {author} {\bibfnamefont {M.}~\bibnamefont {Naghiloo}}, \bibinfo {author} {\bibfnamefont {M.}~\bibnamefont {Abbasi}}, \bibinfo {author} {\bibfnamefont {Y.~N.}\ \bibnamefont {Joglekar}},\ and\ \bibinfo {author} {\bibfnamefont {K.~W.}\ \bibnamefont {Murch}},\ }\bibfield  {title} {\bibinfo {title} {Quantum state tomography across the exceptional point in a single dissipative qubit},\ }\href {https://doi.org/10.1038/s41567-019-0652-z} {\bibfield  {journal} {\bibinfo  {journal} {Nat. Phys.}\ }\textbf {\bibinfo {volume} {15}},\ \bibinfo {pages} {1232} (\bibinfo {year} {2019})}\BibitemShut {NoStop}%
\bibitem [{\citenamefont {B{\"u}ttiker}\ \emph {et~al.}(1983)\citenamefont {B{\"u}ttiker}, \citenamefont {Imry},\ and\ \citenamefont {Landauer}}]{Buttiker1983Josephson}%
  \BibitemOpen
  \bibfield  {author} {\bibinfo {author} {\bibfnamefont {M.}~\bibnamefont {B{\"u}ttiker}}, \bibinfo {author} {\bibfnamefont {Y.}~\bibnamefont {Imry}},\ and\ \bibinfo {author} {\bibfnamefont {R.}~\bibnamefont {Landauer}},\ }\bibfield  {title} {\bibinfo {title} {Josephson behavior in small normal one-dimensional rings},\ }\href {https://doi.org/10.1016/0375-9601(83)90011-7} {\bibfield  {journal} {\bibinfo  {journal} {Phys. Lett. A}\ }\textbf {\bibinfo {volume} {96}},\ \bibinfo {pages} {365} (\bibinfo {year} {1983})}\BibitemShut {NoStop}%
\bibitem [{\citenamefont {L{\'e}vy}\ \emph {et~al.}(1990)\citenamefont {L{\'e}vy}, \citenamefont {Dolan}, \citenamefont {Dunsmuir},\ and\ \citenamefont {Bouchiat}}]{Levy1990Magnetization}%
  \BibitemOpen
  \bibfield  {author} {\bibinfo {author} {\bibfnamefont {L.~P.}\ \bibnamefont {L{\'e}vy}}, \bibinfo {author} {\bibfnamefont {G.}~\bibnamefont {Dolan}}, \bibinfo {author} {\bibfnamefont {J.}~\bibnamefont {Dunsmuir}},\ and\ \bibinfo {author} {\bibfnamefont {H.}~\bibnamefont {Bouchiat}},\ }\bibfield  {title} {\bibinfo {title} {Magnetization of mesoscopic copper rings: {{Evidence}} for persistent currents},\ }\href {https://doi.org/10.1103/PhysRevLett.64.2074} {\bibfield  {journal} {\bibinfo  {journal} {Phys. Rev. Lett.}\ }\textbf {\bibinfo {volume} {64}},\ \bibinfo {pages} {2074} (\bibinfo {year} {1990})}\BibitemShut {NoStop}%
\bibitem [{\citenamefont {{Bary-Soroker}}\ \emph {et~al.}(2008)\citenamefont {{Bary-Soroker}}, \citenamefont {{Entin-Wohlman}},\ and\ \citenamefont {Imry}}]{Bary-Soroker2008Effect}%
  \BibitemOpen
  \bibfield  {author} {\bibinfo {author} {\bibfnamefont {H.}~\bibnamefont {{Bary-Soroker}}}, \bibinfo {author} {\bibfnamefont {O.}~\bibnamefont {{Entin-Wohlman}}},\ and\ \bibinfo {author} {\bibfnamefont {Y.}~\bibnamefont {Imry}},\ }\bibfield  {title} {\bibinfo {title} {Effect of {{Pair Breaking}} on {{Mesoscopic Persistent Currents Well}} above the {{Superconducting Transition Temperature}}},\ }\href {https://doi.org/10.1103/PhysRevLett.101.057001} {\bibfield  {journal} {\bibinfo  {journal} {Phys. Rev. Lett.}\ }\textbf {\bibinfo {volume} {101}},\ \bibinfo {pages} {057001} (\bibinfo {year} {2008})}\BibitemShut {NoStop}%
\bibitem [{\citenamefont {Bluhm}\ \emph {et~al.}(2009)\citenamefont {Bluhm}, \citenamefont {Koshnick}, \citenamefont {Bert}, \citenamefont {Huber},\ and\ \citenamefont {Moler}}]{Bluhm2009Persistent}%
  \BibitemOpen
  \bibfield  {author} {\bibinfo {author} {\bibfnamefont {H.}~\bibnamefont {Bluhm}}, \bibinfo {author} {\bibfnamefont {N.~C.}\ \bibnamefont {Koshnick}}, \bibinfo {author} {\bibfnamefont {J.~A.}\ \bibnamefont {Bert}}, \bibinfo {author} {\bibfnamefont {M.~E.}\ \bibnamefont {Huber}},\ and\ \bibinfo {author} {\bibfnamefont {K.~A.}\ \bibnamefont {Moler}},\ }\bibfield  {title} {\bibinfo {title} {Persistent {{Currents}} in {{Normal Metal Rings}}},\ }\href {https://doi.org/10.1103/PhysRevLett.102.136802} {\bibfield  {journal} {\bibinfo  {journal} {Phys. Rev. Lett.}\ }\textbf {\bibinfo {volume} {102}},\ \bibinfo {pages} {136802} (\bibinfo {year} {2009})}\BibitemShut {NoStop}%
\bibitem [{\citenamefont {{Bleszynski-Jayich}}\ \emph {et~al.}(2009)\citenamefont {{Bleszynski-Jayich}}, \citenamefont {Shanks}, \citenamefont {Peaudecerf}, \citenamefont {Ginossar}, \citenamefont {{von Oppen}}, \citenamefont {Glazman},\ and\ \citenamefont {Harris}}]{Bleszynski-Jayich2009Persistent}%
  \BibitemOpen
  \bibfield  {author} {\bibinfo {author} {\bibfnamefont {A.~C.}\ \bibnamefont {{Bleszynski-Jayich}}}, \bibinfo {author} {\bibfnamefont {W.~E.}\ \bibnamefont {Shanks}}, \bibinfo {author} {\bibfnamefont {B.}~\bibnamefont {Peaudecerf}}, \bibinfo {author} {\bibfnamefont {E.}~\bibnamefont {Ginossar}}, \bibinfo {author} {\bibfnamefont {F.}~\bibnamefont {{von Oppen}}}, \bibinfo {author} {\bibfnamefont {L.}~\bibnamefont {Glazman}},\ and\ \bibinfo {author} {\bibfnamefont {J.~G.~E.}\ \bibnamefont {Harris}},\ }\bibfield  {title} {\bibinfo {title} {Persistent {{Currents}} in {{Normal Metal Rings}}},\ }\href {https://doi.org/10.1126/science.1178139} {\bibfield  {journal} {\bibinfo  {journal} {Science}\ }\textbf {\bibinfo {volume} {326}},\ \bibinfo {pages} {272} (\bibinfo {year} {2009})}\BibitemShut {NoStop}%
\bibitem [{\citenamefont {Carini}\ \emph {et~al.}(1984)\citenamefont {Carini}, \citenamefont {Muttalib},\ and\ \citenamefont {Nagel}}]{Carini1984Origin}%
  \BibitemOpen
  \bibfield  {author} {\bibinfo {author} {\bibfnamefont {J.~P.}\ \bibnamefont {Carini}}, \bibinfo {author} {\bibfnamefont {K.~A.}\ \bibnamefont {Muttalib}},\ and\ \bibinfo {author} {\bibfnamefont {S.~R.}\ \bibnamefont {Nagel}},\ }\bibfield  {title} {\bibinfo {title} {Origin of the {{Bohm-Aharonov Effect}} with {{Half Flux Quanta}}},\ }\href {https://doi.org/10.1103/PhysRevLett.53.102} {\bibfield  {journal} {\bibinfo  {journal} {Phys. Rev. Lett.}\ }\textbf {\bibinfo {volume} {53}},\ \bibinfo {pages} {102} (\bibinfo {year} {1984})}\BibitemShut {NoStop}%
\bibitem [{\citenamefont {Browne}\ \emph {et~al.}(1984)\citenamefont {Browne}, \citenamefont {Carini}, \citenamefont {Muttalib},\ and\ \citenamefont {Nagel}}]{Browne1984Periodicity}%
  \BibitemOpen
  \bibfield  {author} {\bibinfo {author} {\bibfnamefont {D.~A.}\ \bibnamefont {Browne}}, \bibinfo {author} {\bibfnamefont {J.~P.}\ \bibnamefont {Carini}}, \bibinfo {author} {\bibfnamefont {K.~A.}\ \bibnamefont {Muttalib}},\ and\ \bibinfo {author} {\bibfnamefont {S.~R.}\ \bibnamefont {Nagel}},\ }\bibfield  {title} {\bibinfo {title} {Periodicity of transport coefficients with half flux quanta in the {{Aharonov-Bohm}} effect},\ }\href {https://doi.org/10.1103/PhysRevB.30.6798} {\bibfield  {journal} {\bibinfo  {journal} {Phys. Rev. B}\ }\textbf {\bibinfo {volume} {30}},\ \bibinfo {pages} {6798} (\bibinfo {year} {1984})}\BibitemShut {NoStop}%
\bibitem [{\citenamefont {Cheung}\ \emph {et~al.}(1988)\citenamefont {Cheung}, \citenamefont {Gefen}, \citenamefont {Riedel},\ and\ \citenamefont {Shih}}]{Cheung1988Persistent}%
  \BibitemOpen
  \bibfield  {author} {\bibinfo {author} {\bibfnamefont {H.-F.}\ \bibnamefont {Cheung}}, \bibinfo {author} {\bibfnamefont {Y.}~\bibnamefont {Gefen}}, \bibinfo {author} {\bibfnamefont {E.~K.}\ \bibnamefont {Riedel}},\ and\ \bibinfo {author} {\bibfnamefont {W.-H.}\ \bibnamefont {Shih}},\ }\bibfield  {title} {\bibinfo {title} {Persistent currents in small one-dimensional metal rings},\ }\href {https://doi.org/10.1103/PhysRevB.37.6050} {\bibfield  {journal} {\bibinfo  {journal} {Phys. Rev. B}\ }\textbf {\bibinfo {volume} {37}},\ \bibinfo {pages} {6050} (\bibinfo {year} {1988})}\BibitemShut {NoStop}%
\bibitem [{\citenamefont {Byers}\ and\ \citenamefont {Yang}(1961)}]{Byers1961Theoretical}%
  \BibitemOpen
  \bibfield  {author} {\bibinfo {author} {\bibfnamefont {N.}~\bibnamefont {Byers}}\ and\ \bibinfo {author} {\bibfnamefont {C.~N.}\ \bibnamefont {Yang}},\ }\bibfield  {title} {\bibinfo {title} {Theoretical {{Considerations Concerning Quantized Magnetic Flux}} in {{Superconducting Cylinders}}},\ }\href {https://doi.org/10.1103/PhysRevLett.7.46} {\bibfield  {journal} {\bibinfo  {journal} {Phys. Rev. Lett.}\ }\textbf {\bibinfo {volume} {7}},\ \bibinfo {pages} {46} (\bibinfo {year} {1961})}\BibitemShut {NoStop}%
\bibitem [{Note5()}]{Note5}%
  \BibitemOpen
  \bibinfo {note} {$\{ t_j \}_{j=1}^{6} \mkern -5mu = \mkern -5mu \{-0.859915, -0.884918, -0.918446, -0.846311,$ $-1.19937, \mkern -3mu -0.984676 \}$ for numerical results in the Letter.}\BibitemShut {Stop}%
\bibitem [{\citenamefont {Akkermans}\ \emph {et~al.}(1991)\citenamefont {Akkermans}, \citenamefont {Auerbach}, \citenamefont {Avron},\ and\ \citenamefont {Shapiro}}]{Akkermans1991Relation}%
  \BibitemOpen
  \bibfield  {author} {\bibinfo {author} {\bibfnamefont {E.}~\bibnamefont {Akkermans}}, \bibinfo {author} {\bibfnamefont {A.}~\bibnamefont {Auerbach}}, \bibinfo {author} {\bibfnamefont {J.~E.}\ \bibnamefont {Avron}},\ and\ \bibinfo {author} {\bibfnamefont {B.}~\bibnamefont {Shapiro}},\ }\bibfield  {title} {\bibinfo {title} {Relation between persistent currents and the scattering matrix},\ }\href {https://doi.org/10.1103/PhysRevLett.66.76} {\bibfield  {journal} {\bibinfo  {journal} {Phys. Rev. Lett.}\ }\textbf {\bibinfo {volume} {66}},\ \bibinfo {pages} {76} (\bibinfo {year} {1991})}\BibitemShut {NoStop}%
\bibitem [{\citenamefont {Landauer}\ and\ \citenamefont {B{\"u}ttiker}(1985)}]{Landauer1985Resistance}%
  \BibitemOpen
  \bibfield  {author} {\bibinfo {author} {\bibfnamefont {R.}~\bibnamefont {Landauer}}\ and\ \bibinfo {author} {\bibfnamefont {M.}~\bibnamefont {B{\"u}ttiker}},\ }\bibfield  {title} {\bibinfo {title} {Resistance of {{Small Metallic Loops}}},\ }\href {https://doi.org/10.1103/PhysRevLett.54.2049} {\bibfield  {journal} {\bibinfo  {journal} {Phys. Rev. Lett.}\ }\textbf {\bibinfo {volume} {54}},\ \bibinfo {pages} {2049} (\bibinfo {year} {1985})}\BibitemShut {NoStop}%
\bibitem [{Note6()}]{Note6}%
  \BibitemOpen
  \bibinfo {note} {To maintain a unified formalism across normal rings and SNS junctions, we adopt a constant $\Delta $ and refrain from self-consistent calculations of observables for SNS junctions.}\BibitemShut {Stop}%
\bibitem [{\citenamefont {Weisstein}(2024{\natexlab{a}})}]{Weisstein2024Digamma}%
  \BibitemOpen
  \bibfield  {author} {\bibinfo {author} {\bibfnamefont {E.~W.}\ \bibnamefont {Weisstein}},\ }\href {https://functions.wolfram.com/GammaBetaErf/PolyGamma/introductions/DifferentiatedGammas/ShowAll.html} {\bibinfo {title} {Digamma {{Function}}}} (\bibinfo {year} {2024}{\natexlab{a}})\BibitemShut {NoStop}%
\bibitem [{\citenamefont {Weisstein}(2024{\natexlab{b}})}]{Weisstein2024LogGamma}%
  \BibitemOpen
  \bibfield  {author} {\bibinfo {author} {\bibfnamefont {E.~W.}\ \bibnamefont {Weisstein}},\ }\href {https://functions.wolfram.com/GammaBetaErf/LogGamma/introductions/Gammas/ShowAll.html} {\bibinfo {title} {{{LogGamma Function}}}} (\bibinfo {year} {2024}{\natexlab{b}})\BibitemShut {NoStop}%
\bibitem [{Note10()}]{Note10}%
  \BibitemOpen
  \bibinfo {note} {The source code is available at \protect \url {https://github.com/peixinshen/NonHermitianFermiDiracDistributionPersistentCurrent}}\BibitemShut {NoStop}%
\bibitem [{Note7()}]{Note7}%
  \BibitemOpen
  \bibinfo {note} {In the moderate range $U$, the current amplitude may fluctuate in normal rings due to the shift of the effective Fermi level.}\BibitemShut {Stop}%
\bibitem [{\citenamefont {Trivedi}\ and\ \citenamefont {Browne}(1988)}]{Trivedi1988Mesoscopic}%
  \BibitemOpen
  \bibfield  {author} {\bibinfo {author} {\bibfnamefont {N.}~\bibnamefont {Trivedi}}\ and\ \bibinfo {author} {\bibfnamefont {D.~A.}\ \bibnamefont {Browne}},\ }\bibfield  {title} {\bibinfo {title} {Mesoscopic ring in a magnetic field: {{Reactive}} and dissipative response},\ }\href {https://doi.org/10.1103/PhysRevB.38.9581} {\bibfield  {journal} {\bibinfo  {journal} {Phys. Rev. B}\ }\textbf {\bibinfo {volume} {38}},\ \bibinfo {pages} {9581} (\bibinfo {year} {1988})}\BibitemShut {NoStop}%
\bibitem [{\citenamefont {Ferrier}\ \emph {et~al.}(2013)\citenamefont {Ferrier}, \citenamefont {Dassonneville}, \citenamefont {Gu{\'e}ron},\ and\ \citenamefont {Bouchiat}}]{Ferrier2013Phasedependent}%
  \BibitemOpen
  \bibfield  {author} {\bibinfo {author} {\bibfnamefont {M.}~\bibnamefont {Ferrier}}, \bibinfo {author} {\bibfnamefont {B.}~\bibnamefont {Dassonneville}}, \bibinfo {author} {\bibfnamefont {S.}~\bibnamefont {Gu{\'e}ron}},\ and\ \bibinfo {author} {\bibfnamefont {H.}~\bibnamefont {Bouchiat}},\ }\bibfield  {title} {\bibinfo {title} {Phase-dependent {{Andreev}} spectrum in a diffusive {{SNS}} junction: {{Static}} and dynamic current response},\ }\href {https://doi.org/10.1103/PhysRevB.88.174505} {\bibfield  {journal} {\bibinfo  {journal} {Phys. Rev. B}\ }\textbf {\bibinfo {volume} {88}},\ \bibinfo {pages} {174505} (\bibinfo {year} {2013})}\BibitemShut {NoStop}%
\bibitem [{\citenamefont {Dassonneville}\ \emph {et~al.}(2013)\citenamefont {Dassonneville}, \citenamefont {Ferrier}, \citenamefont {Gu{\'e}ron},\ and\ \citenamefont {Bouchiat}}]{Dassonneville2013Dissipation}%
  \BibitemOpen
  \bibfield  {author} {\bibinfo {author} {\bibfnamefont {B.}~\bibnamefont {Dassonneville}}, \bibinfo {author} {\bibfnamefont {M.}~\bibnamefont {Ferrier}}, \bibinfo {author} {\bibfnamefont {S.}~\bibnamefont {Gu{\'e}ron}},\ and\ \bibinfo {author} {\bibfnamefont {H.}~\bibnamefont {Bouchiat}},\ }\bibfield  {title} {\bibinfo {title} {Dissipation and {{Supercurrent Fluctuations}} in a {{Diffusive Normal-Metal--Superconductor Ring}}},\ }\href {https://doi.org/10.1103/PhysRevLett.110.217001} {\bibfield  {journal} {\bibinfo  {journal} {Phys. Rev. Lett.}\ }\textbf {\bibinfo {volume} {110}},\ \bibinfo {pages} {217001} (\bibinfo {year} {2013})}\BibitemShut {NoStop}%
\bibitem [{Note8()}]{Note8}%
  \BibitemOpen
  \bibinfo {note} {The four additional anomalous contributions are \begin {align*} &+P_{N+j+1,j,j+1,N+j} (\omega ) - P_{N+j,j,j+1,N+j+1} (\omega ) \\ &+P_{N+j,j+1,j,N+j+1} (\omega ) - P_{N+j+1,j+1,j,N+j} (\omega ) \protect \,. \end {align*} Due to the local conservation law, $\protect \mathbb {P}(\phi , \omega )$ is uniform $\forall j \in \protect \mathbb {N}$ and thus we set $j$ as the first site of $\protect \mathbb {N}$ in the calculation.}\BibitemShut {Stop}%
\bibitem [{\citenamefont {Riwar}\ \emph {et~al.}(2016)\citenamefont {Riwar}, \citenamefont {Houzet}, \citenamefont {Meyer},\ and\ \citenamefont {Nazarov}}]{Riwar2016Multiterminal}%
  \BibitemOpen
  \bibfield  {author} {\bibinfo {author} {\bibfnamefont {R.-P.}\ \bibnamefont {Riwar}}, \bibinfo {author} {\bibfnamefont {M.}~\bibnamefont {Houzet}}, \bibinfo {author} {\bibfnamefont {J.~S.}\ \bibnamefont {Meyer}},\ and\ \bibinfo {author} {\bibfnamefont {Y.~V.}\ \bibnamefont {Nazarov}},\ }\bibfield  {title} {\bibinfo {title} {Multi-terminal {{Josephson}} junctions as topological matter},\ }\href {https://doi.org/10.1038/ncomms11167} {\bibfield  {journal} {\bibinfo  {journal} {Nat. Commun.}\ }\textbf {\bibinfo {volume} {7}},\ \bibinfo {pages} {11167} (\bibinfo {year} {2016})}\BibitemShut {NoStop}%
\bibitem [{\citenamefont {Pankratova}\ \emph {et~al.}(2020)\citenamefont {Pankratova}, \citenamefont {Lee}, \citenamefont {Kuzmin}, \citenamefont {Wickramasinghe}, \citenamefont {Mayer}, \citenamefont {Yuan}, \citenamefont {Vavilov}, \citenamefont {Shabani},\ and\ \citenamefont {Manucharyan}}]{Pankratova2020Multiterminal}%
  \BibitemOpen
  \bibfield  {author} {\bibinfo {author} {\bibfnamefont {N.}~\bibnamefont {Pankratova}}, \bibinfo {author} {\bibfnamefont {H.}~\bibnamefont {Lee}}, \bibinfo {author} {\bibfnamefont {R.}~\bibnamefont {Kuzmin}}, \bibinfo {author} {\bibfnamefont {K.}~\bibnamefont {Wickramasinghe}}, \bibinfo {author} {\bibfnamefont {W.}~\bibnamefont {Mayer}}, \bibinfo {author} {\bibfnamefont {J.}~\bibnamefont {Yuan}}, \bibinfo {author} {\bibfnamefont {M.~G.}\ \bibnamefont {Vavilov}}, \bibinfo {author} {\bibfnamefont {J.}~\bibnamefont {Shabani}},\ and\ \bibinfo {author} {\bibfnamefont {V.~E.}\ \bibnamefont {Manucharyan}},\ }\bibfield  {title} {\bibinfo {title} {Multiterminal {{Josephson Effect}}},\ }\href {https://doi.org/10.1103/PhysRevX.10.031051} {\bibfield  {journal} {\bibinfo  {journal} {Phys. Rev. X}\ }\textbf {\bibinfo {volume} {10}},\ \bibinfo {pages} {031051} (\bibinfo {year} {2020})}\BibitemShut {NoStop}%
\bibitem [{\citenamefont {Coraiola}\ \emph {et~al.}(2023)\citenamefont {Coraiola}, \citenamefont {Haxell}, \citenamefont {Sabonis}, \citenamefont {Weisbrich}, \citenamefont {Svetogorov}, \citenamefont {Hinderling}, \citenamefont {{ten Kate}}, \citenamefont {Cheah}, \citenamefont {Krizek}, \citenamefont {Schott}, \citenamefont {Wegscheider}, \citenamefont {Cuevas}, \citenamefont {Belzig},\ and\ \citenamefont {Nichele}}]{Coraiola2023Phaseengineering}%
  \BibitemOpen
  \bibfield  {author} {\bibinfo {author} {\bibfnamefont {M.}~\bibnamefont {Coraiola}}, \bibinfo {author} {\bibfnamefont {D.~Z.}\ \bibnamefont {Haxell}}, \bibinfo {author} {\bibfnamefont {D.}~\bibnamefont {Sabonis}}, \bibinfo {author} {\bibfnamefont {H.}~\bibnamefont {Weisbrich}}, \bibinfo {author} {\bibfnamefont {A.~E.}\ \bibnamefont {Svetogorov}}, \bibinfo {author} {\bibfnamefont {M.}~\bibnamefont {Hinderling}}, \bibinfo {author} {\bibfnamefont {S.~C.}\ \bibnamefont {{ten Kate}}}, \bibinfo {author} {\bibfnamefont {E.}~\bibnamefont {Cheah}}, \bibinfo {author} {\bibfnamefont {F.}~\bibnamefont {Krizek}}, \bibinfo {author} {\bibfnamefont {R.}~\bibnamefont {Schott}}, \bibinfo {author} {\bibfnamefont {W.}~\bibnamefont {Wegscheider}}, \bibinfo {author} {\bibfnamefont {J.~C.}\ \bibnamefont {Cuevas}}, \bibinfo {author} {\bibfnamefont {W.}~\bibnamefont {Belzig}},\ and\ \bibinfo {author} {\bibfnamefont {F.}~\bibnamefont {Nichele}},\ }\bibfield  {title} {\bibinfo {title} {Phase-engineering the {{Andreev}} band structure of a three-terminal {{Josephson}} junction},\ }\href {https://doi.org/10.1038/s41467-023-42356-6} {\bibfield  {journal} {\bibinfo  {journal} {Nat. Commun.}\ }\textbf {\bibinfo {volume} {14}},\ \bibinfo {pages} {6784} (\bibinfo {year} {2023})}\BibitemShut {NoStop}%
\bibitem [{\citenamefont {Shen}\ \emph {et~al.}(2021)\citenamefont {Shen}, \citenamefont {Hoffman},\ and\ \citenamefont {Trif}}]{Shen2021Theory}%
  \BibitemOpen
  \bibfield  {author} {\bibinfo {author} {\bibfnamefont {P.-X.}\ \bibnamefont {Shen}}, \bibinfo {author} {\bibfnamefont {S.}~\bibnamefont {Hoffman}},\ and\ \bibinfo {author} {\bibfnamefont {M.}~\bibnamefont {Trif}},\ }\bibfield  {title} {\bibinfo {title} {Theory of topological spin {{Josephson}} junctions},\ }\href {https://doi.org/10.1103/PhysRevResearch.3.013003} {\bibfield  {journal} {\bibinfo  {journal} {Phys. Rev. Res.}\ }\textbf {\bibinfo {volume} {3}},\ \bibinfo {pages} {013003} (\bibinfo {year} {2021})}\BibitemShut {NoStop}%
\bibitem [{\citenamefont {Moskalets}\ and\ \citenamefont {B{\"u}ttiker}(2002)}]{Moskalets2002Floquet}%
  \BibitemOpen
  \bibfield  {author} {\bibinfo {author} {\bibfnamefont {M.}~\bibnamefont {Moskalets}}\ and\ \bibinfo {author} {\bibfnamefont {M.}~\bibnamefont {B{\"u}ttiker}},\ }\bibfield  {title} {\bibinfo {title} {Floquet scattering theory of quantum pumps},\ }\href {https://doi.org/10.1103/PhysRevB.66.205320} {\bibfield  {journal} {\bibinfo  {journal} {Phys. Rev. B}\ }\textbf {\bibinfo {volume} {66}},\ \bibinfo {pages} {205320} (\bibinfo {year} {2002})}\BibitemShut {NoStop}%
\bibitem [{\citenamefont {Blaauboer}(2002)}]{Blaauboer2002Charge}%
  \BibitemOpen
  \bibfield  {author} {\bibinfo {author} {\bibfnamefont {M.}~\bibnamefont {Blaauboer}},\ }\bibfield  {title} {\bibinfo {title} {Charge pumping in mesoscopic systems coupled to a superconducting lead},\ }\href {https://doi.org/10.1103/PhysRevB.65.235318} {\bibfield  {journal} {\bibinfo  {journal} {Phys. Rev. B}\ }\textbf {\bibinfo {volume} {65}},\ \bibinfo {pages} {235318} (\bibinfo {year} {2002})}\BibitemShut {NoStop}%
\bibitem [{\citenamefont {Becerra}\ \emph {et~al.}(2023)\citenamefont {Becerra}, \citenamefont {Trif},\ and\ \citenamefont {Hyart}}]{Becerra2023Quantized}%
  \BibitemOpen
  \bibfield  {author} {\bibinfo {author} {\bibfnamefont {V.~F.}\ \bibnamefont {Becerra}}, \bibinfo {author} {\bibfnamefont {M.}~\bibnamefont {Trif}},\ and\ \bibinfo {author} {\bibfnamefont {T.}~\bibnamefont {Hyart}},\ }\bibfield  {title} {\bibinfo {title} {Quantized {{Spin Pumping}} in {{Topological Ferromagnetic-Superconducting Nanowires}}},\ }\href {https://doi.org/10.1103/PhysRevLett.130.237002} {\bibfield  {journal} {\bibinfo  {journal} {Phys. Rev. Lett.}\ }\textbf {\bibinfo {volume} {130}},\ \bibinfo {pages} {237002} (\bibinfo {year} {2023})}\BibitemShut {NoStop}%
\bibitem [{\citenamefont {Beenakker}(2024)}]{Beenakker2024Josephson}%
  \BibitemOpen
  \bibfield  {author} {\bibinfo {author} {\bibfnamefont {C.~W.~J.}\ \bibnamefont {Beenakker}},\ }\bibfield  {title} {\bibinfo {title} {Josephson effect in a junction coupled to an electron reservoir},\ }\Eprint {https://arxiv.org/abs/2404.13976} {arXiv:2404.13976}  (\bibinfo {year} {2024})\BibitemShut {NoStop}%
\bibitem [{\citenamefont {Pino}\ \emph {et~al.}(2024)\citenamefont {Pino}, \citenamefont {Meir},\ and\ \citenamefont {Aguado}}]{Pino2024Thermodynamics}%
  \BibitemOpen
  \bibfield  {author} {\bibinfo {author} {\bibfnamefont {D.~M.}\ \bibnamefont {Pino}}, \bibinfo {author} {\bibfnamefont {Y.}~\bibnamefont {Meir}},\ and\ \bibinfo {author} {\bibfnamefont {R.}~\bibnamefont {Aguado}},\ }\bibfield  {title} {\bibinfo {title} {Thermodynamics of {{Non-Hermitian Josephson}} junctions with exceptional points},\ }\Eprint {https://arxiv.org/abs/2405.02387} {arXiv:2405.02387}  (\bibinfo {year} {2024})\BibitemShut {NoStop}%
\end{thebibliography}%

\onecolumngrid
\clearpage

\setcounter{secnumdepth}{3}
\makeatletter
\xdef\presupfigures{\arabic{figure}}
\setcounter{equation}{0}
\setcounter{table}{0}
\setcounter{page}{1}
\renewcommand{\thefigure}{S\fpeval{\arabic{figure}-\presupfigures}}
\renewcommand{\theequation}{S\@arabic\c@equation}
\renewcommand{\thetable}{S\@arabic\c@table}
\renewcommand{\thepage}{S\@arabic\c@page}
\newtheorem{lemma}{Lemma}
\newtheorem{theorem}{Theorem}
\newtheorem{corollary}{Corollary}

\begin{center} 
	{\large \bf Supplemental Material: Non-Hermitian Fermi-Dirac Distribution in~Persistent~Current~Transport}
    \vskip 1em
    Pei-Xin Shen$^{1,2,\hyperlink{Email:Shen}{*}}$,
    Zhide Lu$^2$, Jose L. Lado$^3$, and 
    Mircea Trif$^{1,\hyperlink{Email:Trif}{\dagger}}$
    \vskip 0.5em
    \textit{
    $^1$International Research Centre MagTop, Institute of Physics,\\
    Polish Academy of Sciences, Aleja Lotnikow 32/46, PL-02668 Warsaw, Poland\\
    $^2$Institute for Interdisciplinary Information Sciences, Tsinghua University, Beijing 100084, China\\
    $^3$Department of Applied Physics, Aalto University, FI-00076 Aalto, Espoo, Finland
    }\\
    (Dated: \today)
\end{center}

In this Supplemental Material, we cover: (i) fundamental concepts of Green's function and self-energy; (ii) a detailed analysis of the non-Hermitian Fermi-Dirac distribution, including gauge invariance, asymptotic behaviors, and current continuity near EPs; (iii) methods for calculating the current susceptibility using non-Hermitian approaches and exact diagonalization; and (iv) extended numerical verifications, containing low-energy spectra and persistent currents, alongside performance benchmarks across both weak and strong coupling regimes.

\section{Effective Hamiltonian derivation}

The self-energy for normal metals has been pedagogically introduced in Ref.~\cite{Datta2005Quantum}. Here, we outline the main procedures for incorporating the self-energy into the BdG Hamiltonian and obtain the effective Hamiltonian of superconducting systems. 

\subsection{Bare Green's function of a fermionic reservoir}
The tight-binding Hamiltonian of the $N_\textsc{e}$-site 1D fermionic wire with open boundary conditions is given by:
\begin{align}
\label{Eqn:Reservoir}
   H_{\rm res} 
   &= \sum_{j=1}^{N_\textsc{e}} 
   \big( t c^\dagger_j c_{j+1} + t c^\dagger_{j+1} c_j + g c^\dagger_j c_j \big) 
   = \vec{C}^\dagger \mathcal{H}_{\rm res} \vec{C}  \,, \quad \text{with} \quad
   \mathcal{H}_{\rm res} 
   = \sum_{j=1}^{N_\textsc{e}} \big( t \ket{j}\bra{j+1} + t \ket{j+1}\bra{j} + g \ket{j}\bra{j} \big) \,,
\end{align}
where calligraphy is specifically used to represent the first quantized operators, and $\ket{j}$ denotes the single-particle state at the $j$-th site. The exact single-particle eigenvalues and normalized eigenstates of $\mathcal{H}_{\rm res}$ are as follows:
\begin{align}
    \epsilon_k = 2t \cos(k) + g \,, \quad
    \braket{j|\psi_k} = \psi_{kj} 
    = \sqrt{2/(N_\textsc{e}+1)} \sin(kj) \,, \quad
    k \in \{ n \pi /(N_\textsc{e}+1)  \}_{n=1}^{N_\textsc{e}} \,.
    \end{align}
With these quantities, we can calculate its bare Green's function $g_{\rm res} = 1/(\omega - \mathcal{H}_{\rm res})$ in the thermodynamic limit $N_\textsc{e} \rightarrow \infty$,
\begin{align}
    &g_{\rm res}(l,m; \omega) 
    = \sum_k \frac{\braket{l|\psi_k}\braket{\psi_k|m}}{\omega-\epsilon_k} 
    = \frac{1}{\pi} \sum_k \frac{\pi}{N_\textsc{e}+1} \frac{\sqrt{2}\sin(kl) \sqrt{2}\sin(km)}{\omega-\epsilon_k} 
    = \frac{1}{\pi} \int_{0}^{\pi} \mathrm{d} k \frac{2\sin(k l) \sin(k m)}{\omega - [2t \cos(k) + g]} \nonumber \\
    &= \frac{1}{2 \pi} \int_{-\pi}^{\pi} \mathrm{d} k \frac{\mathrm{e}^{+\mathrm{i} k|l-m|} - \mathrm{e}^{+\mathrm{i} k|l+m|}}{\omega - [2t \cos( k) + g]} 
    =\begin{cases} 
    \displaystyle
    \frac{(-x - \sqrt{x^2-1})^{|l-m|} - (-x - \sqrt{x^2-1})^{|l+m|}}{+2t\sqrt{x^2-1}} \,, \quad x = \frac{g - \omega}{2t} < -1 \,. \\
    \displaystyle
    \frac{(-x + \sqrt{x^2-1})^{|l-m|} - (-x + \sqrt{x^2-1})^{|l+m|}}{-2t\sqrt{x^2-1}} \,, \quad x = \frac{g - \omega}{2t} > +1 \,. 
    \end{cases} 
\end{align}
When $\omega$ lies within the bandwidth $-1 \leqslant x \leqslant 1$, we define the retarded ($+$) and advanced ($-$) Green's function:
\begin{align}
    \label{Eqn:ReservoirGFnormal}
    g_{\rm res}^\pm(l,m; \omega) 
    &= \frac{(-x \pm \mathrm{i} \sqrt{1-x^2})^{|l-m|} - (-x \pm \mathrm{i} \sqrt{1-x^2})^{|l+m|}}{\pm 2\mathrm{i}  |t| \sqrt{1-x^2}}
    = \frac{\mathrm{e}^{\pm \mathrm{i} k |l-m|} - \mathrm{e}^{\pm \mathrm{i} k |l+m|}}{\pm 2\mathrm{i} |t| \sin(k)} \,, 
\end{align}
where we use $t<0$, $x = -\cos(k)$ and $-x \pm \mathrm{i} \sqrt{1-x^2} = \mathrm{e}^{\pm \mathrm{i} k}$. 
Next, we focus on the edge Green's function $l = m = 1$,
\begin{align}
\label{Eqn:SurfaceGreenFunction}
    g_\mathrm{edge}^\pm(\omega) \equiv
    g_{\rm res}^\pm(1,1; \omega) 
    &= \mathrm{e}^{\pm \mathrm{i} k}/t
    = \big[ (\omega - g)/2 \mp \mathrm{i} \sqrt{t^2 - [(\omega - g)/2]^2} \big]/t^2 \,,
\end{align}
which plays a crucial role in determining the self-energy of the reservoir connected to the system through a point contact.

\subsection{Self-energy induced by a fermionic reservoir}

Here we derive the self-energy for the $N$-site superconducting system, arising from its interaction with the normal reservoir. To ease the discussion, we assume that the reservoir is coupled to the right end of the system via $H_{\rm tun} = \kappa (c^\dagger_\textsc{n} c_\textsc{n+1} + c^\dagger_\textsc{n+1} c_\textsc{n})$, and rewrite the Hamiltonian of the full system $H_{\rm tot} = \vec{C}^\dagger \mathcal{H}_{\rm tot} \vec{C}/2$ in the basis 
$\vec{C} = (c_1, c^\dagger_1, \dots, c_\textsc{n}, c^\dagger_\textsc{n}, c_\textsc{n+1}, c^\dagger_\textsc{n+1}, \dots, c_\textsc{n+n$_\textsc{e}$}, c^\dagger_\textsc{n+n$_\textsc{e}$})^\textsc{t}$
with the first-quantized Hamiltonian $\mathcal{H}_{\rm tot} = \mathcal{H}_{\rm sys} + \tilde{\mathcal{H}}_{\rm res} + \tilde{\mathcal{H}}_{\rm tun}$. We denote the extended BdG Hamiltonian with a tilde, e.g., $\tilde{\mathcal{H}}_{\rm res} = \mathcal{H}_{\rm res} \otimes \tau_z$ is a $2N_\textsc{e} \times 2 N_\textsc{e}$ matrix as an extension of Eq.~\eqref{Eqn:Reservoir} with the index $j$ running from $N+1$ to $N+N_\textsc{e}$, where $\tau_z$ is the Pauli-$Z$ matrix acting in the particle-hole space and
\begin{align}
    \label{Eqn:TunnelingV}
    \tilde{\mathcal{H}}_{\rm tun}
    = \mathcal{H}_{\rm tun} \otimes \tau_z
    = \big( \kappa \ket{N}\bra{N+1} + \kappa \ket{N+1}\bra{N} \big) \otimes \tau_z =
    \begin{bmatrix}
        0 & \mathcal{V} \\
        \mathcal{V}^\dagger & 0
    \end{bmatrix}\,, \quad
    \mathcal{V} = 
    \begin{bmatrix}
        0 & \cdots & 0\\
        \vdots & \ddots & \vdots \\
        \kappa & \cdots & 0 
    \end{bmatrix} 
    \otimes \tau_z \,.
\end{align}
Given the bare Green's function of the system $g_{\rm sys} = 1/(\omega - \mathcal{H}_{\rm sys})$ and the bare Green's function of the BdG-extended reservoir $\tilde{g}_{\rm res} = 1/(\omega - \tilde{\mathcal{H}}_{\rm res})$, we compute the Green's function $G_{\rm tot}$ of the full system by Dyson equation,
\begin{equation}
    G_{\rm tot} = g_{\rm tot} + g_{\rm tot} \tilde{\mathcal{H}}_{\rm tun} G_{\rm tot} \quad \Longleftrightarrow \quad
    \begin{bmatrix}
        G_{\rm sys}   & G_{12} \\
        G_{21} & G_{\rm res}
    \end{bmatrix} = 
    \begin{bmatrix}
        g_{\rm sys}   & 0 \\
        0 & \tilde{g}_{\rm res}
    \end{bmatrix} + 
    \begin{bmatrix}
        g_{\rm sys}   & 0 \\
        0 & \tilde{g}_{\rm res}
    \end{bmatrix}
    \begin{bmatrix}
        0 & \mathcal{V}\\
        \mathcal{V}^\dagger & 0
    \end{bmatrix}
    \begin{bmatrix}
        G_{\rm sys}   & G_{12} \\
        G_{21} & G_{\rm res}
    \end{bmatrix} \,.
\end{equation}
By eliminating $G_{12}$ and $G_{21}$, we have $G_{\rm sys} = g_{\rm sys} + g_{\rm sys} (\mathcal{V} \tilde{g}_{\rm res}\mathcal{V}^\dagger) G_{\rm sys} \equiv g_{\rm sys} + g_{\rm sys} \Sigma(\omega) G_{\rm sys}$ and thus obtain
\begin{align}
    G_{\rm sys}(\omega) = \frac{1}{1/g_{\rm sys} - \Sigma(\omega)} = \frac{1}{\omega- [\mathcal{H}_{\rm sys} + \Sigma(\omega)]} \equiv \frac{1}{\omega- \mathcal{H}_{\rm eff}(\omega)} \,.
    \label{Eqn:GreenEffetiveHamiltonian}
\end{align}
To obtain the explicit form of the retarded self-energy $\Sigma(\omega) = \mathcal{V} \tilde{g}^+_{\rm res} (\omega) \mathcal{V}^\dagger$, it is more convenient to represent the matrices back to the basis $\vec{C} = (c_1, \dots, c_\textsc{n}, c_\textsc{n+1}, \dots, c_\textsc{n+n$_\textsc{e}$}, c^\dagger_1, \dots, c^\dagger_\textsc{n}, c^\dagger_\textsc{n+1}, \dots, c^\dagger_\textsc{n+n$_\textsc{e}$})^\textsc{t}$, which involves exchanging the particle-hole space and the lattice space, e.g., $\tilde{\mathcal{H}}_{\rm res} = \tau_z \otimes \mathcal{H}_{\rm res}$. 
As such, we can express the retarded Green's function of the BdG-extended reservoir with $g^\pm_{\rm res}$ as
\begin{equation}
    \tilde{g}^+_{\rm res}(\omega) = \frac{1}{\omega + \mathrm{i} \eta -\tilde{\mathcal{H}}_{\rm res}} = 
    \begin{bmatrix} \displaystyle
        +1/(+ \omega + \mathrm{i} \eta - \mathcal{H}_{\rm res}) & 0 \\
        0 & -1/(- \omega - \mathrm{i} \eta - \mathcal{H}_{\rm res})
    \end{bmatrix} 
    = 
    \begin{bmatrix}
        +g^+_{\rm res}(+\omega) & 0 \\
        0 & - g^-_{\rm res}(-\omega)
    \end{bmatrix} \,,
\end{equation}
where $\eta$ is a positive infinitesimal linewidth. Using Eqs.~\eqref{Eqn:SurfaceGreenFunction} and \eqref{Eqn:TunnelingV}, the self-energy is given by
\begin{align}
    \Sigma(\omega)
    &= \kappa^2
    \begin{bmatrix}
        +g^+_\mathrm{edge}(+\omega) & 0 \\
        0 & -g^-_\mathrm{edge}(-\omega) 
    \end{bmatrix} \otimes \ket{N}\bra{N}
    = 
    \frac{\kappa^2}{t^2}
    \left[ \frac{\omega - g \tau_z}{2} - \mathrm{i} \sqrt{t^2 - \left(\frac{\omega - g \tau_z}{2} \right)^2} \right] \otimes \ket{N}\bra{N} \,.
    \label{Eqn:SelfEnergy}
\end{align}
The expression above represents the self-energy for the right-edge coupling. By replacing the lattice index $N$ with $l$ in Eq.~\eqref{Eqn:SelfEnergy}, one can similarly derive the self-energy $\Sigma_l(\omega)$ for a system whose $l$-site is connected to the reservoir.
The self-energy $\Sigma(\omega)$ takes care of the couplings and contains all essential information about the external reservoir. Hence, as shown in Eq.~\eqref{Eqn:GreenEffetiveHamiltonian}, $\mathcal{H}_{\rm eff}(\omega) = \mathcal{H}_{\rm sys} + \Sigma(\omega)$ effectively governs the dynamics of the system coupled to the reservoir. 
So far, we have not made any approximation except for $N_\textsc{e} \rightarrow \infty$. 
Since we are interested in the low-energy sector, we adopt the wide-band limit $\Sigma(\omega) \approx \Sigma(0)$
and interpret $\mathcal{H}_{\rm eff} \equiv \mathcal{H}_{\rm eff}(0)$ as an effective Hamiltonian for the system \cite{Landi2022Nonequilibrium}. 
To this end, the self-energy for the superconducting system is $\Sigma_l(0) = - \kappa^2/t^2 (\tau_z g /2 + \mathrm{i} \sqrt{t^2 - (g/2)^2} ) \otimes \ket{l}\bra{l}$, which also applies to normal metallic systems when $\tau_z$ is removed. The effective Hamiltonian of the superconducting system possesses a particle-hole symmetry $\mathcal{H}_{\rm eff} = - \tau_x \mathcal{H}^*_{\rm eff} \tau_x$, which is inherited from the self-energy $\Sigma(0) = - \tau_x \Sigma^*(0) \tau_x$, where $\tau_x$ is the Pauli-$X$ acting in the particle-hole space. This symmetry ensures that its complex spectra consist of pairs $(+\varepsilon_n, - \varepsilon^*_n)$, as illustrated in Fig.~\figref{Fig:SketchSpectrum}{c} and Fig.~\figref{Fig:SpectraComparison}{a}. 

\section{Non-Hermitian Fermi-Dirac distribution}

\subsection{Derivation}

To compute the correlator $\braket{c^\dagger_i c_j} = \int \braket{j|\rho(\omega)|i} f_\textsc{fd}(\omega) \mathrm{d} \omega$, the conventional way is to decompose $\rho(\omega)$ into two parts:
\begin{align}
\rho (\omega)
    &= \frac{\rm i}{2 \pi} \sum_n \left[\frac{\ket{\psi^\textsc{r}_n}\bra{\psi^\textsc{l}_n}}{\omega - \varepsilon_n} - \frac{\ket{\psi^\textsc{l}_n}\bra{\psi^\textsc{r}_n}}{\omega - \varepsilon_n^*} \right] 
    = \frac{\rm i}{2 \pi} \sum_n \left[\frac{\ket{\psi^\textsc{r}_n}\bra{\psi^\textsc{l}_n}}{\omega - (\mathrm{Re}\,\varepsilon_n + \mathrm{i}\, \mathrm{Im}\,\varepsilon_n)} - \frac{\ket{\psi^\textsc{l}_n}\bra{\psi^\textsc{r}_n}}{\omega - (\mathrm{Re}\,\varepsilon_n - \mathrm{i}\, \mathrm{Im}\,\varepsilon_n)} \right] \\
    &= \frac{\rm i}{2 \pi} \sum_n \left[ (\ket{\psi^\textsc{r}_n}\bra{\psi^\textsc{l}_n} - \ket{\psi^\textsc{l}_n}\bra{\psi^\textsc{r}_n}) \frac{(\omega - \mathrm{Re}\,\varepsilon_n)}{(\omega - \mathrm{Re}\,\varepsilon_n)^2 + (\mathrm{Im}\,\varepsilon_n)^2} + (\ket{\psi^\textsc{r}_n}\bra{\psi^\textsc{l}_n} + \ket{\psi^\textsc{l}_n}\bra{\psi^\textsc{r}_n}) \frac{\mathrm{i}\,\mathrm{Im}\,\varepsilon_n}{(\omega - \mathrm{Re}\,\varepsilon_n)^2 + (\mathrm{Im}\,\varepsilon_n)^2} \right] \,. \nonumber
\end{align}
In the context of a Hermitian system where $\mathrm{Im}\, \varepsilon_n \rightarrow 0$, the first part contributes to a principal value (PV) while the second part becomes a Lorentzian. In the weak coupling limit $\ket{\psi^\textsc{l}_n} \approx \ket{\psi^\textsc{r}_n}$, the PV term is close to zero and thus is often omitted in the literature \cite{Rickayzen1980Green}. However, we have found that this term is crucial for correctly evaluating observables at EPs, which are typically encountered in the strong-coupling regime. As such, we choose to calculate the correlator by analyzing the branch cut of $\ln (\varepsilon)$ along the negative real axis. Before delving into the technical details, we first introduce the following useful identities:
\begin{align}
    \label{Eqn:LogIdentities}
    \lim_{\eta \rightarrow 0} \ln (\epsilon \pm \mathrm{i} \eta) = \ln |\epsilon| \pm \mathrm{i} \pi \Theta(-\epsilon) \,, \quad
    \lim_{\omega \rightarrow -\infty} \ln (\omega + \varepsilon_n) = \mathfrak{C} + \mathrm{i} \pi\, \mathrm{sgn}(\mathrm{Im} \, \varepsilon_n) \,, \quad
    \sum_n \big( \psi^\textsc{l*}_{ni}\psi^\textsc{r}_{nj} \mathfrak{C} \big)
    = \delta_{ij}  \mathfrak{C} \,,
\end{align}
where $\varepsilon_n \in \mathbb{C}$, $\epsilon \in \mathbb{R}$, $\eta > 0$, and $\mathfrak{C}$ represents any $\varepsilon_n$-independent constant. This includes the scenario where $\mathfrak{C}$ grows logarithmically towards infinity, rendering any initial dependence on $\varepsilon_n$ negligible in that limit. Consequently, we obtain
\begin{align}
\label{Eqn:CorrelatorZero}
\braket{c^\dagger_i c_j} 
    &= \frac{\rm i}{2 \pi} \int_{-\infty}^0 \sum_n \left(\frac{\psi^\textsc{l*}_{ni}\psi^\textsc{r}_{nj}}{\omega - \varepsilon_n} - \frac{\psi^\textsc{r*}_{ni}\psi^\textsc{l}_{nj}}{\omega - \varepsilon_n^*} \right) \mathrm{d} \omega
    = \frac{\mathrm{i}}{2\pi} \sum_n \left[ \psi^\textsc{l*}_{ni}\psi^\textsc{r}_{nj} \ln (\omega - \varepsilon_n) - \psi^\textsc{r*}_{ni}\psi^\textsc{l}_{nj} \ln (\omega - \varepsilon^*_n)  \right] \big|_{-\infty}^0 \nonumber \\
    &= \frac{\mathrm{i}}{2\pi} \sum_n \left[ \psi^\textsc{l*}_{ni}\psi^\textsc{r}_{nj} \Big(\ln (- \varepsilon_n) - (\mathfrak{C} +\mathrm{i} \pi) \Big) - \psi^\textsc{r*}_{ni}\psi^\textsc{l}_{nj} \Big( \ln (- \varepsilon^*_n) - (\mathfrak{C} -\mathrm{i} \pi) \Big) \right] \nonumber \\
    &= \frac{\mathrm{i}}{2\pi} \sum_n \left[ \psi^\textsc{l*}_{ni}\psi^\textsc{r}_{nj} \Big( \ln|\varepsilon_n| + \mathrm{i} (\arg \varepsilon_n + \pi) - (+\mathrm{i} \pi) \Big) - \psi^\textsc{r*}_{ni}\psi^\textsc{l}_{nj} \Big( \ln|\varepsilon_n| - \mathrm{i} (\arg \varepsilon_n + \pi) - (-\mathrm{i} \pi) \Big)  \right] \nonumber \\
    &= \frac{\mathrm{i}}{2\pi} \sum_n \left(\psi^\textsc{l*}_{ni}\psi^\textsc{r}_{nj} \ln \varepsilon_n - \psi^\textsc{r*}_{ni}\psi^\textsc{l}_{nj} \ln \varepsilon_n^* \right)
    \equiv \frac{1}{2 \mathrm{i}} \sum_n \left[\psi^\textsc{l*}_{ni}\psi^\textsc{r}_{nj} f_{\rm eff}(\varepsilon_n) - \psi^\textsc{r*}_{ni}\psi^\textsc{l}_{nj} f^*_{\rm eff}(\varepsilon_n) \right] \,,
\end{align}
where we use the identities in Eq.~\eqref{Eqn:LogIdentities} to cancel out $\mathfrak{C}$.
For $T \neq 0$, we
can perform the integrals using the residue theorem:
\begin{align}
    \label{Eqn:CorrelatorFinite}
    \braket{c^\dagger_i c_j} 
    &= \frac{\rm i}{2 \pi} \int_{-\infty}^\infty \frac{1}{\mathrm{e}^{\beta \omega}+1} \sum_n \left(\frac{\psi^\textsc{l*}_{ni}\psi^\textsc{r}_{nj}}{\omega - \varepsilon_n} - \frac{\psi^\textsc{r*}_{ni}\psi^\textsc{l}_{nj}}{\omega - \varepsilon_n^*} \right) \mathrm{d} \omega \nonumber \\
    &= \frac{\mathrm{i}}{2\pi} \sum_n \Bigg\{\psi^\textsc{l*}_{ni}\psi^\textsc{r}_{nj} \frac{1}{2} \left[\frac{-2\pi \mathrm{i}}{\mathrm{e}^{\beta \varepsilon_n}+1} - \sum^\infty_{m = 1} \left(m-\frac{1}{2}+\frac{\mathrm{i} \beta \varepsilon_n}{2\pi}\right)^{-1} - \sum^\infty_{m = 1} \left(m-\frac{1}{2}-\frac{\mathrm{i} \beta \varepsilon_n}{2\pi}\right)^{-1} \right] \nonumber \\
    &\qquad \quad \;\;\;\, - \psi^\textsc{r*}_{ni}\psi^\textsc{l}_{nj} \frac{1}{2} \left[\frac{+2\pi \mathrm{i}}{\mathrm{e}^{\beta \varepsilon_n^*}+1} - \sum^\infty_{m = 1} \left(m-\frac{1}{2}+\frac{\mathrm{i} \beta \varepsilon_n^*}{2\pi}\right)^{-1} - \sum^\infty_{m = 1} \left(m-\frac{1}{2}-\frac{\mathrm{i} \beta \varepsilon^*_n}{2\pi}\right)^{-1} \right] \Bigg\} \nonumber \\
    &= \frac{\mathrm{i}}{2\pi} \sum_n \left\{ \psi^\textsc{l*}_{ni}\psi^\textsc{r}_{nj} \left[\Psi \left(\frac{1}{2}+\frac{\mathrm{i} \beta \varepsilon_n}{2\pi}\right) - \frac{\mathrm{i} \pi}{2} \right] - \psi^\textsc{r*}_{ni}\psi^\textsc{l}_{nj} \left[\Psi \left(\frac{1}{2}+\frac{\mathrm{i} \beta \varepsilon_n}{2\pi}\right) - \frac{\mathrm{i} \pi}{2} \right]^* \right\} \nonumber \\
    &\equiv \frac{1}{2 \mathrm{i}} \sum_n \left[\psi^\textsc{l*}_{ni}\psi^\textsc{r}_{nj} f_{\rm eff}(\varepsilon_n, \beta) - \psi^\textsc{r*}_{ni}\psi^\textsc{l}_{nj} f^*_{\rm eff}(\varepsilon_n, \beta) \right] \,.
\end{align}
Above, we use the following identities to simplify the sum of residue at $\{(m - 1/2) 2 \pi \mathrm{i} / \beta \}_{m \in \mathbb{Z}}$ to the digamma function $\Psi$ \cite{Weisstein2024Digamma}: 
\begin{align}
    \label{Eqn:GammaIdentities}
    \sum^M_{m=1} \frac{1}{m+z-1} = \Psi(M+z) - \Psi(z) \,, \, z \in \mathbb{C} \, 
    &\quad \Longrightarrow \quad 
    \sum^\infty_{m = 1} \left(m-\frac{1}{2} \pm \frac{\mathrm{i} \beta \varepsilon}{2\pi}\right)^{-1} = \mathfrak{C} - \Psi \left(\frac{1}{2} \pm \frac{\mathrm{i} \beta \varepsilon}{2\pi} \right) \,, \nonumber \\
    \Psi(1 - z) - \Psi(z) = \pi \cot (\pi z) \,, \, z \in \mathbb{C} \,
    &\quad \Longrightarrow \quad
    \Psi \left(\frac{1}{2} - \frac{\mathrm{i} \beta \varepsilon}{2\pi} \right) - \Psi \left(\frac{1}{2} + \frac{\mathrm{i} \beta \varepsilon}{2\pi} \right) = \frac{\pi \mathrm{i}}{\mathrm{e}^{+\beta \varepsilon}+1} - \frac{\pi \mathrm{i}}{\mathrm{e}^{-\beta \varepsilon}+1} \,, 
\end{align}
and utilize identities in Eq.~\eqref{Eqn:LogIdentities} to eliminate $\mathfrak{C}$ and finally obtain Eq.~\eqref{Eqn:CorrelatorFinite} with the NH Fermi-Dirac distribution $f_{\rm eff}(\varepsilon, \beta)$. It is evident that the biorthogonality of the wavefunctions is essential to cancel out the divergence in the integral. The anomalous correlator $\braket{c_i c_j} = \int \braket{j|\rho(\omega)|N+i} f_\textsc{fd}(\omega) \mathrm{d} \omega$ can be obtained similarly by replacing $i$ with $N+i$ in Eqs.~\eqref{Eqn:CorrelatorZero}-\eqref{Eqn:CorrelatorFinite}. Therefore, a general quadratic Hermitian operator $O = \vec{C}^\dagger \mathcal{O} \vec{C}/\mathbf{2}$ can be computed by
\begin{align}
\braket{O}
= \frac{1}{2\mathrm{i}} \sum_n \left[ \frac{\braket{\psi^\textsc{l}_n|\mathcal{O}|\psi^\textsc{r}_n}}{\mathbf{2}} f_{\rm eff}(\varepsilon_n) - \frac{\braket{\psi^\textsc{r}_n|\mathcal{O}|\psi^\textsc{l}_n}}{\mathbf{2}} f^*_{\rm eff}(\varepsilon_n) \right] 
= \mathrm{Im} \sum_n \frac{\braket{\psi^\textsc{l}_n|\mathcal{O}|\psi^\textsc{r}_n}}{\mathbf{2}} f_{\rm eff}(\varepsilon_n)
= \frac{\mathrm{Im} \mathrm{Tr} [ \mathcal{O} f_{\rm eff}(\mathcal{H}_{\rm eff})]}{\mathbf{2}} \,,
\end{align}
while other non-quadratic operators can be calculated using Wick's theorem.

\subsection{Gauge invariance}

The NH Fermi-Dirac distribution $f_{\rm eff}$ acquires extra degrees of freedom owing to its association with a set of biorthogonal single-particle eigenstates. To examine its gauge invariance, we start with a Lemma for the correlators in Eqs.~\eqref{Eqn:CorrelatorZero}-\eqref{Eqn:CorrelatorFinite}:
\begin{lemma}\label{Thm:Correlators} 
    Given two $\varepsilon_n$-independent complex numbers $C_{1}$ and $C_{2}$, applying two separate transformations on the NH Fermi-Dirac distribution $f_{\rm eff} \rightarrow f_{\rm eff} + C_1$ associated with $\psi^\mathrm{L*}_{ni}\psi^\mathrm{R}_{nj}$, and $f^*_{\rm eff} \rightarrow [f_{\rm eff} + C_2]^*$ associated with $\psi^\mathrm{R*}_{ni}\psi^\mathrm{L}_{nj}$ will keep $\braket{c_i c_j}$ unchanged while transforming
    \begin{align}
    \label{Eqn:CorrelatorGauge}
    \braket{c^\dagger_i c_j}
    \rightarrow \braket{c^\dagger_i c_j} + \frac{1}{2\mathrm{i}} \delta_{ij} (C_1 - C^*_2) \,. 
    \end{align}
\end{lemma}
\noindent Lemma~\ref{Thm:Correlators} naturally follows from the biorthogonality of wavefunctions in Eq.~\eqref{Eqn:LogIdentities} and contributes the following theorem:
\begin{theorem}\label{Thm:GeneralOperator} 
    Observables are invariant under the gauge transformation on the NH Fermi-Dirac distribution $f_{\rm eff} \rightarrow f_{\rm eff} + \mathbb{R}$. 
\end{theorem}
\begin{proof}[Proof]
    According to Eq.~\eqref{Eqn:CorrelatorGauge} in Lemma~\ref{Thm:Correlators}, the correlator $\braket{c^\dagger_i c_j}$ remains unchanged when $C_1 = C^*_2$. In general, $C_1$ is not required to be equal to $C_2$. However, this results in two NH Fermi-Dirac distributions---$f_{\rm eff} + C_1$ associated with $\psi^\textsc{l*}_{ni}\psi^\textsc{r}_{nj}$ and $[f_{\rm eff} + C_2]^*$ associated with $\psi^\textsc{r*}_{ni}\psi^\textsc{l}_{nj}$---within the same correlator. To this end, we set $C_1 = C_2 = C$ and find $\braket{c^\dagger_i c_j} \rightarrow \braket{c^\dagger_i c_j} + \delta_{ij} \mathrm{Im} (C)$, which directly leads to 
    \begin{align}
    \label{Eqn:TraceOGauge}
    \braket{O} \rightarrow \braket{O} + \mathrm{Tr} (\mathcal{O}) \, \mathrm{Im} (C) \,.
    \end{align}
    Therefore, the observable computed by a general quadratic Hermitian operator in Eq.~\eqref{Eqn:HermitianO} is gauge-invariant when $C \in \mathbb{R}$.
\end{proof}
\noindent Theorem~\ref{Thm:GeneralOperator} directly implies that any real constant $C$ with energy dimension in $f_{\rm eff} (\varepsilon) = - (1/\pi)\ln (\varepsilon/ C)$ will not affect the values of observables. For simplicity, we can set $C = 1$ and place $\varepsilon$ alone inside the logarithm.
\begin{corollary}\label{Thm:TracelessOperator} 
    Observables associated with traceless operators are invariant under gauge transformations of two separate NH Fermi-Dirac distributions $f_{\rm eff} \rightarrow f_{\rm eff} + C_1$ associated with $\psi^\mathrm{L*}_{ni}\psi^\mathrm{R}_{nj}$, and $f^*_{\rm eff} \rightarrow [f_{\rm eff} + C_2]^*$ associated with $\psi^\mathrm{R*}_{ni}\psi^\mathrm{L}_{nj}$, where $C_1$ and $C_2$ are two $\varepsilon_n$-independent complex numbers.
\end{corollary}
\noindent This Corollary can be readily derived from Eq.~\eqref{Eqn:CorrelatorGauge} in Lemma~\ref{Thm:Correlators} and Eq.~\eqref{Eqn:TraceOGauge} in Theorem~\ref{Thm:GeneralOperator}. An paradigmatic instance for a traceless operator is the persistent current operator $\mathcal{J}$ studied in this work. To obtain the analytical expression in Eq.~\eqref{Eqn:CurrentImTrLog} at zero temperature, we transform $f_{\rm eff}$ by assigning specific values to $C_{1}$ and $C_{2}$:
\begin{align}
    I(\phi)
    &= \frac{\mathrm{i}}{2\pi} \sum_n \left( \frac{\mathrm{d} \varepsilon_n}{\mathrm{d} \phi} \ln \varepsilon_n - \frac{\mathrm{d} \varepsilon^*_n}{\mathrm{d} \phi} \ln \varepsilon^*_n  \right)
    = \frac{\mathrm{i}}{2\pi} \sum_n \left[ \frac{\mathrm{d} \varepsilon_n}{\mathrm{d} \phi} (\ln \varepsilon_n + 1) - \frac{\mathrm{d} \varepsilon^*_n}{\mathrm{d} \phi} (\ln \varepsilon^*_n + 1)  \right], \quad \text{with} \quad C_1 = C_2 = -\frac{1}{\pi} \,, \nonumber \\
    &= \frac{\mathrm{i}}{2\pi} \frac{\mathrm{d}}{\mathrm{d} \phi} \sum_n \left( \varepsilon_n \ln \varepsilon_n - \varepsilon^*_n \ln \varepsilon^*_n  \right)
    = - \frac{1}{\pi} \frac{\mathrm{d}}{\mathrm{d} \phi} \mathrm{Im} \sum_n \varepsilon_n \ln \varepsilon_n 
    = - \frac{1}{\pi} \frac{\mathrm{d}}{\mathrm{d} \phi} \mathrm{Im} \mathrm{Tr} \left( \mathcal{H}_{\rm eff} \ln \mathcal{H}_{\rm eff} \right) \,.
\end{align}
One can proceed to transform with $C_1 = 0,\; C_2 = + \mathrm{i}$ and exploit the particle-hole symmetry of superconductors, $\varepsilon_{-n} = - \varepsilon^*_n$, to obtain $\mathrm{i} / \pi \, \partial_\phi \mathrm{Tr} ( \mathcal{H}_{\rm eff} \ln \mathcal{H}_{\rm eff})$ that is specific for supercurrents. In the same vein, the persistent current at $T \neq 0$ is given by
\begin{align}
I(\phi)
    &= \frac{\mathrm{i}}{2\pi} \sum_n \left\{ \frac{\mathrm{d} \varepsilon_n}{\mathrm{d} \phi} \left[ \Psi \left( \frac{1}{2} + \frac{\mathrm{i} \beta \varepsilon_n}{2\pi} \right)-\frac{\mathrm{i} \pi }{2} \right] - \frac{\mathrm{d} \varepsilon^*_n}{\mathrm{d} \phi} \left[ \Psi \left( \frac{1}{2} + \frac{\mathrm{i} \beta \varepsilon_n}{2\pi} \right)-\frac{\mathrm{i} \pi }{2} \right]^*  \right\} \,, \qquad \qquad \; \text{with} \quad C_1 = C_2 = - \frac{\mathrm{i}}{2} \,, \nonumber \\
    &= \frac{\mathrm{i}}{2\pi} \sum_n \left[ \frac{\mathrm{d} \varepsilon_n}{\mathrm{d} \phi} \Psi \left( \frac{1}{2} + \frac{\mathrm{i} \beta \varepsilon_n}{2\pi} \right) - \frac{\mathrm{d} \varepsilon^*_n}{\mathrm{d} \phi} \Psi \left( \frac{1}{2} -\frac{\mathrm{i} \beta \varepsilon^*_n}{2\pi} \right) \right] \nonumber
    = \frac{1}{\beta} \frac{\mathrm{d}}{\mathrm{d} \phi} \sum_n \left[ \mathrm{log\Gamma} \left( \frac{1}{2} + \frac{\mathrm{i} \beta \varepsilon_n}{2\pi} \right) + \mathrm{log\Gamma} \left( \frac{1}{2} + \frac{\mathrm{i} \beta \varepsilon_n}{2\pi} \right)^* \right] \nonumber \\
    &= \frac{2}{\beta} \frac{\mathrm{d}}{\mathrm{d} \phi} \sum_n \mathrm{Re} \,\mathrm{log\Gamma} \left( \frac{1}{2} + \frac{\mathrm{i} \beta \varepsilon_n}{2\pi} \right) 
    = \frac{2}{\beta} \frac{\mathrm{d}}{\mathrm{d} \phi} \mathrm{Re} \,\mathrm{Tr}\, \mathrm{log\Gamma} \left( \frac{1}{2} + \frac{\mathrm{i} \beta }{2\pi} \mathcal{H}_{\rm eff} \right) \,,
\end{align}
where $\mathrm{log\Gamma}$ is the log-gamma function \cite{Weisstein2024LogGamma}, whose derivative is the digamma function $\Psi$.
Similarly, the particle-hole symmetry of superconductors implies $\varepsilon_{-n} = - \varepsilon^*_n$, which in turn can be used to eliminate the $\mathrm{Re}$ operation for supercurrents.

\subsection{Asymptotic behavior}

We denote $\epsilon = \mathrm{Re}\, \varepsilon$ and summarize the asymptotic behavior of $f_{\rm eff}$ in the Hermitian limit where $\mathrm{Im}\, \varepsilon \rightarrow 0$,
\begin{align}
f_{\rm eff} (\varepsilon) =
\begin{cases}
    \displaystyle
    - \frac{1}{\pi} \ln |\epsilon| + \mathrm{i} \Theta(-\epsilon) \,, & T = 0 \,, \\
    \displaystyle
    - \frac{1}{\pi} \mathrm{Re} \left[ \Psi \left(\frac{1}{2} + \frac{\mathrm{i} \beta \epsilon}{2\pi} \right) \right] + \mathrm{i} \frac{1}{\mathrm{e}^{\beta \epsilon} + 1} \,, & T \neq 0 \,, \\
\end{cases}
\; \Longrightarrow \;
\mathrm{Im}\, f_{\rm eff} (\varepsilon) =
\begin{cases}
    \displaystyle
    f_\textsc{fd}(\epsilon) = \Theta(-\epsilon) \,, & T = 0 \,, \\
    \displaystyle
    f_\textsc{fd}(\epsilon, \beta) = \frac{1}{\mathrm{e}^{\beta \epsilon} + 1} \,, & T \neq 0 \,, \\
\end{cases}
\end{align}
which is illustrated in Fig.~\figref{Fig:TemparatureInteraction}{c}. 
To extract the asymptotic behavior of Eq.~\eqref{Eqn:EffectiveDistributionBeta} for the inverse temperature $\beta=1/k_\textsc{b} T$ we write:
\begin{align}
f_{\rm eff}(\varepsilon, \beta) 
=
\begin{cases}
    \displaystyle
    - \frac{1}{\pi} \left[ \cancel{\ln \left(\frac{\beta}{2\pi} \right)} + \ln (\varepsilon) - \frac{\pi^2}{6 \beta^2 \varepsilon^2} \right] 
    \cong - \frac{1}{\pi} \left( \ln \varepsilon - \frac{\pi^2}{6 \beta^2 \varepsilon^2} \right) \,,
    & \beta \rightarrow \infty \,, \\
    \displaystyle
    - \frac{1}{\pi} \left[ \cancel{\Psi \left(\frac{1}{2} \right)} - \mathrm{i} \pi \Big( \frac{1}{2} - \frac{\beta \varepsilon}{4} \Big) \right] 
    \cong \mathrm{i} \Big( \frac{1}{2} - \frac{\beta \varepsilon}{4} \Big) \,,
    & \beta \rightarrow 0 \,, \\
\end{cases}
\end{align}
where we use $\cong$ to represent the gauge equivalence after eliminating real constants (denoted by a slash) according to Theorem~\ref{Thm:GeneralOperator}. The first line goes back to the $T=0$ expression. The imaginary of the second line is consistent with the expansion of $f_\textsc{fd}(\varepsilon, \beta)$ at $\beta \rightarrow 0$. The asymptotic behavior of the persistent current $I(\phi)$ in the Hermitian limit is summarized below: 
\begin{align*}
I(\phi)=
\begin{cases}
    \displaystyle
    -\frac{1}{\pi} \frac{\mathrm{d}}{\mathrm{d} \phi} \mathrm{Im} \sum_n \epsilon_n \Big[ \cancel{\ln |\epsilon_n|} - \mathrm{i} \pi \Theta(-\epsilon_n) \Big] 
    = \underbrace{\frac{\mathrm{d}}{\mathrm{d} \phi} \sum_{n \leqslant 0} \epsilon_n 
    = \frac{\mathrm{d} E_0}{\mathrm{d} \phi}}_{\text{normal metal}} 
    = \underbrace{2 \frac{\mathrm{d}}{\mathrm{d} \phi} \sum_{n \leqslant 0} \frac{\epsilon_n}{2} 
    = 2 \frac{\mathrm{d} E_0}{\mathrm{d} \phi}}_{\text{superconductor}} \,, & T=0 \,, \\
    \displaystyle
    \frac{1}{\beta} \frac{\mathrm{d}}{\mathrm{d} \phi} \sum_n \left[ \cancel{\ln (2\pi)} + \xcancel{\frac{\beta \epsilon_n}{2}} - \ln (1 + \mathrm{e}^{\beta \epsilon_n}) \right]
    = \frac{\mathrm{d}}{\mathrm{d} \phi} \left(\frac{\ln Z }{-\beta}\right) 
    = \frac{\mathrm{d} F}{\mathrm{d} \phi} \,, & \text{normal metal,} \\
    \displaystyle
    \frac{1}{\beta} \frac{\mathrm{d}}{\mathrm{d} \phi} \sum_n \left[ \cancel{\ln (2\pi)} - \ln (\mathrm{e}^{+\beta \epsilon_n/2} + \mathrm{e}^{-\beta \epsilon_n/2}) \right] 
    = 2 \frac{\mathrm{d}}{\mathrm{d} \phi} \sum_{n \geqslant 0} \left[ - \frac{1}{\beta} \ln (\mathrm{e}^{+\beta \epsilon_n/2} + \mathrm{e}^{-\beta \epsilon_n/2}) \right] 
    = 2 \frac{\mathrm{d} F}{\mathrm{d} \phi} \,, & \text{superconductor,}
\end{cases}
\end{align*}
where $Z$ is the partition function, $F$ is the free energy, and we remove the cross terms based on the fact that the trace of the Hamiltonian is independent of $\phi$. To obtain the last two lines at finite temperatures, we should utilize the reflection identity of $\mathrm{log\Gamma}$ \cite{Weisstein2024LogGamma} to expand the $\mathrm{Re}\mathrm{Tr}$ term in Eq.~\eqref{Eqn:CurrentReTrLogGamma}:
\begin{align}
\mathrm{log\Gamma}(1-z)+\mathrm{log\Gamma}(z)=\ln (\pi)-\ln \sin (\pi  z)\,, \quad z = \frac{1}{2} + \frac{\mathrm{i} \beta \epsilon}{2\pi}\,, \quad - \frac{1}{2} < \mathrm{Re} (z) <  \frac{\pi}{2} \,, \quad \epsilon \in \mathbb{R} \,.
\end{align}
Similarly, for the asymptotic behavior of the persistent current $I(\phi,\beta)$ in Eq.~\eqref{Eqn:CurrentReTrLogGamma} with respect to $\beta$, we find:
\begin{align*}
\begin{cases}
    \displaystyle
    2 \frac{\mathrm{d}}{\mathrm{d} \phi} \mathrm{Re} \sum_n \left[ \frac{\mathrm{i} }{2 \pi} \left(\varepsilon_n \ln \varepsilon_n + \xcancel{\varepsilon_n \ln \left(\frac{\mathrm{i} \beta }{2 \pi \mathrm{e}}\right)}\right) + \cancel{\frac{\ln (2 \pi)}{2\beta}} + \frac{\mathrm{i} \pi}{12\beta^2 \varepsilon_n} \right] 
    \cong -\frac{1}{\pi} \frac{\mathrm{d}}{\mathrm{d} \phi} \sum_n \mathrm{Im} \left( \varepsilon_n \ln \varepsilon_n + \frac{\pi^2}{6\beta^2 \varepsilon_n} \right) \,,
    & \beta \rightarrow \infty \,, \\
    \displaystyle
    2 \frac{\mathrm{d}}{\mathrm{d} \phi} \mathrm{Re} \sum_n \left[ \cancel{\frac{\ln \pi}{2\beta}} + \xcancel{\frac{ \mathrm{i} \Psi(1/2)}{2\pi} \varepsilon_n} - \frac{\beta}{16} \varepsilon^2_n \right] 
    = - \frac{\beta}{4} \mathrm{Re} \sum_n \left( \varepsilon_n \frac{\mathrm{d} \varepsilon_n}{\mathrm{d} \phi} \right) \,,
    & \beta \rightarrow 0 \,,
\end{cases}
\end{align*}
where the cross terms are eliminated because the trace of the Hamiltonian is independent of $\phi$. The first line reduces to the zero temperature expression in Eq.~\eqref{Eqn:CurrentImTrLog}, and the second line agrees with the $\beta/4$ scaling as the Hermitian case. The coherent persistent current will decrease to zero due to significant thermal fluctuations as $\beta \rightarrow 0$.

\subsection{Current continuity at EPs}

In the vicinity of a second-order EP at $\varepsilon_\textsc{ep}$, the eigenvalues $\varepsilon_\pm(\phi)$ follow a square-root scaling as a function of $\phi$:
\begin{align}
    \label{Eqn:SquareRootEP}
    \varepsilon_\pm(\phi) = \varepsilon_\textsc{ep} \pm \sqrt{s(\phi)} = \mu - \mathrm{i} \nu(\phi) \pm \sqrt{s(\phi)}, \quad
    \varepsilon'_\pm(\phi) = - \mathrm{i} \nu'(\phi) \pm \frac{s'(\phi)}{2\sqrt{s(\phi)}} \,,
\end{align}
where $\nu(\phi)$ and $s(\phi)$ are bounded functions and smoothly vary as $\phi$ \cite{Heiss2012Physics}. We use an EP-subscript to represent the value of a variable at EP. The derivatives of two EP-modes $\varepsilon'_\pm(\phi_\textsc{ep})$ are divergent at the EP due to the branch point inside the square root $\lim_{\phi \rightarrow \phi_\textsc{ep}} s(\phi) = 0$. In the following discussion on differentiability, we omit the explicit dependence on $\phi$ unless otherwise stated. 
We first analyze the trace term in the persistent current $I(\phi) = -1/\pi\; \partial_\phi \mathrm{Im} \mathrm{Tr} (\mathcal{H}_{\rm eff} \ln \mathcal{H}_{\rm eff})$ at $T = 0$,
\begin{align}
    \mathrm{Tr} (\mathcal{H}_{\rm eff} \ln \mathcal{H}_{\rm eff}) 
    &= \sum_{n = \pm} \varepsilon_n \ln \varepsilon_n + \sum_{n \neq \pm} \varepsilon_n \ln \varepsilon_n 
    = \frac{\varepsilon_{+} - \varepsilon_{-}}{2} (\ln \varepsilon_{+} - \ln \varepsilon_{-}) + \frac{\varepsilon_{+} + \varepsilon_{-}}{2} (\ln \varepsilon_{+} + \ln \varepsilon_{-}) + \sum_{n \neq \pm} \varepsilon_n \ln \varepsilon_n \nonumber \\
    &= \sqrt{s(\phi)} (\ln \varepsilon_{+} - \ln \varepsilon_{-}) + [\mu - \mathrm{i} \nu(\phi)] (\ln \varepsilon_{+} + \ln \varepsilon_{-}) + \sum_{n \neq \pm} \varepsilon_n \ln \varepsilon_n \,.
\end{align}
The last term above is from the bulk states and is differentiable near EPs, thus we can only analyze the EP-modes $\varepsilon_\pm(\phi)$,
\begin{align}
\frac{\mathrm{d}}{\mathrm{d} \phi} \sum_{n = \pm} \varepsilon_n \ln \varepsilon_n
    &= \varepsilon'_{+} (1 + \ln \varepsilon_{+}) + \varepsilon'_{-} (1 + \ln \varepsilon_{-}) 
    = \frac{\varepsilon'_{+} - \varepsilon'_{-}}{2} (\ln \varepsilon_{+} - \ln \varepsilon_{-}) + \frac{\varepsilon'_{+} + \varepsilon'_{-}}{2} (2 + \ln \varepsilon_{+} + \ln \varepsilon_{-}) \\
    &= s'(\phi) \frac{\ln \big(\varepsilon_\textsc{ep} + \sqrt{s(\phi)} \big) - \ln \big(\varepsilon_\textsc{ep} - \sqrt{s(\phi)} \big)}{2\sqrt{s(\phi)}} - \mathrm{i} \nu'(\phi) \left[ 2 + \ln \big(\varepsilon_\textsc{ep} + \sqrt{s(\phi)} \big) + \ln \big(\varepsilon_\textsc{ep} - \sqrt{s(\phi)} \big) \right] \nonumber \,,
\end{align}
where we substitute the square-root expansion in Eq.~\eqref{Eqn:SquareRootEP}. Using $\lim_{\phi \rightarrow \phi_\textsc{ep}} s(\phi) = 0$, we obtain
\begin{align}
    \label{Eqn:LogEPs}
    \lim_{\phi \rightarrow \phi_\textsc{ep}} \frac{\mathrm{d}}{\mathrm{d} \phi} \sum_{n = \pm} \varepsilon_n \ln \varepsilon_n
    &=\frac{s'(\phi_\textsc{ep})}{\varepsilon_\textsc{ep}} - \mathrm{i} \nu'(\phi_\textsc{ep}) ( 2 + 2 \ln \varepsilon_\textsc{ep} ) 
    = -\pi \left[s'(\phi_\textsc{ep}) f'_{\rm eff}(\varepsilon_\textsc{ep})  - 2 \mathrm{i} \nu'(\phi_\textsc{ep}) \left( f_{\rm eff}(\varepsilon_\textsc{ep}) - \frac{1}{\pi} \right) \right] \,.
\end{align}
Following similar steps, we obtain the contribution of two EP-modes in the trace term at $T \neq 0$,
\begin{align}
    \label{Eqn:LogGammaEPs}
    \lim_{\phi \rightarrow \phi_\textsc{ep}} \frac{\mathrm{d}}{\mathrm{d} \phi} \sum_{n = \pm} \mathrm{log\Gamma} \left( \frac{1}{2} + \frac{\mathrm{i} \beta \varepsilon_n}{2\pi} \right)
    &= \frac{\beta}{2\mathrm{i}} \left[s'(\phi_\textsc{ep}) f'_{\rm eff}(\varepsilon_\textsc{ep}, \beta)  - 2 \mathrm{i} \nu'(\phi_\textsc{ep}) \left( f_{\rm eff}(\varepsilon_\textsc{ep}, \beta) - \frac{\mathrm{i}}{2} \right) \right] \,.
\end{align}
Due to the analyticity of $f_{\rm eff} (\varepsilon)$ in the lower half of the complex plane, its derivative $f'_{\rm eff} (\varepsilon_\textsc{ep})$ exists. Hence, both Eqs.~\eqref{Eqn:LogEPs} and \eqref{Eqn:LogGammaEPs} are bounded, ensuring that the persistent current $I(\phi)$ exhibits a continuous and smooth behavior in the vicinity of EPs.

\section{Current Susceptibility}

\subsection{Non-Hermitian approach}

The current susceptibility in superconducting systems can be expanded into a set of four-point electronic correlators by substituting the site-resolved current operator in the Heisenberg picture $J (\tau) = -\mathrm{i} t_j [ c^\dagger_j (\tau) c_{j+1} (\tau) - c^\dagger_{j+1} (\tau) c_j (\tau)]$:
\begin{align}
\label{Eqn:Susceptibility2FourCorrelator}
&\Pi(\tau) 
= - \mathrm{i} \Theta(\tau) \braket{[J(\tau), J(0)]}
= - \mathrm{i} \Theta(\tau) (- t_j^2) 
\braket{ [c^\dagger_j (\tau) c_{j+1} (\tau) - c^\dagger_{j+1} (\tau) c_j (\tau),  c^\dagger_j c_{j+1} - c^\dagger_{j+1} c_j ] } \\
&= \mathrm{i} \Theta(\tau) t_j^2 \left[ \braket{ [c^\dagger_j (\tau) c_{j+1} (\tau), c^\dagger_j c_{j+1}] } - \braket{ [c^\dagger_j (\tau) c_{j+1} (\tau), c^\dagger_{j+1} c_j] } + \braket{ [c^\dagger_{j+1} (\tau) c_j (\tau), c^\dagger_{j+1} c_j] } - \braket{ [c^\dagger_{j+1} (\tau) c_j (\tau), c^\dagger_j c_{j+1}] } \right] \,, \nonumber
\end{align}
where we denote $c_j \equiv c_j(0)$ and omit the subscript in $J \equiv J_j$, since $j$ is always set at the first site of $\mathbb{N}$. Using identities for the propagator $\mathrm{i}\tilde{g}(\tau)$ and the lesser Green's function \cite{Economou2006Green},
\begin{align}
    \{ c_i(\tau), c^\dagger_j \} = \mathrm{i}\tilde{g}_{i,j}(\tau) \,, \quad
    \{ c_i(\tau), c_j \} &= \mathrm{i}\tilde{g}_{i,N+j}(\tau) \,, \quad
    \{ c^\dagger_i(\tau), c^\dagger_j \} = \mathrm{i}\tilde{g}_{N+i,j}(\tau) \,, \\
    \braket{c^\dagger_i c_j(\tau)} = \int \mathrm{d} \omega \mathrm{e}^{-\mathrm{i} \omega \tau} f_\textsc{fd}(\omega) \rho_{j,i}(\omega) \,,
    \braket{c_i c_j(\tau)} &= \int \mathrm{d} \omega \mathrm{e}^{-\mathrm{i} \omega \tau} f_\textsc{fd}(\omega) \rho_{j,N+i}(\omega) \,,
    \braket{c^\dagger_i c^\dagger_j(\tau)} = \int \mathrm{d} \omega \mathrm{e}^{-\mathrm{i} \omega \tau} f_\textsc{fd}(\omega) \rho_{N+j,i}(\omega) \,, \nonumber
\end{align}
we can transform the four-point electronic correlator to an integral involving the density of state operator $\rho_{i,j}(\omega) \equiv \braket{i|\rho(\omega)|j}$,
\begin{align}
\label{Eqn:FourCorrelator2Rho}
    \braket{ [c^\dagger_{i} (\tau) c_{j} (\tau), c^\dagger_k c_{l}] }
    &= \int_{-\infty}^\infty \mathrm{d} \omega_1 \mathrm{d} \omega_2 f_\textsc{fd}(\omega_1) \Big[
        \rho_{j,k}(\omega_2) \rho_{l,i}(\omega_1) \mathrm{e}^{-\mathrm{i} (\omega_1 - \omega_2) (-\tau)} -
        \rho_{l,i}(\omega_2) \rho_{j,k}(\omega_1) \mathrm{e}^{-\mathrm{i} (\omega_1 - \omega_2) (+\tau)} \nonumber \\
    &\qquad \qquad +
        \rho_{N+k,i}(\omega_2) \rho_{j,N+l}(\omega_1) \mathrm{e}^{-\mathrm{i} (\omega_1 - \omega_2) (+\tau)} - 
        \rho_{j,N+l}(\omega_2) \rho_{N+k,i}(\omega_1) \mathrm{e}^{-\mathrm{i} (\omega_1 - \omega_2) (-\tau)} \Big] \,.
\end{align}
Substituting Eq.~\eqref{Eqn:FourCorrelator2Rho} back to each four-point correlator in Eq.~\eqref{Eqn:Susceptibility2FourCorrelator}, and then performing the Fourier transform, we get 
\begin{align}
&\Pi(\omega) = \int \mathrm{d} \tau \Pi(\tau) \mathrm{e}^{+\mathrm{i} \omega \tau}
= (- t_j^2) \int_{-\infty}^\infty \mathrm{d} \omega_1 \mathrm{d} \omega_2 \frac{f_\textsc{fd}(\omega_1) - f_\textsc{fd}(\omega_2)}{\omega + \omega_1 - \omega_2 + \mathrm{i} \eta} \\
&\times \Big\{[\rho_{j+1,j}(\omega_1) \rho_{j+1,j}(\omega_2) - \rho_{j,j}(\omega_1) \rho_{j+1,j+1}(\omega_2) + \rho_{j,j+1}(\omega_1) \rho_{j,j+1}(\omega_2) - \rho_{j+1,j+1}(\omega_1) \rho_{j,j}(\omega_2)] + \nonumber \\
&[\rho_{N+j+1,j}(\omega_1) \rho_{j+1,N+j}(\omega_2) - \rho_{N+j,j}(\omega_1) \rho_{j+1,N+j+1}(\omega_2) + \rho_{N+j,j+1}(\omega_1)  \rho_{j,N+j+1}(\omega_2) - \rho_{N+j+1,j+1}(\omega_1) \rho_{j,N+j}(\omega_2)]
\Big\} \,. \nonumber
\end{align}
Using the identify $\lim _{\eta \rightarrow 0^{+}}[1/(x \pm \mathrm{i} \eta)] = \mathrm{PV} \left(1/x\right) \mp \mathrm{i} \pi \delta(x)$, we obtain its imaginary part
\begin{align}
&\mathrm{Im} \; \Pi(\omega) 
= \pi t_j^2 \int_{-\infty}^\infty \mathrm{d} \omega_1 \mathrm{d} \omega_2 \delta(\omega + \omega_1 - \omega_2) [f_\textsc{fd}(\omega_1) - f_\textsc{fd}(\omega_2)] \\
&\times \mathrm{Re} \Big\{ 
[\rho_{j+1,j}(\omega_1) \rho_{j+1,j}(\omega_2) - \rho_{j,j}(\omega_1) \rho_{j+1,j+1}(\omega_2) + \rho_{j,j+1}(\omega_1) \rho_{j,j+1}(\omega_2) - \rho_{j+1,j+1}(\omega_1)  \rho_{j,j}(\omega_2)] + \nonumber \\
&
[\rho_{N+j+1,j}(\omega_1) \rho_{j+1,N+j}(\omega_2) - \rho_{N+j,j}(\omega_1) \rho_{j+1,N+j+1}(\omega_2) + \rho_{N+j,j+1}(\omega_1) \rho_{j,N+j+1}(\omega_2) - \rho_{N+j+1,j+1}(\omega_1) \rho_{j,N+j}(\omega_2)] \Big\} \,. \nonumber 
\end{align}
Specifically in the zero temperature  $f_\textsc{fd}(\omega) = \Theta(-\omega)$, we can exploit the $\delta$-function to integrate out $\omega_2$ and thus obtain Eq.~\eqref{Eqn:ImaginarySusceptibility} in the main text.
The integral $P_{ijkl} (\omega) \equiv \int_{-\infty}^0 \braket{i|\rho(\omega')|j} \braket{k|\rho(\omega + \omega')|l} \mathrm{d} \omega'$ can be analytically derived as follows:
\begin{align}
\label{Eqn:PintegralFullForm}
&P_{ijkl}(\omega)
= \Big(\frac{\rm i}{2 \pi}\Big)^2 \int_{-\infty}^0 \sum_{nm} 
\left[\frac{\braket{i|\psi^\textsc{r}_n}\braket{\psi^\textsc{l}_n|j}}{\omega' - \varepsilon_n} - \frac{\braket{i|\psi^\textsc{l}_n}\braket{\psi^\textsc{r}_n|j}}{\omega' - \varepsilon_n^*} \right] 
\left[\frac{\braket{k|\psi^\textsc{r}_m}\braket{\psi^\textsc{l}_m|l}}{\omega' + \omega - \varepsilon_m} - \frac{\braket{k|\psi^\textsc{l}_m}\braket{\psi^\textsc{r}_m|l}}{\omega' + \omega - \varepsilon_m^*} \right] \mathrm{d} \omega' \nonumber \\
&= \Big(\frac{\rm i}{2 \pi}\Big)^2 \sum_{nm} \int_{-\infty}^0 
\left[ \frac{\psi^\textsc{r}_{ni}\psi^\textsc{l*}_{nj}}{\omega' - \varepsilon_n} \frac{\psi^\textsc{r}_{mk}\psi^\textsc{l*}_{ml}}{\omega' + \omega - \varepsilon_m} - \frac{\psi^\textsc{r}_{ni}\psi^\textsc{l*}_{nj}}{\omega' - \varepsilon_n} \frac{\psi^\textsc{l}_{mk}\psi^\textsc{r*}_{ml}}{\omega' + \omega - \varepsilon_m^*} + \frac{\psi^\textsc{l}_{ni}\psi^\textsc{r*}_{nj}}{\omega' - \varepsilon_n^*} \frac{\psi^\textsc{l}_{mk}\psi^\textsc{r*}_{ml}}{\omega' + \omega - \varepsilon_m^*} - \frac{\psi^\textsc{l}_{ni}\psi^\textsc{r*}_{nj}}{\omega' - \varepsilon_n^*} \frac{\psi^\textsc{r}_{mk}\psi^\textsc{l*}_{ml}}{\omega' + \omega - \varepsilon_m} \right] \mathrm{d} \omega' \nonumber \\
&= \Big(\frac{\rm i}{2 \pi}\Big)^2 \sum_{nm} 
\Big[+ \psi^\textsc{r}_{ni}\psi^\textsc{l*}_{nj}\psi^\textsc{r}_{mk}\psi^\textsc{l*}_{ml} \frac{\ln (- \varepsilon_n) - \ln(-\varepsilon_m + \omega)}{\varepsilon_n - \varepsilon_m + \omega} - \psi^\textsc{r}_{ni}\psi^\textsc{l*}_{nj}\psi^\textsc{l}_{mk}\psi^\textsc{r*}_{ml} \frac{\ln (+ \varepsilon_n) - \ln(+\varepsilon^*_m - \omega)}{\varepsilon_n - \varepsilon^*_m + \omega} \nonumber \\
&\qquad \qquad \qquad\;\; + \psi^\textsc{l}_{ni}\psi^\textsc{r*}_{nj}\psi^\textsc{l}_{mk}\psi^\textsc{r*}_{ml} \frac{\ln (- \varepsilon^*_n) - \ln(-\varepsilon^*_m + \omega)}{\varepsilon^*_n - \varepsilon^*_m + \omega} - \psi^\textsc{l}_{ni}\psi^\textsc{r*}_{nj} \psi^\textsc{r}_{mk}\psi^\textsc{l*}_{ml} \frac{\ln (+ \varepsilon^*_n) - \ln(+\varepsilon_m - \omega)}{\varepsilon^*_n - \varepsilon_m + \omega} \Big] \,.
\end{align}
One can simplify Eq.~\eqref{Eqn:PintegralFullForm} into Eq.~\eqref{Eqn:Pintegral} by using $f_{\rm eff} (\varepsilon) = - (1/\pi)\ln \varepsilon$ and the definition of $p^{nm\pm}_{ijkl}$ in the main text.

\subsection{Exact diagonalization}

\begin{figure}[b]
    \centering
    \begin{minipage}{0.49\textwidth}
        \centering
        \includegraphics[width=\textwidth]{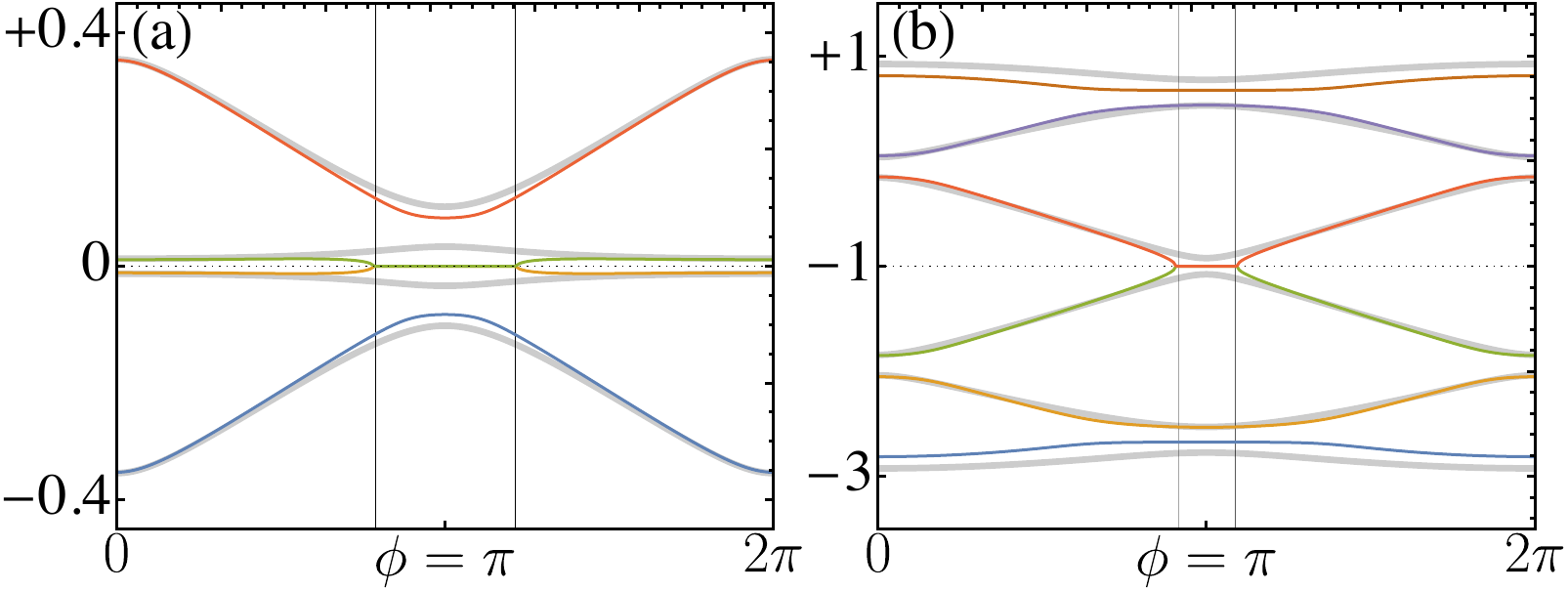}
        \caption{Supplemental plots for Fig.~\ref{Fig:SketchSpectrum}: (a) SNS junction and (b) normal metallic ring. The real parts of the low-energy spectra (colorful lines) vary with $\phi$, compared to the real spectra of the isolated system (gray lines). Their curvatures have been altered by dissipation.}
        \label{Fig:SpectraComparison}
    \end{minipage}\hfill
    \begin{minipage}{0.49\textwidth}
        \centering
        \includegraphics[width=\textwidth]{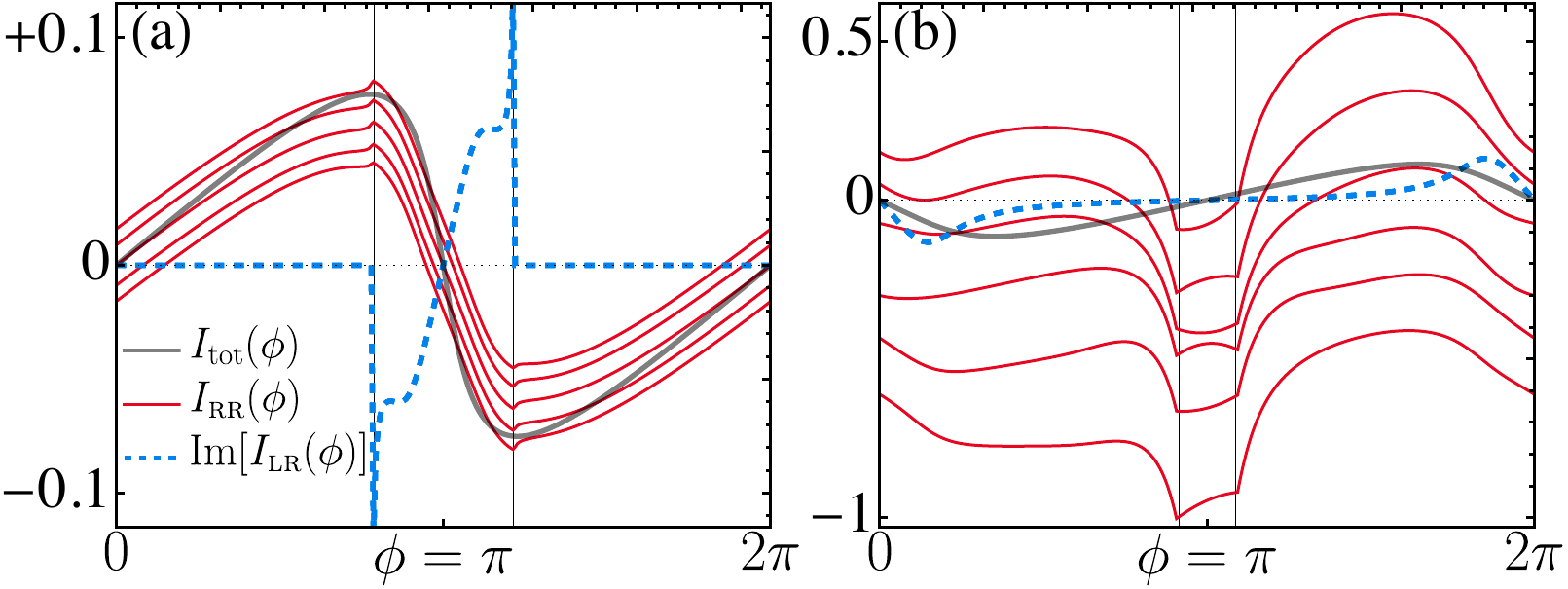}
        \caption{Supplemental plots for Fig.~\ref{Fig:EPsComparison}: (a) SNS junction and (b) normal metallic ring, where $I_{\rm tot}(\phi)$ (gray) acts as a reference. $I_\textsc{rr}(\phi)$ (red) shifts at different $j \in \mathbb{N}$, while $\mathrm{Im}[I_\textsc{lr}(\phi)]$ (dashed) displays the same divergent or finite curve as $\mathrm{Re}[I_\textsc{lr}(\phi)]$.}
        \label{Fig:EPsComparisonSupp}
    \end{minipage}
\end{figure}

\begin{figure}[t]
    \centering
    \begin{minipage}{0.49\textwidth}
        \centering
        \includegraphics[width=\textwidth]{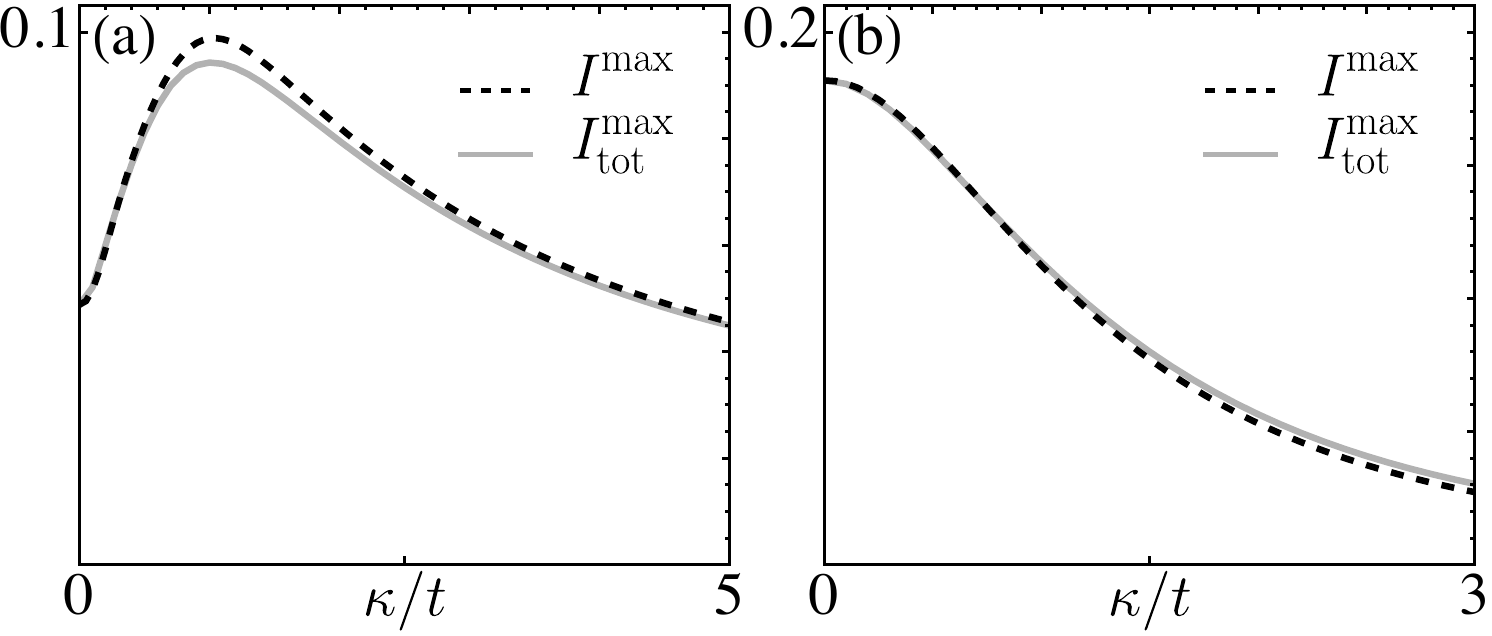}
        \caption{The current amplitudes variation with the coupling strength $\kappa$: (a) increases for small $\kappa$ in SNS junctions, but decreases for larger $\kappa$. (b) monotonically reduces as $\kappa$ increases in normal rings. The current amplitude $I^{\rm max}$ calculated by Eq.~\eqref{Eqn:CurrentImTrLog} closely matches the exact diagonalization $I^{\rm max}_{\rm tot}$ across all $\kappa$ examined.}
        \label{Fig:CurrentAmplitude}
    \end{minipage}\hfill
    \begin{minipage}{0.49\textwidth}
        \centering
        \includegraphics[width=\textwidth]{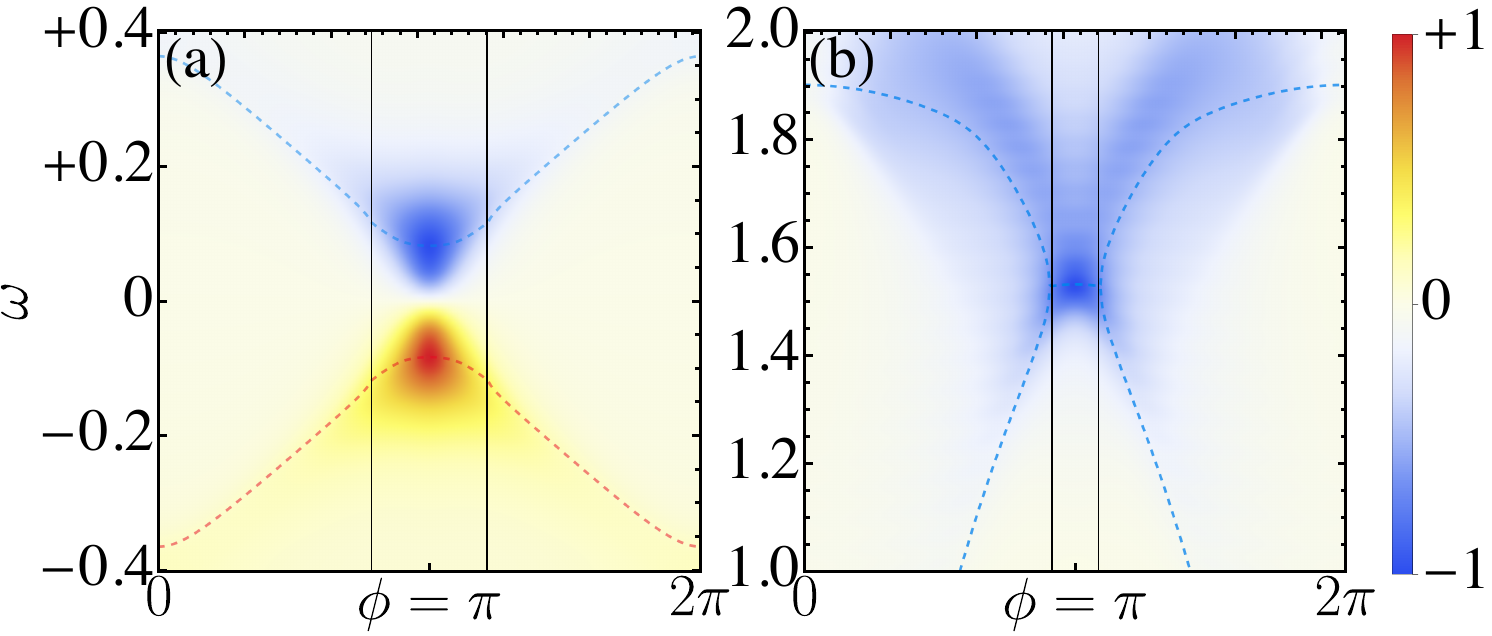}
        \caption{Supplemental plots for Fig.~\ref{Fig:Susceptibility}, where $\mathrm{Im}\,\Pi (\phi, \omega)$ is solved by exact diagonalization of Eq.~\eqref{Eqn:ExactSusceptibility}. We set $\eta = 0.03$ for both (a) the SNS junction and (b) the normal mesoscopic ring. Same as the NH case, the level broadened effect is evident between two EPs (black lines). The other parameters are the same as in Fig.~\ref{Fig:SketchSpectrum}.}
        \label{Fig:SusceptibilitySupp}
    \end{minipage}
\end{figure}

In the Hermitian case, the time evolution of the fermionic operator in superconducting systems can be expressed as $c_{j}(\tau) = \sum_n [u_{n}(j) \mathrm{e}^{-\mathrm{i} \epsilon_n \tau} d_{n} + v^*_{n}(j) \mathrm{e}^{+\mathrm{i} \epsilon_n \tau} d^\dagger_{n}]$, where it involves a set of Bogoliubons $d_{n}$ with energy $\epsilon_n \geqslant 0$ and the wavefunction $\psi_{n} = [\psi_{n}(1), \dots, \psi_n(N), \psi_n(N+1), \dots, \psi_n (2N)]^\textsc{t} = [u_{n}(1), \dots, u_n(N), v_n(1), \dots, v_n (N)]^\textsc{t}$. The current operator reads
\begin{align}
\label{Eqn:ExactCurrent}
J (\tau) &= - t_j \sum_{nm} \Big[ A_{nm} \mathrm{e}^{+\mathrm{i} (\epsilon_n - \epsilon_m) \tau} d^\dagger_{n} d_{m} + B_{nm} \mathrm{e}^{-\mathrm{i} (\epsilon_n + \epsilon_m) \tau} d_{n} d_{m} + B^*_{mn} \mathrm{e}^{+\mathrm{i} (\epsilon_n + \epsilon_m) \tau} d^\dagger_{n} d^\dagger_{m} + C_{nm} \Big] \,, \nonumber \\
A_{nm} &= \mathrm{i} [u^*_{n}(j) u_{m}(j+1) - u^*_{n}(j+1) u_{m}(j)] - \mathrm{i} [v_{m}(j) v^*_{n}(j+1) - v_{m}(j+1) v^*_{n}(j)] \,, \nonumber \\
B_{nm} &= \mathrm{i} [v_{n}(j) u_{m}(j+1) - v_{n}(j+1) u_{m}(j)]\,, \qquad
C_{nm} = \mathrm{i} \delta_{nm} [v_{n}(j) v^*_{n}(j+1) - v_{n}(j+1) v^*_{n}(j)]  \,.
\end{align}
The current susceptibility $\Pi(\tau)$ and its Fourier transform $\Pi(\omega)$ can be calculated using the occupation of Bogoliubons $\braket{d^\dagger_n d_m} = \delta_{nm} f_\textsc{fd}(\epsilon_n)$ (refer to Appendix~D of Ref.~\cite{Shen2023Majoranamagnon} for a pedagogical derivation),
\begin{align}
\label{Eqn:ExactSusceptibility}
\Pi(\tau) &= - \mathrm{i} \Theta(\tau) \braket{[J(\tau), J(0)]}
    = - \mathrm{i} \Theta(\tau) t_j^2 \sum_{nm} \Big\{
    A_{nm} A_{mn} \mathrm{e}^{+\mathrm{i} (\epsilon_n - \epsilon_m) \tau} [ f_\textsc{fd}(\epsilon_n) - f_\textsc{fd}(\epsilon_m) ] \nonumber \\
    &\quad + B_{nm} (B^*_{nm} - B^*_{mn}) \mathrm{e}^{-\mathrm{i} (\epsilon_n + \epsilon_m) \tau} [ f_\textsc{fd}(-\epsilon_n) - f_\textsc{fd}(\epsilon_m)] + B^*_{mn} (B_{mn} - B_{nm}) \mathrm{e}^{+\mathrm{i} (\epsilon_n + \epsilon_m) \tau} [ f_\textsc{fd}(\epsilon_n) - f_\textsc{fd}(-\epsilon_m)] \Big\} \,, \nonumber \\
\Pi(\omega) &= \int \mathrm{d} \tau \Pi(\tau) \mathrm{e}^{+\mathrm{i} \omega \tau}
    = t_j^2 \sum_{nm} \bigg[ A_{nm} A_{mn} \frac{f_\textsc{fd}(\epsilon_n) - f_\textsc{fd}(\epsilon_m)}{\omega + \epsilon_n - \epsilon_m + \mathrm{i} \eta} \nonumber \\
    &\quad + B_{nm} (B^*_{mn} - B^*_{nm}) \frac{f_\textsc{fd}(\epsilon_m) + f_\textsc{fd}(\epsilon_n) - 1}{\omega - \epsilon_n - \epsilon_m + \mathrm{i} \eta} + B^*_{mn} (B_{nm} - B_{mn}) \frac{1 - f_\textsc{fd}(\epsilon_m) - f_\textsc{fd}(\epsilon_n)}{\omega + \epsilon_n + \epsilon_m + \mathrm{i} \eta} \bigg] \,.
\end{align}
We remark that Eqs.~\eqref{Eqn:ExactCurrent} and \eqref{Eqn:ExactSusceptibility} still apply to normal metals by setting all the wavefunctions $v_n = 0$. This enables us to perform an exact diagonalization of the full Hermitian system $H_{\rm tot}$ including a large reservoir, whose results are presented in Fig.~\ref{Fig:SusceptibilitySupp} and show excellent agreement with the effective NH Hamiltonian approach in Fig.~\ref{Fig:Susceptibility}.

\section{Additional numerical verifications}

We provide additional low-energy spectra of NH systems in Fig.~\ref{Fig:SpectraComparison} as a supplement to Fig.~\ref{Fig:SketchSpectrum}, comparing them with the spectra of the decoupled Hermitian system. The slopes of these spectra indicate whether a mode's contribution to the current is positive or negative. In the parameter setting of SNS junctions shown in Fig.~\figref{Fig:SpectraComparison}{a}, the EP-modes $\varepsilon_\pm$ do not dominate the current; instead, it is primarily contributed by the adjacent mode below. Because of their opposite slopes, in the Hermitian case, the lower-energy mode $\varepsilon_-$ has a negative contribution to the current, which is effectively counterbalanced by $\varepsilon_+$ due to level broadening in the NH case. In the normal ring illustrated in Fig.~\figref{Fig:SpectraComparison}{b}, however, we observe that any positive contributions from single-particle modes are consistently moderated by the dissipation-induced level broadening.

Fig.~\ref{Fig:EPsComparisonSupp} is presented to complement the analysis of persistent currents discussed in Fig.~\ref{Fig:EPsComparison}. The RR-basis current $I_\textsc{rr}(\phi)$ exhibits peculiar shifts at different sites within the system, highlighting its violation of local conservation laws. Additionally, the imaginary part of the LR-basis current $\mathrm{Im}[I_\textsc{lr}(\phi)]$ is included to illustrate its similarity with its real part $\mathrm{Re}[I_\textsc{lr}(\phi)]$. 

We further investigate how the current amplitudes $I^{\rm max}$ vary with coupling strength $\kappa/t$ across different regimes in Fig.~\ref{Fig:CurrentAmplitude}. Panel (a) highlights the dual role of the thermal reservoir in the SNS junctions, where it enhances the current in the weak coupling regime but acts to suppress it under strong coupling conditions. Panel (b) shows a monotonic decrease in current amplitudes with increasing coupling strength in normal rings, aligning with the notion that dissipation suppresses the current amplitudes. In both models, the current amplitudes computed by Eq.~\eqref{Eqn:CurrentImTrLog} closely match those obtained through exact diagonalization across all examined coupling regimes, confirming its non-perturbative nature.

Fig.~\ref{Fig:SusceptibilitySupp} presents supplemental data for Fig.~\ref{Fig:Susceptibility}. The current susceptibility calculated by Eq.~\eqref{Eqn:ImaginarySusceptibility} is consistent with Eq.~\eqref{Eqn:ExactSusceptibility} through exact diagonalization of the entire Hermitian system, including a large reservoir. This confirms the signatures of EPs and validates the effectiveness of our analytical approach in capturing the dynamics of NH systems.

\begin{figure*}
    \centering
    \includegraphics[width=\textwidth]{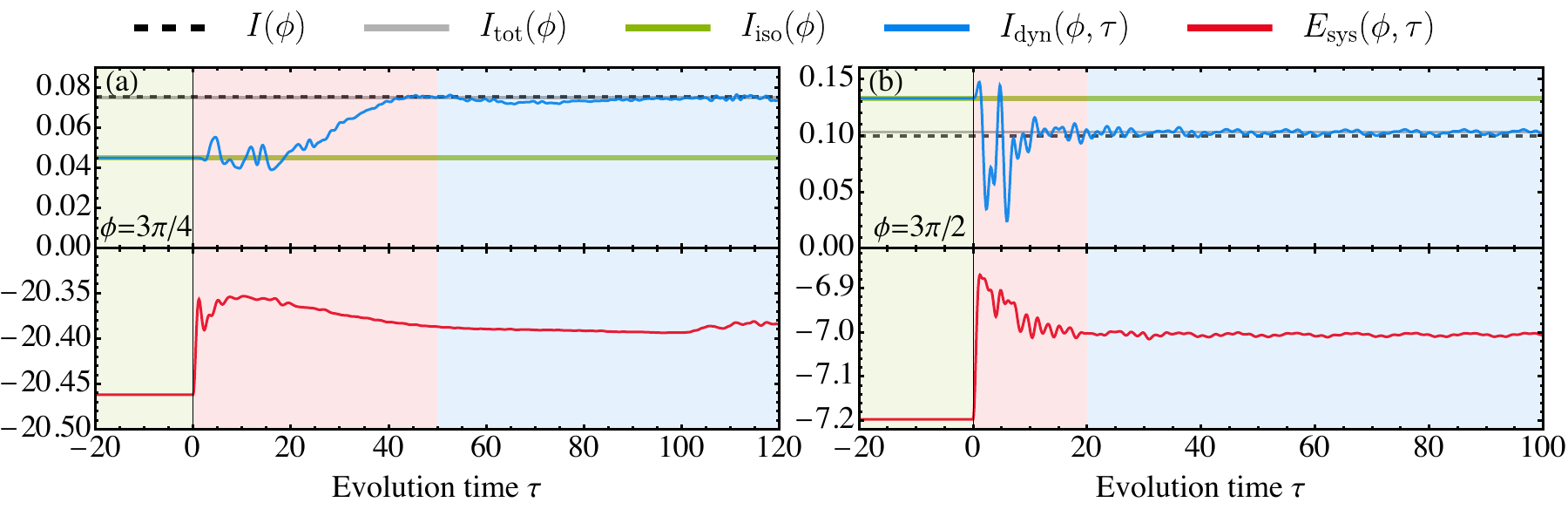}	
    \caption{Quench dynamics of (a) an SNS junction and (b) a normal ring coupled to thermal reservoirs. The upper panels show the evolution of persistent currents $I_{\rm dyn} (\phi, \tau)$ with a fixed $\phi$, and the lower panels present the corresponding evolution of the system energy $E_{\rm sys} (\phi, \tau)$. The green background ($\tau<0$) represents the decoupled regime where both systems are initially prepared and remain in their many-body ground states. At $\tau=0$, the systems are attached to the thermal reservoirs, entering a transient regime (red background) where they exchange energy with the reservoirs. In the long-time limit (blue background), both systems reach equilibrium, where $E_{\rm sys} (\phi, \tau)$ fluctuates around a fixed value and $I_{\rm dyn} (\phi, \tau)$ approaches $I(\phi)$ in Eq.~\eqref{Eqn:CurrentImTrLog} with small fluctuations. We set $N_\textsc{e} = 101$ for the SNS junction and $N_\textsc{e} = 201$ for the normal ring while other parameters are the same as in Fig.~\ref{Fig:SketchSpectrum}. The numerical values of $I(\phi)$, $I_{\rm tot}(\phi)$, and $I_{\rm iso}(\phi)$ are consistent with Fig.~\ref{Fig:EPsComparison}.}
    \label{Fig:CurrentDynamics}
\end{figure*}

Finally, we present numerical dynamics in Fig.~\ref{Fig:CurrentDynamics} to explicitly show that our formalism can accurately capture quantum many-body observables as the system reaches equilibrium in the long-time limit. Specifically, we follow the methodology from Eqs.~(3.1)-(3.5) in Ref.~\cite{Loss1991Dephasing}, adapting it to our systems: the SNS junction and the normal ring. For simplicity, we consider zero temperature and prepare the system and the reservoir in their many-body ground states, respectively, with their single-particle modes below the Fermi level fully occupied. For $\tau < 0$, the total system is governed by the decoupled Hamiltonian $H_{\rm dec} = H_{\rm sys} + H_{\rm res}$. Hence, the system energy $E_{\rm sys} (\phi, \tau)$ remains constant and the persistent current $I_{\rm dyn} (\phi, \tau) = I_{\rm iso} (\phi)$. However, at $\tau = 0$, we quench the system by coupling it to the reservoir, resulting in the total Hamiltonian $H_{\rm tot} = H_{\rm sys} + H_{\rm res} + H_{\rm tun}$. Consequently, the dynamic evolution of the total system for $\tau>0$ is governed by the unitary $U_{\rm tot}(\tau) = \exp(- \mathrm{i} H_{\rm tot} \tau)$ with $\hbar=1$. The persistent current $I_{\rm dyn} (\phi, \tau)$ can be calculated by the expectation value of the site-resolved current operator $J_j$.

As shown in Fig.~\ref{Fig:CurrentDynamics}, the dynamics can be categorized into three regimes: ($i$) the decoupled regime (green background) where both the system and the reservoir remain in their many-body ground states; ($ii$) the transient regime (red background) characterized by energy exchange with the reservoir; and ($iii$) the equilibrium regime (blue background) where the system energy $E_{\rm sys}(\phi, \tau)$ fluctuates, but its mean value averaged over the time window almost remains constant. The persistent current $I_{\rm dyn}(\phi, \tau)$ approaches the theoretical value $I(\phi)$ computed by our Eq.~\eqref{Eqn:CurrentImTrLog} with small fluctuations that encode the dephasing caused by the reservoir. This dynamic simulation confirms the validity of our formalism for capturing quantum many-body observables in the long-time limit where the system reaches equilibrium.

\end{document}